\documentclass{article}

\usepackage[T1]{fontenc}
\usepackage[latin9, utf8]{inputenc}
\usepackage[a4paper]{geometry}
\usepackage{url}
\usepackage[dvipsnames]{xcolor}

\usepackage{mathrsfs, amsthm, amsmath, amssymb, mathtools, bbm, physics} %Math Packages
\usepackage{float}

\usepackage{graphicx} % Required for inserting images
\usepackage[thicklines]{cancel}
\usepackage{tikz-cd}
\usetikzlibrary{decorations.pathmorphing}
\usepackage{setspace}
\usepackage{slashed}

\numberwithin{equation}{section}
\numberwithin{figure}{section}
\theoremstyle{plain}
\newtheorem{thm}{\protect\theoremname}[section]
\theoremstyle{plain}
\newtheorem{prop}[thm]{\protect\propositionname}
\newtheorem*{star-prop}{Proposition}

\makeatother
  \providecommand{\propositionname}{Proposition}

\usepackage[unicode=true,pdfusetitle, bookmarks=true,bookmarksnumbered=false,bookmarksopen=false, breaklinks=true,pdfborder={0 0 0},backref=false,colorlinks=true]{hyperref}
\usepackage{color, soul}
\usepackage{url}
\usepackage{fancyhdr}

\definecolor{halfgray}{gray}{0.55}
\definecolor{webgreen}{rgb}{0,.5,0}
\definecolor{webbrown}{rgb}{.6,0,0}
\definecolor{Maroon}{cmyk}{0, 0.87, 0.68, 0.32}
\definecolor{RoyalBlue}{cmyk}{1, 0.50, 0, 0}
\definecolor{Black}{cmyk}{0, 0, 0, 0}

\hypersetup{
    urlcolor=Lavender, linkcolor=RoyalBlue, citecolor=Green, colorlinks=true, linkcolor=blue
}

\usepackage{authblk}

\newcommand{\maps}{\rightarrow}

\def\1{I}

\def\AA{\mathbf{A}}
\def\BB{\mathbf{B}}

\def\JJ{\mathbf{J}}
\def\LL{\mathbf{L}}

\def\qq{\mathbf{q}}
\def\hh{\mathbf{h}}

\def\EEL{\mathbf{EL}}
\def\TTh{\mathbf{\Theta}}
\def\tth{\boldsymbol{\theta}}
\def\jj{\mathbf{j}}
\def\vv{\mathbbm{v}}
\def\OOmega{\boldsymbol{\Omega}}
\def\oomega{\boldsymbol{\omega}}
\def\ssigma{\boldsymbol{\sigma}}
\def\sssigma{\boldsymbol{\varsigma}}
\def\eell{\boldsymbol{\ell}}
\def\ttau{\mathfrak{t}}
\def\nn{\mathfrak{n}}
\def\uu{\mathfrak{u}}
\def\ss{\mathfrak{s}}
\def\aa{\mathfrak{a}}

\def\B{\mathcal{B}}
\def\C{\mathcal{C}}

\def\E{\mathcal{E}}
\def\F{\mathcal{F}}

\def\H{\mathcal{H}}

\def\L{\mathcal{L}}
\def\M{\mathcal{M}}
\def\N{\mathcal{N}}

\def\P{\mathcal{P}}
\def\Q{\mathcal{Q}}
\def\R{\mathcal{R}}
\def\S{\mathcal{S}}
\def\T{\mathcal{T}}

\def\V{\mathcal{V}}

\def\dv{\mathbbm{d}}
\def\dve{\widetilde{\mathbbm{d}}}

\def\dvbar{\slashed{\dv}}

\definecolor{mygreen}{RGB}{26, 122, 51}

\makeindex
\begin{document}
%--------------------------------------------------------------------

\title{A Covariant Phase Space Approach to Einstein-\AE{}ther Gravity}

\author[1,2,3]{Walter Arata\thanks{\href{mailto:warata@sissa.it}{warata@sissa.it}}}
\author[1,2,3]{Stefano Liberati\thanks{\href{mailto:liberati@sissa.it}{liberati@sissa.it}}}
\author[1,2,3]{Giulio Neri\thanks{\href{mailto:gneri@sissa.it}{gneri@sissa.it}}}

\affil[1]{SISSA--International School for Advanced Studies, via Bonomea 265, 34136 Trieste, Italy}
\affil[2]{INFN, Sezione di Trieste, via Valerio 2, 34127 Trieste, Italy}
\affil[3]{IFPU--Institute for Fundamental Physics of the Universe, via Beirut 2, 34014 Trieste, Italy}

\maketitle

%Begin Frontmatter ------------------------------------------------

\abstract{
Black hole thermodynamics in Lorentz-violating gravity is subtle because different excitations propagate at different speeds and hence identify different causal horizons. We revisit Einstein--\AE{ther} gravity using the covariant phase space formalism with boundaries and derive a consistent first law for stationary black holes. For a mode of propagation speed $c_s$, we introduce a disformal frame in which the corresponding causal horizon is a Killing horizon, so that the standard Wald-type derivation can be carried out. The result is then mapped back to the original frame, where it mantains the same structure. The associated horizon charge contains, besides the usual Komar term, an irreducible entropic \AE{ther} contribution that can be interpreted as heat due to the \AE{ther} flux across the horizon; accordingly, the total entropy splits into a gravitational part and an \AE{ther} part. We further develop an extended-thermodynamics framework in which the couplings of the theory are allowed to vary, obtaining generalized Smarr relations. Finally, we analyze the probe-mode limit $c_s \to +\infty$, clarifying its connection to universal-horizon thermodynamics and resolving the apparent tension in the literature between approaches that (i) fix the entropy to be proportional to the area and infer a corresponding temperature, and (ii) impose the Hawking temperature associated with modes peeling from the universal horizon and infer the entropy. Once the independent \AE{ther} contribution is properly taken into account, the two prescriptions are reconciled.}

\clearpage

\tableofcontents{}

%End Frontmatter ----------------------------------------------------

\newpage

%Begin Mainmatter ---------------------------------------------------
\section{Introduction}

As of today, black hole thermodynamics remains one of the sharpest probes of the interplay between gravity, quantum physics and statistical mechanics. In General Relativity, the laws of black hole mechanics can be established as classical theorems, yet their physical consistency ultimately relies on quantum effects such as Hawking radiation; conversely, gravitational dynamics itself can be recast in thermodynamic language. It is therefore natural to ask how robust this correspondence is: does it persist in generic diffeomorphism-invariant theories, or does it rely on additional structural properties that are special to General Relativity? With this in mind, Lorentz-violating theories represent an especially informative case study, as they modify the causal structure of spacetime in a controlled but conceptually drastic way. 

Among Lorentz-violating theories of gravity, Einstein--\AE{ther} theory~\cite{Jacobson:2000xp} and Ho\v{r}ava--Lifshitz gravity~\cite{Horava:2009uw} play a central role. Einstein--\AE{ther} theory is generally-covariant and couples the metric to a dynamical time-like unit vector field, the \textit{\AE{ther}}, which breaks local Lorentz invariance by selecting a preferred local time direction. Ho\v{r}ava--Lifshitz gravity, on the other hand, has a manifestly Lorentz-breaking formulation where space and time are treated on different grounds. The latter, in particular, is given by the value of a scalar field called the \textit{Khronon}.

A distinctive consequence is that the causal structure is no longer governed by a single metric light cone. Even in vacuum, linearized perturbations split into multiple propagating modes: the usual massless spin-2 excitation associated with the metric, plus a spin-1 and a spin-0 excitation associated with the \AE{ther},\footnote{The spin-1 excitation is absent in Ho\v{r}ava--Lifshitz gravity.} each characterized by its own speed of propagation relative to the \AE{ther} frame. Therefore, black hole spacetimes generically exhibit multiple distinct causal horizons, and the relevant notion of “no escape” depends on which excitation is used to probe the geometry.

This multi-horizon structure immediately complicates the thermodynamic interpretation, possibly leading to inconsistencies without resorting to the ultraviolet completion of the theory~\cite{Dubovsky:2006vk,Eling:2007qd,Jacobson:2010fat}. For example, it becomes unclear which horizon should set the temperature, which area should be associated with an entropy (assuming a Bekenstein-like formula holds), and the peeling of which modes should provide the physical notion of surface gravity. The standard Wald construction cannot be applied naively in these cases as the metric Killing horizon is not, in general, the horizon relevant for all dynamical degrees of freedom. 
In addition, Lorentz-violating theories admit an extended notion of causal boundary, the \emph{universal horizon}~\cite{Blas:2011ni}, which can trap arbitrarily fast signals in spacetimes foliated by constant-Khronon hypersurfaces. While universal horizons are geometrically well defined, their thermodynamic role is not transparent in the vacuum sector of the theory where all modes propagate at finite speed (and hence cannot connect the universal horizon to infinity), but it becomes compelling once modified dispersion relations are allowed and superluminal, non-relativistic propagation is possible (see~\cite{DelPorro:2023lbv} and references therein).  

These issues have been confronted before using covariant phase space methods in Einstein--\AE{ther} gravity, but early analyses encountered obstructions, especially when attempting to move beyond the simplest asymptotics. In particular, while consistent first laws were obtained in asymptotically flat settings with universal horizons, extending the construction to more general asymptotics (notably AdS) proved problematic. 
Motivated by these open problems, we employ the more recent formalism of covariant phase spaces in spacetimes \emph{with boundaries} to derive the first law of black hole thermodynamics in a framework that remains well-defined for generic asymptotics. 

Our strategy in this work is to isolate the thermodynamics associated with a specific --- yet arbitrary --- mode of propagating speed $c_s$ (henceforth, the $s$-mode). Using a well-known trick in the Einstein--\AE{ther} literature, we perform a disformal transformation on the metric to make the causal horizon of the $s$-mode coincide with the Killing horizon of the disformal metric. In this frame, the surface gravity computed from the Killing generator matches the peeling surface gravity of that mode by construction (see~\cite{Cropp:2013zxi} for definitions and details), and the usual construction à-la Wald can be carried out in a controlled setting. Once we obtain the first law associated with the $s$-mode in the disformal frame, we map it back to the original frame and show that it retains the same formal structure. In particular, a new term appears that is absent in General Relativity, reflecting the \AE{ther} flow contribution to the Killing current across the horizon. This flux can be interpreted as heat and therefore, via the Clausius relation $\dvbar Q_{\textnormal{\AE}} = T \, \dv S_{\textnormal{\AE}}$, is associated with an entropy term, where the temperature is determined by the peeling rate of the $s$-mode at the horizon.

The second goal of this work is to clarify the status of universal horizon thermodynamics and the origin of apparently conflicting prescriptions in the literature. Two complementary viewpoints have been advocated: one fixes an area law for the universal horizon entropy and infers the associated temperature~\cite{Bhattacharyya:2014kta}, while the other imposes the universal horizon temperature --- suggested by ray-tracing arguments for infinitely fast excitations~\cite{DelPorro:2023lbv} --- and infers the form of the entropy~\cite{Pacilio:2017swi}. In the first case one finds a universal temperature which is not apparently related to any form of Hawking radiation; in the second case one obtains an entropy which does not always follow an area law. 

We connect these two frameworks by studying the formal infinite speed probe-mode limit, $c_s \to +\infty$. In this limit, the peeling surface gravity of the $s$-mode continuously approaches the surface gravity of the universal horizon. This procedure allows us to recover a consistent first law --- and thereby a Smarr formula --- at the universal horizon. Within the covariant phase space formalism, we show that the key ingredient missing in previous derivations is exactly the \AE{ther} Killing flux contribution. In situations where the problem is characterized by a single scale (such as the Schwarzschild geometry), it was possible to conflate this \AE{ther} term with the usual black hole entropy contribution by a simple rescaling of the Hawking temperature or of the $S/A$ proportionality factor. In the presence of additional scales (such as a cosmological constant), this flux term becomes genuinely independent and cannot in general be reabsorbed.

In the following subsection we shall fix our notation and conventions. In Section~\ref{sec: CPS Formalism} and~\ref{sec: First Law BH}, we recall the essential elements of the covariant phase space formalism, with particular emphasis on boundary terms, and the general derivation of the first law of black hole thermodynamics within the framework.
The reader who is already familiar with the formalism can skip these review sections.
In Section~\ref{sec: EA Theory} we first review Einstein--\AE{ther} theory and summarize its relation to Khronometric theory. Then, in Section~\ref{sec: CPS for EA}, we compute the theory-specific ingredients that are needed for its covariant phase space analysis, including its symplectic form and Noether charges. In Section~\ref{sec: Examples} we illustrate the general construction on four explicit solutions, including examples previously discussed in~\cite{Bhattacharyya:2014kta}. Finally, we shall draw our conclusions in Section~\ref{sec:conclus}.

\subsection{Notation and Conventions}\label{sec: conventions}

Let us start by introducing some notation and conventions we shall adopt in this work.
In most of our computations we make use of the notation associated with the variational bicomplex framework, introduced by Anderson in~\cite{And92}. 
This is a convenient and systematic framework to study field theories, symmetries and conservation laws. Any field theory is defined on a suitable (either vector or principal) bundle $F$ over a manifold $\M$ (the spacetime manifold, for us). Fields are sections of this bundle $F$ and the space of all such sections constitutes the configuration space. The key idea of the variational bicomplex is to work not only with the fields, but also with all of their derivatives, treated as independent variables. This naturally leads to the notion of the \emph{jet bundle} of the bundle $F$, an enlarged space in which every possible derivative of the fields appears as a coordinate. On this space, one can define a differential--geometric structure that organizes the calculus of variations.

The bicomplex is built from differential forms on the jet bundle, equipped with two commuting exterior derivatives:
\begin{itemize}
  \item the \emph{horizontal exterior derivative} $d$, which corresponds to the exterior derivative acting along the base manifold $\M$ of the jet bundle,
  \item the \emph{vertical exterior derivative} $\dv$, which corresponds to variations of the fields and acts along the fibers,
\end{itemize}
Both exterior derivatives are nilpotent, i.e.~$d^2 = 0$ and $\dv^2 = 0$.
It is worth mentioning that the horizontal and vertical derivatives here are taken to commute $\dv \circ d = d \circ \dv$, in contrast to the mathematicians' convention where they anti-commute $\dv \circ d = - d \circ \dv$.

Since the full exterior derivative on the jet bundle can be decomposed into horizontal and vertical exterior derivative,~\footnote{More explicitly, $d_{\mathrm{full}} = d + (-1)^{r}\dv$, where $r$ here is the horizontal degree of the differential form this operator is acting on.} the same is true for any differential form on the jet bundle.\footnote{The name variational \emph{bicomplex} comes exactly from this: the de Rham complex on the jet bundle can be separated into two orthogonal complexes, one associated with the horizontal exterior derivative $d$ and one to the vertical exterior derivative $\dv$.} Hence we will say that a differential form has a \emph{bi-degree} $(p, q)$ where $p$ is the usual form degree on the exterior cotangent bundle $\bigwedge T^\ast\M$, and $q$ is the vertical degree, morally how many variations in field space we have taken to construct that form. Vertical forms are sections of the exterior cotangent bundle of the configuration space.

On the configuration space, we also define the tangent bundle and denote its vector fields by $\vv$ (we will use subscripts to emphasize the type of flow they induce). We will denote differential forms (both on field space and on the manifold $\M$) with boldface symbols, while the dual vector fields (when they can be defined) will have the same symbol but in normal text. The interior product between a vector $V$ and a $p$-form $\boldsymbol{\eta}$ in $\M$ is denoted by $\iota_V \boldsymbol{\eta}$ and is defined as the $(p - 1)$-form $(\iota_V \boldsymbol{\eta})(X_1, \dots X_{p-1}) = \boldsymbol{\eta}(V,X_1,\dots,X_{p - 1})$. Similarly, we denote the interior product of a vertical vector $\vv$ with a vertical $q$-form $\boldsymbol{\alpha}$ by $\mathbb{I}_\vv \boldsymbol{\alpha}$, which is defined as the $(q - 1)$-form such that $(\mathbb{I}_\vv \boldsymbol{\alpha})(\vv_1, \dots \vv_{q - 1}) = \boldsymbol{\alpha}(\vv,\vv_1,\dots,\vv_{q - 1})$. As for standard functions, which cannot be contracted with a vector ($\iota_V f = 0, \forall f \in C^\infty(\M)$), the contraction of a vertical vector $\vv$ with a $(p, 0)$-form gives $0$, as a $(p, 0)$-form does not have arguments to insert $\vv$ in.
With this in mind, we hence denote by $\mathbb{L}_\vv$ the Lie derivative in configuration space along the vector field $\vv$, while the usual Lie derivative on the manifold along the vector field $\xi$ will be denoted as $\pounds_\xi$. Cartan's Magic Formula holds for both definitions of Lie derivatives, meaning $\mathbb{L}_\vv \boldsymbol{\eta} = \mathbb{I}_\vv (\dv \boldsymbol{\eta}) + \dv (\mathbb{I}_\vv \boldsymbol{\eta})$ and $\pounds_\xi \boldsymbol{\eta} = \iota_\xi (d \boldsymbol{\eta}) + d (\iota_\xi \boldsymbol{\eta})$ for any $(p, q)$-form $\boldsymbol{\eta}$ on the variational bicomplex.

In this variational bicomplex notation, a Lagrangian density $\mathcal{L} \, \boldsymbol{\epsilon}_\M$ is a horizontal $n$-form $\LL$. Its vertical exterior derivative $\dv \LL$ produces the Euler--Lagrange $(n, 1)$-form, whose vanishing encodes the equations of motion. We denote equalities that hold \textit{on-shell} of these equations with the symbol $\doteq$. More precisely, this means that the expression on the left-hand side of $\doteq$ is to be understood as pulled-back along the embedding of the solution space $\mathrm{Sol}$ into the configuration space $\C$, $i_{\mathrm{Sol}} : \mathrm{Sol} \hookrightarrow \C$. Both spaces will be explicitly defined in Subsection~\ref{subsec: field spaces}.

Conservation laws are identified with elements of the horizontal cohomology of the bicomplex: for instance, a conserved vector current $J^\mu$ corresponds to a horizontal $(n - 1)$-form $\JJ$ such that $d \JJ \doteq 0$.

Throughout this paper, tensors with symmetrized indices $T^{(\mu_1 \dots \mu_N)}$ (resp. antisymmetrized $T^{[\mu_1 \dots \mu_N]}$) are defined with the $1/N!$ factor.
Finally, greek indices always run from $0$ to $n$, while latin indices run from $1$ to $n$.

\section{Covariant Phase Space Formalism}\label{sec: CPS Formalism}

In this section, we review the covariant phase space formalism, developed by several authors between the 1980s and 1990s. This approach provides a coordinate-independent framework to study symmetries, conserved quantities, and canonical transformations. In particular, the covariant phase space formalism allows for a very natural description of diffeomorphism charges, making it especially well suited for the study of black hole thermodynamics in generally covariant theories such as Einstein--\AE{ther} gravity.

In reviewing the formalism, we will focus on gravitational field theories. Also, we will point out the key ingredients that will allow us to derive a first law for black hole solutions as prescribed in~\cite{Wald:1993nt}, albeit the possibilities of this formalism are more far-reaching --- see~\cite{Harlow:2019yfa, Cattaneo:2023wxd, Margalef-Bentabol:2020teu} for a broader and more detailed discussion of the topic.

As previously mentioned, we take spacetime to be a $n$-dimensional, simply connected and oriented manifold $\M$ endowed with a metric $g$. We also assume that this manifold can be foliated by suitably chosen hypersurfaces $\Sigma$ such that, topologically, $\M \cong [t_1, t_2] \times \Sigma$. As illustrated in Fig.~\ref{fig: spacetime}, the boundary of the spacetime is $\partial \M = \Sigma_1 \cup \Sigma_2 \cup \Gamma$, where $\Sigma_i$ are the top and bottom lids at time $t_i$ and $\Gamma = \bigcup\limits_{t \in [t_1,t_2]} \partial \Sigma_t$ is the spatial boundary. 

The intersection between the spatial boundary $\Gamma$ and the two lids form two non-smooth corners $\partial \Sigma_i := \Gamma \cap \Sigma_i$. For any piece of the boundary of $\M$ and the corners, we have an embedding $i_\N: \N \maps \M$ for $\N \in {\Gamma, \Sigma_i, \partial\Sigma_i}$.
Since the spacetime is oriented, we choose the volume element of the manifold as $\boldsymbol{\epsilon}_\M := \sqrt{|\det g|} \, dx^0 \wedge \dots \wedge dx^{n - 1}$.

If we call $\ttau$ the time-like, normalized, future-pointing 1-form that is everywhere orthogonal to the slices $\Sigma$ and $\nn$ the space-like, normalized, radially pointing 1-form that is orthogonal to $\Gamma$, we can also define the volume element induced on each part of the boundary as
\begin{equation}
    \begin{cases}
        \boldsymbol{\epsilon}_\M \eval_\Sigma =: \ttau \wedge \boldsymbol{\epsilon}_\Sigma \quad \Rightarrow \quad \boldsymbol{\epsilon}_\Sigma \eval_\Gamma = \nn \wedge \boldsymbol{\epsilon}_{\partial\Sigma} \,;\\
        \boldsymbol{\epsilon}_\M \eval_\Gamma =: \nn \wedge \boldsymbol{\epsilon}_\Gamma \quad \Rightarrow \quad \boldsymbol{\epsilon}_\Gamma \eval_\Sigma = -\ttau \wedge \boldsymbol{\epsilon}_{\partial\Sigma} \,,
    \end{cases}
\end{equation}
where in polar coordinates $\boldsymbol{\epsilon}_{\partial\Sigma}$ will always be $\boldsymbol{\epsilon}_{\partial\Sigma} = r^2 \, \sin^2\theta \, d \theta \wedge d \varphi$. All these elements are depicted in Fig.~\ref{fig: spacetime}.
\begin{figure}[H]
    \centering
    \begin{tikzpicture}[scale=0.6]

        %Top lid
        \draw (0,0) arc(0:180:4cm and 0.5cm);
        \draw (-8,0) arc(180:360:4cm and 0.5cm);

        %Gamma
        \draw (0,0) .. controls (-1,-6) .. (0,-12);
        \draw (-8,0) .. controls (-7,-6) .. (-8,-12);

        %Bottom lid
         \draw[dashed] (0,-12) arc(0:180:4cm and 0.5cm);
         \draw (-8,-12) arc(180:360:4cm and 0.5cm);

         \draw[dashed] (-0.75,-6) arc(0:180:3.25cm and 0.5cm);
         \draw (-7.25,-6) arc(180:360:3.25cm and 0.5cm);
        
        \draw (0.5, -2) node {{$\Gamma$}};
        \draw (-4, -12) node {{$\Sigma_1$}};
        \draw (-4, 0) node {{$\Sigma_2$}};
        \draw (-4, -6) node {{$\Sigma$}};
        \draw (0.7, 0) node {{$\partial\Sigma_2$}};
        \draw (0.7,-12) node {{$\partial\Sigma_1$}};
        
        \draw[->] (-8.5,-12) -- (-8.5,0);
        \draw (-9, 0) node {{$t$}};

        \draw[->] (-4,0.3) -- (-4,0.8) node [anchor = east] {{$\ttau$}};
        \draw[->] (-4,-11.7) -- (-4,-11.2) node [anchor = east] {{$\ttau$}};
        \draw[->] (-4,-5.7) -- (-4,-5.2) node [anchor = east] {{$\ttau$}};
        \draw[->] (-0.75,-6) -- (-0.25,-6) node [anchor = west] {{$\nn$}};

    \end{tikzpicture}
    \caption{Representation of the cylindrical spacetime under consideration, with its different parts: the lateral boundary $\Gamma$, the two lids $\Sigma_i$, the two corners $\partial\Sigma_i$ and a generic slice $\Sigma$.}
    \label{fig: spacetime}
\end{figure}
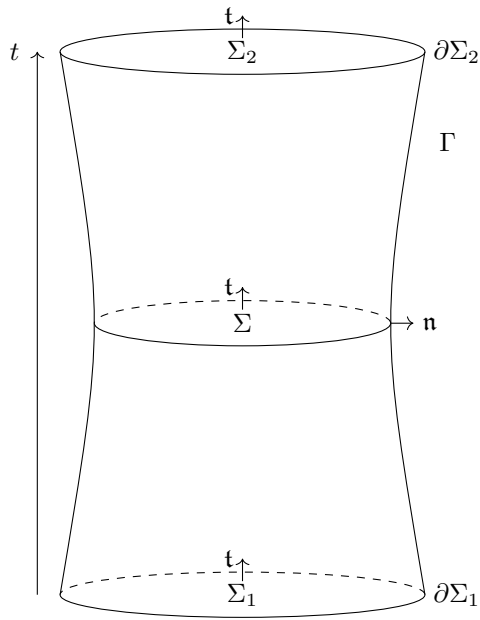

\subsection{Field Space}\label{subsec: field spaces}

Let us now consider a field theory on top of this spacetime, whose field content will be collectively called $\phi$. Suppose also that this system is described by an action functional
\begin{equation}
    \S[\phi] = \int_\M \LL + \int_\Gamma \eell \,,
\end{equation}
where $\LL$ is the \emph{bulk Lagrangian} and $\eell$ is the \emph{boundary Lagrangian}. The boundary Lagrangian $\eell$ has support only on $\Gamma$, and it will be useful to define proper boundary conditions there. Instead, the two lids are viewed as hypersurfaces where to define the initial and final state taken into account and they are \emph{not} used to choose the theory itself. Indeed, different boundary conditions defined on $\Gamma$ lead to different theories, while different conditions on the two lids just correspond to different states one can consider.

In order to construct a proper phase space~\footnote{A \emph{phase space} is a manifold $\F$ endowed with a \emph{symplectic form} $\OOmega$, the latter being a \textit{closed} (i.e.~the exterior derivative of $\OOmega$ is $0$), \textit{non-degenerate} (i.e.~$\OOmega(\vv_1, \vv_2) = 0 \, \forall \, \vv_2 \in T\F \Rightarrow \vv_1 = 0$) 2-form on $\F$.} we can start by defining the \emph{configuration space}, meaning the space of all the possible field configurations, $\C := \{\mathrm{all \ possible \ field \ configurations \ \phi} \}$. As mentioned before, this  configuration space contains all sections of some fiber bundle $F$ over $\M$.

Since the manifold has a boundary, we might want to restrict our configuration space to only those configurations that satisfy some boundary conditions $\BB$ at $\Gamma$. For example, we might impose Dirichlet ($\dv \phi |_\Gamma = 0$), Neumann ($n^\mu \, \nabla_\mu \dv \phi |_\Gamma = 0$), or mixed/Robin boundary conditions. Let us emphasize that the specific choice of boundary conditions is guided by physical considerations, but it is fundamental to guarantee the well-posedness of the initial value problem.

We thus define the \emph{boundary compatible configuration space} as the subset of the configuration space for which the chosen boundary conditions are satisfied, $\B := \{\phi \in \C \ | \ \BB = 0 \}$. Similarly, we define the \emph{solution space} $\mathrm{Sol} := \{ \phi \in \C \ | \ \EEL[\phi] = 0 \}$ as the set of fields configurations that satisfy the Euler-Lagrange equations obtained by extremizing the action functional $\S[\phi]$.
The intersection of these spaces will give us the space of physically relevant field solutions which we take to be our \emph{pre-phase space} $\widetilde{\P} := \mathrm{Sol} \cap \B = \{\phi \in \C \ | \ \EEL = 0 \ \mathrm{and} \ \BB = 0 \}$.
This manifold can be endowed with a closed 2-form $\OOmega$ (see next section for details), called the \emph{pre-symplectic form} starting from the action. The prefix ``pre'' is due to the fact that $\OOmega$ is still generically degenerate, hence not fully symplectic. In this work, we omit the issue of degeneracy and refer to these objects just as phase space and symplectic form, writing $\P$ and $\boldsymbol{\Omega}$.

\subsection{Dynamics}\label{dynamics}

From the variational principle, we know that the field configurations that satisfy the Euler--Lagrange equations are related to the ones for which the action functional is stationary, up to terms that are integrated over the two lids $\Sigma_i$. Taking the variation of the action, we get
\begin{equation}\label{eq: variation of S}
    \dv \S[\phi] = \int_\M \dv \LL + \int_\Gamma \dv \eell \,.
\end{equation}
The variation of the bulk Lagrangian $\LL$ is related to the Euler--Lagrange equations as~\footnote{This result can be derived by performing integrations by part after taking the variation of the Lagrangian; see for example~\cite{Wald:1990mme} for a proof using indices and~\cite{bridges2010multisymplectic} for a proof using the variational bicomplex formalism.}
\begin{equation}\label{eq: variation of L_0}
    \dv \LL = \EEL \, \dv \phi + d \TTh \,,
\end{equation}
where the term $\TTh$ in Eq.~\eqref{eq: variation of L_0} is the so called \emph{symplectic boundary potential}. In the language of the variational bicomplex, $\TTh$ is a $(n-1, 1)$-form, which depends linearly on the field variations $\dv \phi$. This will be the starting point to construct the symplectic form.
Plugging Eq.~\eqref{eq: variation of L_0} into Eq.~\eqref{eq: variation of S}, we have
\begin{equation}\label{eq: action varaition onshell}
    \dv \S[\phi] = \int_\M \EEL \, \dv \phi + \int_{\Sigma_1}^{\Sigma_2} \TTh + \int_\Gamma \Bigl( \TTh + \dv \eell \Bigr) \,,
\end{equation}
where the symbol $\int_{\Sigma_1}^{\Sigma_2}$ means $\int_{\Sigma_2} - \int_{\Sigma_1}$. 

The middle term in the right hand side.~of the above expression, located on the two lids, is of the kind that is allowed by the variational principle. On the other hand, the right-most one is not allowed as it is. To make it compatible with the variational principle, we can restrict the form of the boundary Lagrangian $\eell$ in such a way that its variation cancels the (pull-backed) symplectic potential on $\Gamma$, up to terms that are localized on the only part of $\Gamma$ that also belongs to the lids, meaning the corners. This can be done by requiring
\begin{equation}\label{eq: dL_1}
    i_\Gamma^\ast \TTh + \dv \eell =: \BB + d \tth \,,
\end{equation}
where $\tth$ is called the \emph{symplectic corner potential}. Indeed, upon using the \emph{boundary conditions} $\BB$, the only term remaining in Eq.~\eqref{eq: dL_1} is a total derivative of $\tth$ which, after integration, will contribute only on the corners. Plugging Eq.~\eqref{eq: dL_1} in Eq.~\eqref{eq: action varaition onshell} gives
\begin{equation}
    \dv \S[\phi] = \int_\M \EEL \, \dv \phi + \int_\Gamma \BB + \int_{\Sigma_1}^{\Sigma_2} \Bigl( \TTh - d \tth \Bigr) \,,
\end{equation}
which is compatible with the variational principle once fields are restricted to phase space $\P$ (recall that this means that we go on-shell and impose the boundary conditions). Notice that, in the last term, we changed the sign of the derivative of $\tth$ because the corner orientation induced from $\Gamma$ is the opposite to the one induced from $\Sigma$.

We can now define the \emph{symplectic current} $\oomega$ as the restriction to the phase space of the (vertical) variation of the last term in the previous equation
\begin{equation}
    \oomega :\doteq \dv \Bigl( i_\Sigma^\ast \TTh - d \tth \Bigr) \eval_{\P} \,,
\end{equation}
and the associated \emph{symplectic form} on $\P$ as
\begin{equation}\label{eq: Omega}
    \OOmega_\Sigma := \int_\Sigma \oomega \,.
\end{equation}
The fact that the symplectic form is closed is crucial in what comes next.

If there exists a vector field $\vv$ under whose flow the 2-form $\OOmega_\Sigma$ is left invariant, i.e.~$\mathbb{L}_\vv \OOmega_\Sigma = 0$, where $\mathbb{L}_\vv$ is the Lie derivative along $\vv$, this implies --- thanks to Cartan's Magic Formula --- that the contraction of $\vv$ with the symplectic form $\mathbb{I}_\vv \OOmega$ is closed and hence locally exact:
\begin{equation}\label{eq: vOmega = -dH}
    0 = \mathbb{L}_\vv \OOmega_\Sigma = \dv \bigl( \mathbb{I}_\vv \OOmega_\Sigma \bigr) + \mathbb{I}_\vv \dv \OOmega_\Sigma \quad \implies \quad \mathbb{I}_\vv \OOmega_\Sigma = - \dv \Q_\Sigma[\vv] \,,
\end{equation}
for some locally defined $\Q_\Sigma[\vv]$.\footnote{The minus sign is conventional.}
We will be interested in vector fields $\vv$ for which $\Q_\Sigma[\vv]$ can be defined globally in phase space. When this function exists, we call it the canonical generator of the transformation induced by $\vv$. For example, when the vector field is the one that generates evolution in time $\vv_t$, $\Q[\vv_t]$ will be the Hamiltonian of the system.

\subsection{Symmetries and Noether Charges}

A powerful tool to compute the canonical generator $\Q_\Sigma[\vv]$ for symmetry vector fields is Noether's theorem. This holds for \emph{variational symmetries}, i.e.~field space vector fields such that $\mathbb{L}_\vv \S = 0$. Since we derived the symplectic form from the action, these transformations also preserve the symplectic form and thereby satisfy Eq.~\eqref{eq: vOmega = -dH}. In this work, we will focus on these variational symmetries.

For these kind of transformations, Noether's theorem states that there exists a \emph{Noether current}
\begin{equation}\label{eq: J_1}
    \JJ[\vv] := \mathbb{I}_\vv \TTh - \ssigma[\vv] \,,
\end{equation}
where $\ssigma[\vv]$ is such that
\begin{equation}\label{eq: sigma1}
    d \ssigma[\vv] = \mathbb{I}_\vv \dv \LL \,.
\end{equation}
This current is then conserved on-shell.

As shown in~\cite{Wald:1990mme}, if one has a closed $(p, q)$-form (with $p < n$) that is locally constructed from fields, then it is exact, meaning there exists a $(p-1, q)$-form --- again locally constructed from fields --- such that the former is the vertical derivative of the latter. This is the case for the Noether current $\JJ[\vv]$, that is closed on-shell:
\begin{equation}\label{eq: Noecharge}
    d \JJ[\vv] = d \mathbb{I}_\vv \TTh - d \ssigma[\vv] = \mathbb{I}_\vv \dv \LL - \EEL \, \mathbb{I}_\vv \dv \phi - \mathbb{I}_\vv \dv \LL \doteq 0 \quad \Rightarrow \quad \JJ[\vv] \doteq: d \qq[\vv] \,,
\end{equation}
with $\qq[\vv]$ a $(n - 2)$-form that we call the \emph{Noether charge}.

The argument of~\cite{Wald:1990mme} to infer the existence of $\qq[\vv]$ is very general. There is no universal algorithm to extract Noether charges, but, in concrete cases, we usually work with the dual vector $J^\mu[\vv]$ and try to express it as the divergence of an antisymmetric, rank $(n - 2)$ tensor $q^{\mu\nu}[\vv]$: the form dual of this tensor will be a candidate Noether charge.

The conserved current $\JJ[\vv]$ is the key ingredient to construct the canonical generator $\Q_\vv$ by integrating this over a slice $\Sigma$, as its variation is related to the contraction of $\vv$ with $\TTh$. Moreover, this integral will be independent of any bulk variation of the slice because $\JJ[\vv]$ is closed.
However, we are also interested in variations of $\Sigma$ that affect its boundary. To obtain a canonical generator that is invariant under these boundary variations, we need to add a term that compensate the flux of the current across the boundary. The term we have to add is the \emph{boundary Noether current} $\jj[\vv]$, constructed with terms that are intrinsically defined on the boundary
\begin{equation}\label{eq: J_2}
    \jj[\vv] := \mathbb{I}_\vv \tth - \sssigma[\vv] \,,
\end{equation}
where $\sssigma[\vv]$ is defined by the condition 
\begin{equation}\label{eq: sigma2}
    \mathbb{I}_\vv \dv \eell =: d \sssigma[\vv] - i^\ast_\Gamma \ssigma[\vv] \,.
\end{equation}
This condition holds only for vector fields $\vv$ that annihilate $\BB$ or when restricted to $\P$ where $\BB = 0$.

Therefore, the \emph{canonical generator} $\Q_\Sigma[\vv]$ associated with the variational symmetry vector field $\vv$ is
\begin{equation}\label{eq: Hamiltonian}
    \Q_\Sigma[\vv] :\doteq \int_\Sigma \JJ[\vv] - \int_{\partial\Sigma} \jj[\vv] \overset{\eqref{eq: Noecharge}}{=} \int_{\partial\Sigma} \Bigl( \qq[\vv] - \jj[\vv] \Bigr) \,.
\end{equation}
The following proposition then holds true.
\begin{prop}\label{prop:sliceindep}
    The canonical generator defined in Eq.~\eqref{eq: Hamiltonian} does not depend on the slice of integration $\Sigma$. Moreover the vertical variation of this quantity is related to the symplectic form contraction as in Eq.~\eqref{eq: vOmega = -dH}.
\end{prop}
We shall give a proof for this statement in Appendix~\ref{app: proofprop} and, as a consequence of it, we shall drop the subscript $\Sigma$ in $\Q_\Sigma[\vv]$ from now on.

\subsection{Diffeomorphism Invariance and Covariant Lagrangians}

When working with gravitation theories, one generally requires invariance under the full group of spacetime diffeomorphisms. Diffeomorphisms also induce transformations on fields.

For the diffeomorphism generated by $\xi$ to be a symmetry of our phase space, we ask that both the Euler-Lagrange equations and the boundary conditions are left invariant, that is (we use $\delta$ to refer to a generic variation)
\begin{equation}\label{eq: diffeo constraints}
    \begin{cases}
        \delta_\xi \EEL = 0 \,;\\
        \delta_\xi \BB = 0 \,.
    \end{cases}
\end{equation}
A sufficient condition for the first equality to be satisfied is that $\vv_\xi$ is a variational symmetry (meaning it keeps the action invariant). On the other hand, the presence of a boundary already breaks the full group of diffeomorphisms of $\M$, which means that the action is at most invariant under those diffeomorphisms that preserve the position of the boundary: this is equivalent to ask $\iota_\xi \nn = 0$ everywhere on $\Gamma$. The second equality in Eq.~\eqref{eq: diffeo constraints}, of course, depends on the boundary conditions chosen.

For those theories whose Lagrangians are covariant~\footnote{A \textit{tensorial field} $T$ on configuration space transforms as $\delta_\xi T = \mathbb{L}_{\vv_\xi} T$, that is with a configuration space Lie derivative. If $T$ is constructed as a local combination of fields, it is also a tensor in spacetime, which means that it changes under the diffeomorphism generated by $\xi$ as $\delta_\xi T = \pounds_\xi T$, where we recall that $\pounds$ is the spacetime Lie derivative. Hence, we say that an object $T$ is \emph{covariant} under $\xi$ if it satisfies
\begin{equation}\label{eq: covariance}
    \mathbb{L}_{\vv_\xi} T =  \pounds_\xi T \,.
\end{equation}
}
and for diffeomorphisms that satisfy Eq.~\eqref{eq: diffeo constraints}, we can explicitly compute the terms $\ssigma[\vv_\xi]$ and $\sssigma[\vv_\xi]$ in the bulk and boundary Noether currents Eq.~\eqref{eq: J_1}, Eq.~\eqref{eq: J_2}. Indeed, using Cartan's Magic formula, we can show that
\begin{equation}
    \delta_\xi \LL = \mathbb{L}_{\vv_\xi} \LL = \iota_{\vv_\xi} \dv \LL + \dv (\iota_{\vv_\xi} \LL) \,,
\end{equation}
where the second term vanishes because $\LL$ is a $(n, 0)$-form. Moreover, because of covariance
\begin{equation}
    \delta_\xi \LL = \pounds_\xi \LL = \iota_\xi d \LL + d \bigl( \iota_\xi \LL \bigr) \,,
\end{equation}
where the last term vanishes because $\LL$ is a top-form on $\M$. Comparing the two, we can deduce that $\ssigma[\vv_\xi] = \iota_\xi \LL$. Analogous considerations for $\eell$ lead us to the result
\begin{equation}\label{eq: sigmas diffeo}
    \ssigma[\vv_\xi] = \iota_\xi \LL \,, \quad \mathrm{and} \quad  \sssigma[\vv_\xi] = \iota_{\xi} \eell \,,
\end{equation}

So, using Eq.~\eqref{eq: J_2} together with Eq.~\eqref{eq: sigmas diffeo}, the charge associated to a diffeomorphism that satisfies the condition Eq.~\eqref{eq: diffeo constraints} has the final form
\begin{equation}\label{eq: diffeo charge}
    \boxed{\Q[\vv_\xi] \doteq \int_{\partial\Sigma} \Bigl( \qq[\vv_\xi] - \jj[\vv_\xi] \Bigr) = \int_{\partial\Sigma} \Bigl( \qq[\vv_\xi] - \iota_{\vv_\xi} \tth + \iota_{\xi} \eell \Bigr)} \,.
\end{equation}

\section{First Law of Black Hole Thermodynamics and Smarr Formula}\label{sec: First Law BH}

Consider now a stationary black hole metric, meaning a solution $\phi^\circ$ of the Euler--Lagrange equations with a Killing vector field $\xi$ (i.e.~$\pounds_\xi \phi^\circ = 0$) that has a causal boundary $H$. In General Relativity, $H$ is the Killing horizon of $\xi$, but in different theories $H$ is not necessarily a Killing horizon, nor a null surface in general.

We take the lateral boundary of the manifold $\M$ as $\Gamma \cup H$.
We also define $\S_\R := \Gamma \cap \Sigma$, and $\H := H \cap \Sigma$,  meaning the intersection of $\Gamma$ and $H$ with a generic slice $\Sigma$, to distinguish them from the two corners $\partial \Sigma_i$. See Fig.\ref{fig: spacetime with internal boundary}. Here $\R$ stands for a regulator scale that one will eventually send to infinity.
\begin{figure}[H]
    \centering
    \begin{tikzpicture}[scale=0.6]

        %Top lid
        \draw (0,0) arc(0:180:4cm and 0.5cm);
        \draw (-8,0) arc(180:360:4cm and 0.5cm);

        %Top horizon
        \draw (-3,0) arc(0:180:1cm and 0.1cm);
        \draw (-5,0) arc(180:360:1cm and 0.1cm);

        %Gamma
        \draw (0,0) .. controls (-1,-6) .. (0,-12);
        \draw (-8,0) .. controls (-7,-6) .. (-8,-12);

        %Bottom lid
         \draw[dashed] (0,-12) arc(0:180:4cm and 0.5cm);
         \draw (-8,-12) arc(180:360:4cm and 0.5cm);

         %Bottom horizon
        \draw[dashed] (-3,-12) arc(0:180:1cm and 0.1cm);
        \draw[dashed] (-5,-12) arc(180:360:1cm and 0.1cm);

        %Generic slice
         \draw[dashed] (-0.75,-6) arc(0:180:3.25cm and 0.5cm);
         \draw (-7.25,-6) arc(180:360:3.25cm and 0.5cm);

         %Generic horizon
        \draw[dashed] (-3,-6) arc(0:180:1cm and 0.1cm);
        \draw[dashed] (-5,-6) arc(180:360:1cm and 0.1cm);

        %Internal boundary
        \draw[dashed] (-3,0) -- (-3,-12);
        \draw[dashed] (-5,0) -- (-5,-12);
        
        \draw (0.5, -2) node {{$\Gamma$}};
        \draw (-2, -12) node {{$\Sigma_1$}};
        \draw (-2, 0) node {{$\Sigma_2$}};
        \draw (-2, -6) node {{$\Sigma$}};
        \draw (0.7, 0) node {{$\partial \Sigma_2$}};
        \draw (0.7,-12) node {{$\partial \Sigma_1$}};
        \draw (-5.6,-8) node {{$H$}};
        \draw (-5.6,-6) node {{$\H$}};
        \draw (-7.8,-6) node {{$\S_\R$}};
        
        \draw[->] (-8.5,-12) -- (-8.5,0);
        \draw (-9, 0) node {{$t$}};

        \draw[->] (-4,0.3) -- (-4,0.8) node [anchor = east] {{$\ttau$}};
        \draw[->] (-4,-11.7) -- (-4,-11.2) node [anchor = east] {{$\ttau$}};
        \draw[->] (-6.5,-5.9) -- (-6.5,-5.4) node [anchor = east] {{$\ttau$}};
        \draw[->] (-0.75,-6) -- (-0.25,-6) node [anchor = west] {{$\nn$}};

    \end{tikzpicture}
    \caption{Spacetime in presence of an internal boundary $H$}
    \label{fig: spacetime with internal boundary}
\end{figure}
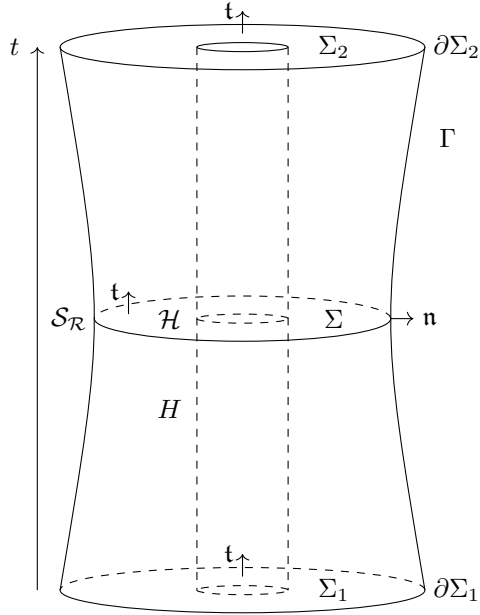

In this case, the variation of the canonical generator Eq.~\eqref{eq: diffeo charge} will contain two pieces, one integrated over the external boundary $\S_\R$ and one over the internal boundary $\H$:
\begin{equation}\label{eq: first law with omega}
    \dv \Q[\vv_\xi] \doteq \dv \int_{\S_\R} \Bigl( \qq[\vv_\xi] - \jj[\vv_\xi] \Bigr) - \dv \int_\H \qq[\vv_\xi] + \int_\H \iota_\xi \TTh \,.
\end{equation}

For simple choices of $\xi$, the integral over the corner $\S_\R$ have a clear physical interpretation: 
\begin{itemize}
    \item If $\xi = \partial_t$, then
    \begin{equation}\label{eq: dM}
        \dv M := \dv \int_{\S_\R} \Bigl( \qq[\vv_t] - \jj[\vv_t] \Bigr) \,,
    \end{equation}
    is the variation of the system energy. 

    \item If $\xi = \partial_\varphi$, then
    \begin{equation}\label{eq: dJ}
        \dv J := - \dv \int_{\mathcal S_\R} \Bigl( \qq[\vv_\varphi] - \jj[\vv_\varphi] \Bigr) \,,
    \end{equation}
    is the variation of the system angular momentum.
\end{itemize}
In the limit where the boundary is at spatial infinity, the two quantities above respectively correspond to the ADM mass and angular momentum. 

On the other hand, the quantity evaluated on the horizon $\H$ in Eq.~\eqref{eq: first law with omega} will be related to the product of the \emph{temperature} $T$ and the \emph{entropy} $S$ of the black hole, as we will soon show.

According to Proposition~\ref{prop:sliceindep}, the total charge Eq.~\eqref{eq: first law with omega} is related to the symplectic form contraction as
\begin{equation}\label{eq: dQ = -ivomega}
    \dv \Q [\vv_\xi] = - \iota_{\vv_\xi} \int_\Sigma \boldsymbol{\omega} \,.
\end{equation}
Since both the boundary and corner symplectic potentials are linear in the field variations, so it is the symplectic current $\boldsymbol{\omega}$. Given that the vector field $\xi$ is Killing, we have $\iota_{\vv_\xi} \dv \phi = 0$ (we assume that fields are covariant), hence the contraction $\iota_{\vv_\xi} \boldsymbol{\omega}$ on the right hand side of Eq.~\eqref{eq: dQ = -ivomega} vanishes and we obtain a relation between the quantities defined on the external boundary and those defined on $\H$, namely
\begin{equation}\label{eq:1stlawnoomega}
    \dv \int_{\S_\R} \Bigl( \qq[\vv_\xi] - \jj[\vv_\xi] \Bigr) \doteq \dv \int_\H \qq[\vv_\xi] - \int_\H \iota_\xi \TTh \,.
\end{equation}
Let us show how this relation can be interpreted as the first law of black hole thermodynamics.

As Wald proved in his seminal paper~\cite{Wald:1993nt}, in order to interpret the first term on the right hand side of Eq.~\eqref{eq:1stlawnoomega} as an entropy contribution, we have to restrict our focus on stationary variations around the black hole geometry that \textit{do not move the horizon}. To do so, we have to require that variations preserve the rescaled vector field $\hat{\xi} := \xi/\kappa_\textsc{h}$, so defined to have unit surface gravity.\footnote{If we define the surface gravity of the horizon to be either $\kappa$-generator, $\nabla_\mu \xi_\nu = \kappa_\textsc{h} \, \epsilon_{\mu\nu}$, or $\kappa$-inaffinity, $\xi^\mu \, \nabla_\mu \xi^\nu = \kappa_\textsc{h} \, \xi^\nu$, we can see that substituting $\hat{\xi}$ we get $\nabla_\mu \hat{\xi}_\nu = \epsilon_{\mu\nu}$ and $\hat{\xi}^\mu \, \nabla_\mu \hat{\xi}^\nu = \hat{\xi}^\nu$.} This implies that the diffeomorphism must be \textit{field-dependent}  
\begin{equation}\label{eq: variazione di xi}
    \dv \hat{\xi}^\mu = \dv \left( \dfrac{\xi^\mu}{\kappa_\textsc{h}} \right) = \dfrac{\dv \xi^\mu}{\kappa_\textsc{h}} - \dv\kappa_\textsc{h} \, \dfrac{\xi^\mu}{\kappa_\textsc{h}^2} \overset{\mathrm{impose}}{=} 0 \quad \implies \quad \dv \xi^\mu = \dfrac{\dv\kappa_\textsc{h}}{\kappa_\textsc{h}} \, \xi^\mu \neq 0 \,.
\end{equation}
As a consequence of this field-dependence, Eq.~\eqref{eq: dQ = -ivomega} no longer holds as written, but an additional term appears, namely
\begin{equation}\label{eq: first law fdvf Q}
    \dv \Q[\vv_\xi] - \Q[\vv_{\dv\xi}] \doteq - \int_\Sigma \iota_{\vv_\xi} \oomega \,.
\end{equation}
This new term is essential for the thermodynamical interpretation.

Expanding Eq.~\eqref{eq: first law fdvf Q} we get a contribution at the asymptotic boundary and one on the horizon. We can choose the variation of the Killing vector field such that it vanishes on $\Gamma$ --- so that the asymptotic quantities can be defined as before (see Eqs.~\eqref{eq: dM} and ~\eqref{eq: dJ}) --- while it behaves like Eq.~\eqref{eq: variazione di xi} on the horizon.

By linearity, the charge density $\qq[\vv_\xi]$ on the horizon $\H$ can be written as the product of $\kappa_\textsc{h}$ and the charge density computed with $\hat{\xi}$ (namely $\qq[\vv_{\hat{\xi}}]$). Therefore its variation would naively include contributions from both factors for the Leibniz rule. However, the subtraction in Eq.~\eqref{eq: first law fdvf Q} \emph{cancels exactly} the term associated with $\dv \kappa_\textsc{h}$:
\begin{equation}\label{eq.dq-qd}
    \dv \qq[\vv_\xi] - \qq[\vv_{\dv\xi}] = \dv \Bigl( \kappa_\textsc{h} \, \qq[\vv_{\hat{\xi}}] \Bigr) - \dfrac{\dv \kappa_\textsc{h}}{\kappa_\textsc{h}} \, \qq[\vv_\xi] = \kappa_\textsc{h} \, \dv \qq[\vv_{\hat\xi}] \,,
\end{equation}
If we write (the integral of) $\qq[\vv_\xi]$ as the product of temperature and entropy $T \, S$, we can make the analogy
\begin{gather}
    \dv (\kappa_\textsc{h} \, \qq[\vv_{\hat{\xi}}]) - \dfrac{\dv \kappa_\textsc{h}}{\kappa_\textsc{h}} \, \qq[\vv_\xi] = \kappa_\textsc{h} \, \dv\qq[\vv_{\hat\xi}] \,, \\
    \dv (T \, S) - \dv T \, S = T \, \dv S \,,
\end{gather}
leading to the following identifications
\begin{equation}\label{eq: entropy and temperature}
    S_\textsc{bh} := 2\pi \int_\H \qq[\vv_{\hat{\xi}}] \,, \quad \mathrm{and} \quad T_\textsc{bh} := \dfrac{\kappa_\textsc{h}}{2\pi} \,.
\end{equation}
These are precisely the definitions given in~\cite{Wald:1993nt}.\footnote{The $2\pi$ factor is chosen in such a way that the entropy reproduces the Bekenstein--Hawking one $S_\textsc{bh} = A_\textsc{bh}/4 \, G_\mathrm{N}$ in General Relativity. The temperature also reproduces the celebrated Hawking temperature of quantum field theories in black hole spacetimes.}

So far, we have not talked about the \emph{symplectic flux} $\iota_{\xi} \TTh$ in Eq.~\eqref{eq:1stlawnoomega}. Remember that $\xi$ is the Killing vector field of a stationary black hole. Assuming that (i) the horizon is bifurcating so that we can take $\H$ to be the bifurcation surface (where $\xi = 0$), and (ii) the integrand $\iota_\xi \TTh$ is regular on $\H$, the symplectic flux contribution vanishes. This is what happens in General Relativity. On the other hand, the second condition is not generically satisfied in alternative theories of gravity. In particular, it fails in the case of Einstein--\AE{ther} theory because the \AE{ther} flux diverges as we approach the Killing horizon. This will be the main point of Subsection \ref{subsec: symp flux AE}, where this behavior will lead us to unveil a contribution to the first law given by the \AE{ther} itself.

Finally, if the horizon-generating Killing vector field is a linear combination of time translation and rotation $\xi = \partial_t + \Omega_\textsc{h} \, \partial_\varphi$, we get
\begin{equation}\label{eq: First Law BH}
    \boxed{\dv M = T_\textsc{bh} \, \dv S_\textsc{bh} + \Omega_\textsc{h} \, \dv J + \int_\H \iota_{\hat{\xi}} \TTh} \,,
\end{equation}
where $\Omega_\textsc{h}$ is the angular velocity of the black hole horizon. The first law of black hole thermodynamics can be obviously extended to include additional terms coming from gauge fields~\cite{Ortin:2002qb,Elgood:2020svt}.

\subsection{Scale Invariance, Extended Thermodynamics, and Smarr Formula}\label{subsec: Weyl Invariance}

An important relation that can be derived from the first law of black hole thermodynamics is the so-called Smarr formula. This formula is relevant as it directly relates the parameters that characterize a black hole solution, unlike the first law which relates infinitesimal variations of the parameters. For example, the Smarr formula for a Schwarzschild black hole in General Relativity in generic spacetime dimensions reads $(n - 3) \, M = (n  - 2) \, T \, S$.

The Smarr formula is deeply connected to the transformation properties of the black hole parameters under \textit{rescaling}. Indeed it can be derived with a scaling argument from the first law: under an overall change of length scale, in pure General Relativity, the parameters of the solution transform homogeneously, with a degree that is fixed by its length dimension; therefore, each parameter is a homogeneous function of the others. Applying Euler's theorem to this function leads to the desired formula.

In order to make contact with the covariant phase space formalism, we reinterpret these considerations in terms of the solution space. Notice that an overall change in length scale can also be interpreted as a transformation of the ``ruler'' we use to measure those scales, i.e.~the metric tensor. 

A dilatation of the metric --- and of any other field in the theory --- is realized by the global part of the group of Weyl transformations. Under an infinitesimal Weyl transformation, the field $\phi$ typically changes as
\begin{equation}
    \delta_W \phi := w(\phi) \, \varepsilon \, \phi \,,
\end{equation}
where $w(\phi)$ is some number called the \emph{Weyl weight} of $\phi$, and $\varepsilon$ is the small parameter of the expansion. We focus here on \emph{global} Weyl transformations, for which $\varepsilon$ does not depend on the spacetime point.

A theory is said to be \emph{globally Weyl invariant} if the variation of the action $\S$ under this transformation vanishes. This occurs if there exist at least one combination of Weyl weights such that the contributions from each field cancel exactly. 
A generic theory is not invariant under these transformations, as different terms in the action generally transform with different weights.

As shown in~\cite{Kubiznak:2016qmn}, a possible way to circumvent this issue is to allow coupling constants appearing in the action to transform as well, i.e.
\begin{equation}
    \delta_W \, c_i := w(c_i) \, \varepsilon \, c_i \,,
\end{equation}
for suitably chosen weights $w(c_i)$.
The framework in which couplings are allowed to vary under global Weyl transformations is called \emph{extended thermodynamics}~\cite{Kubiznak:2016qmn, Kastor:2016bph, Kastor:2016bph}. 
Within a covariant phase space approach, the transformation of couplings is unnatural. To encompass the extended framework, one must enlarge the configuration space to include couplings as well, that is writing a theory with dynamical couplings.

To avoid any possible confusion, we will denote by $\dve$ and $\widetilde{\vv}$ respectively the vertical exterior derivative and  a generic vector field acting on the enlarged configuration space.

Repeating the computation that led us to Eq.~\eqref{eq: first law with omega} in the extended framework we get
\begin{equation}\label{eq: Ext first law with Omega}
    \begin{aligned}
        \dve \int_{\S_\R} \Bigl( \qq[\widetilde{\vv}_\xi] - \jj[\widetilde{\vv}_\xi] \Bigr) & + \Psi^i \, \dve c_i - \int_{\S_\R} \iota_{\widetilde{\xi}} \BB - \dve \int_\H \qq[\widetilde{\vv}_\xi] + \int_\H \qq[\widetilde{\vv}_{\dve\xi}] + \int_\H \iota_\xi \TTh = - \iota_{\widetilde{\vv}_\xi} \int_\Sigma \oomega \,,
      \end{aligned}
\end{equation}
where 
\begin{equation}\label{eq: variabile coniugata al potenziale chimico}
    \Psi^i := \int_\Sigma \frac{\partial\JJ[\vv_\xi]}{\partial c_i} - \int_{\S_\R} \frac{\partial \jj[\vv_\xi]}{\partial c_i}
\end{equation}
is the conjugated variable to the \emph{chemical potential} $c_i$. See~\cite{Neri:2024qgb} for a more detailed discussion.

Since boundary conditions are not generally invariant under a Weyl transformation, the variation variation induced by a Weyl transformation is not tangent to the boundary compatible configuration space $\B$. To account for this fact, we have to add the term $\int_{\S_\R} \iota_{\xi} \BB$, that generalizes Eq.~\eqref{eq: first law with omega}. Eq.~\eqref{eq: Ext first law with Omega} is the extended thermodynamics version of Eq.~\eqref{eq: dQ = -ivomega}. Similarly, if $\xi$ is a Killing vector field, the right hand side vanishes and we get a relation between the variations on the left hand side.

If we consider again a solution that admits a Killing vector field $\xi$ of the form $\xi = \partial_t + \Omega_\textsc{h} \, \partial_\varphi$, we can use the definitions Eq.~\eqref{eq: dM}, Eq.~\eqref{eq: dJ} and Eq.~\eqref{eq: entropy and temperature} to turn Eq.~\eqref{eq: Ext first law with Omega} into the \textit{first law of black hole extended thermodynamics}:
\begin{equation}\label{eq: ext thermo first law}
    \boxed{\dve M - \Omega_\textsc{h} \, \dve J + \Psi^i \, \dve c_i - T \, \dve S - \int_{\S_\R} \iota_{\xi} \BB = 0} \,.
\end{equation}
If we call $\widetilde{\vv}_W$ the vector on the enlarged configuration space that generates the infinitesimal Weyl transformation, we can contract the previous equation with $\widetilde{\vv}_W$ to obtain
\begin{equation}
    (n-2) \, M = (n-2) \, T \, S + (n-2) \, \Omega_\textsc{h} \, J - w(c_i) \, \Psi^i \, c_i + \int_{\S_\R} \iota_{\xi}(\iota_{\vv_W} \BB) \,.
\end{equation}
The term involving $\BB$ depends on the theory one is considering, but, when evaluated at infinity, it typically reduces to another factor of the total mass~\cite{Neri:2024qgb}, thereby leading to the desired \textit{Smarr formula}
\begin{equation}
     \boxed{(n-3) \, M = (n-2) \, T \, S + (n-2) \, \Omega_\textsc{h} \, J - w(c_i) \, \Psi^i \, c_i} \,.
\end{equation}

\section{Einstein--\AE{ther} Theory}\label{sec: EA Theory}

Consider a spacetime $(\M, g)$ of dimension $n$ and mostly-plus signature. We assume that there also exists a dynamical $1$-form field $\uu$, called the \emph{\AE{ther}}. The dual vector field $g^{\mu\nu} \, \uu_\nu$ will be denoted by $u^\mu$.
Einstein--\AE{ther} theory is a generally covariant vector-tensor theory that describes the dynamics of $g$ and $\uu$.
In the simplest case, the action of this theory can be written as the Einstein-Hilbert action (with a possibly non-zero cosmological constant $\Lambda$), plus the most general two-derivative action for the \AE{ther}, as given in~\cite{Jacobson:2000xp}.
In addition, a Lagrange multiplier $\lambda$ is introduced to enforce $\uu$ to be time-like and of unit norm. Thus we have
\begin{equation}\label{eq: EA action}
    \S_{\textnormal{\AE}}[g, \uu] = \dfrac{1}{16\pi \, G_\mathrm{N}} \, \int_\M \boldsymbol{\epsilon}_\M \Bigl[ R - 2 \, \Lambda + \mathcal{L}_\uu + \lambda \, \Bigl( g^{\mu\nu} \, \uu_\mu \, \uu_\nu + 1 \Bigr) \Bigr] \,,
\end{equation}
where
\begin{equation}\label{eq: Aether Lagrangian}
    \mathcal{L}_\uu = - Z^{\mu\nu} {}_{\rho\sigma} \, \nabla_\mu u^\rho \, \nabla_\nu u^\sigma \,,
\end{equation}
and~\footnote{This tensor is sometimes called $K$, but we call it $Z$, as in~\cite{Pacilio:2017emh} to avoid confusion with the extrinsic curvature.}
\begin{equation}
    Z^{\mu\nu} {}_{\rho\sigma} = c_1 \, g^{\mu\nu} \, g_{\rho\sigma} + c_2 \, \delta^\mu_\rho \, \delta^\nu_\sigma + c_3 \, \delta^\mu_\sigma \, \delta^\nu_\rho - c_4 \, u^\mu \, u^\nu \, g_{\rho\sigma} \,.
\end{equation}
The kinetic Lagrangian density for the \AE{ther} vector then reads~\cite{Jacobson:2000xp}
\begin{equation}
    \L_\uu = - c_1 \, \nabla_\mu \uu_\nu \, \nabla^\mu u^\nu - c_2 \, \nabla_\mu u^\mu \, \nabla_\nu u^\nu - c_3 \, \nabla_\mu \, \uu_\nu \, \nabla^\nu u^\mu + c_4 \, \aa_\mu \, a^\mu \,,
\end{equation}
where $a^\mu := u^\rho \, \nabla_\rho u^\mu$ is the \AE{ther} \emph{acceleration} vector field.

General covariance is maintained thanks to the fact that $\uu$ is a dynamical field. However, as the expectation value of the \AE{ther} is necessarily different from $0$, every solution breaks local Lorentz invariance because the \AE{ther} provides a preferred frame. As a matter of fact, this theory is intended to be an effective description of Lorentz-violating physics (see e.g.~\cite{Liberati:2013xla} and references therein).

For convenience, we introduce the following shorthand notation for the sum of different couplings
\begin{equation}\label{eq: c_ij}
    c_{i_1,\dots,i_p} := c_{i_1} + \dots + c_{i_p} \,,
\end{equation}
so that, e.g.~$c_{13} := c_1 + c_3$ and $c_{123} := c_1 + c_2 + c_3$.

An interesting feature of this theory is that the action can be written in the so called \textit{fluidodynamical variables}. If we decompose the covariant derivative of the \AE{ther} as
\begin{equation}\label{eq: decomposizione fluidodinamica}
    \nabla_\mu \uu_\nu = \dfrac{1}{n-1} \, p_{\mu\nu} \, \vartheta + \sigma_{\mu\nu} + \omega_{\mu\nu} - \uu_\mu \, \aa_\nu \,,
\end{equation}
where
\begin{equation}\label{eq: hydro quantities}
    p_{\mu\nu} := g_{\mu\nu} + \uu_\mu \uu_\nu \,, \quad \vartheta := \nabla_\mu u^\mu \,, \quad \sigma_{\mu\nu} := p_{(\mu}^\alpha \, p_{\nu)}^\beta \, \nabla_{\alpha} \uu_{\beta} - \dfrac{1}{n-1} \, p_{\mu\nu} \, \vartheta \,, \quad \omega_{\mu\nu} := p_{[\mu}^\alpha \, p_{\nu]}^\beta \, \nabla_\alpha \uu_\beta \,,
\end{equation}
are respectively the \emph{projector} on the leaves orthogonal to the \AE{ther} and the \emph{expansion}, the \emph{shear}, and the \emph{twist} of the \AE{ther} congruence, we can recast the action in the following form
\begin{equation}\label{eq:fluidodynamical action}
    \S_{\textnormal{\AE}}[g, \uu] = \dfrac{1}{16\pi \, G_\mathrm{N}} \, \int_\M \boldsymbol{\epsilon}_\M \, \Bigl[ R - 2\Lambda - c_\vartheta \, \vartheta^2 - c_\sigma \, \sigma^2 - c_\omega \, \omega^2 + c_a \, a^2 + \lambda \, \Bigl( \uu_\mu \, \uu_\nu \, g^{\mu\nu} + 1 \Bigr) \Bigr] \,,
\end{equation}
where
\begin{equation}
    c_\vartheta := c_{13} + (n - 1) \, c_2 \,, \quad c_\sigma := c_{13} \,, \quad c_\omega := c_1 - c_3 \,, \quad c_a := c_{14} \,.
\end{equation}

In principle, we could look for solutions of the Euler-Lagrange equations for any possible value of the couplings $c_i$. However, a lot of these values have been ruled out by physical considerations such as positivity of total energy, absence of naked singularities, and so on~\cite{Jacobson:2007veq}. Most of the solutions known in the literature have been found in sectors of the theory where some of the coefficients of the fluidodynamical description Eq.~\eqref{eq:fluidodynamical action} are set to zero, which is equivalent to make the corresponding term of the \AE{ther} flux a cyclic variable.

As we will see shortly, the coupling constants are also related to the propagation speed of the various degrees of freedom of the theory. For example, the strong observational evidence that the spin-2 perturbation of the metric propagates with the speed of light restricts $c_{13}$ to be very small.
We invite the reader to consult~\cite{Bhattacharyya:2014kta, Jacobson:2007fh, Jacobson:2007veq} and the references therein for a more comprehensive review of these bounds.

\subsection{Relation with the Infrared Limit of Ho\v{r}ava--Lifshitz Gravity}

In~\cite{Jacobson:2010mx}, the author introduced a variant of Einstein--\AE{ther} theory where the \AE{ther} vector is taken to be normalized and hypersurface-orthogonal by construction. At the level of the action Eq.~\eqref{eq:fluidodynamical action}, this can be achieved by taking the formal limit $\lambda\to +\infty\,, \, c_\omega \to +\infty$. This theory has been called ``Khronometric Theory'' or ``$\T$-Theory'', and it has been shown to be the infrared limit of Ho\v{r}ava--Lifshitz gravity~\cite{Jacobson:2010mx}. In Khronon theory, the form field $\uu$ can be traded for a scalar field $\T$ (called the \textit{Khronon}) defined by the relation
\begin{equation}\label{eq: u in terms of T}
    \uu := - N \, d \T \,,
\end{equation}
where $N := (- g^{-1}(d \T, d \T))^{-1/2}$ is a normalization factor and the minus sign is chosen such that $\T$ grows in the direction of the vector field $u$, namely $u^\mu \, \partial_\mu \T > 0$. The level sets of the Khronon $\T$, to which the \AE{ther} is orthogonal, provide a preferred foliation of spacetime, which can be used to simplify the canonical analysis of the theory.

Since the unit-norm constraint is implicit in the definition Eq.~\eqref{eq: u in terms of T} of $\uu$, the resulting action is nothing but the Einstein--\AE{ther} action without the Lagrange multiplier, and with $\uu$ being a shorthand for the normalized gradient of $\T$ as in Eq.~\eqref{eq: u in terms of T},
\begin{equation}
    \S_\T[g,\T] = \dfrac{1}{16\pi \, G_\mathrm{N}} \, \int_\M \boldsymbol{\epsilon}_\M \, \Bigl[ R - 2 \, \Lambda + \L_\uu \Bigr] \,.
\end{equation}
On the other hand, $\uu$ is not necessarily hypersurface-orthogonal in Einstein--\AE{ther} theory. According to Frobenius' theorem, this is true if and only if
\begin{equation}\label{eq: twist}
    \uu \wedge d \uu = 0 \,.
\end{equation}
It is easy to check that this condition is automatically satisfied if $\uu$ has the form Eq.~\eqref{eq: u in terms of T}.
In~\cite{Bhattacharyya:2014kta} it has been shown that the space of solutions of twist-free Einstein--\AE{ther} that contain a black hole region is a proper subset of the black hole solution space of Ho\v{r}ava--Lifshitz gravity, while the two solution spaces are equivalent when restricted to static and spherically symmetric configurations in dimension $n = 4$.

\subsection{Universal Horizons}\label{subsec:UH}

In Lorentz-violating theories of gravity such as Einstein–\AE{ther} theory, the causal structure of spacetime differs fundamentally from that of General Relativity. The presence of the \AE{ther} field selects a preferred local inertial frame and therefore breaks local Lorentz invariance. As a consequence, causality is no longer governed by a single metric light cone. Instead, field perturbations decompose into several independent modes, each propagating with its own characteristic speed relative to the \AE{ther} frame. The causal structure is therefore described by a set of distinct light cones rather than a unique one.

In Einstein–\AE{ther} theory there are three propagating modes: a spin-2 tensor, a spin-1 vector, and a spin-0 scalar mode. The tensor mode corresponds to the usual graviton. Their propagation speeds are given by
\begin{equation}\label{eq: mode speed}
    \begin{cases}
        c^2_{spin \, 2} = \dfrac{1}{(1 - c_{13})} \, c^2 \,, \\
        c^2_{spin \, 1} = \dfrac{2 \, c_1 - c_1^2 + c_3^2}{2 \, c_{14}(1 - c_{13})} \, c^2 \,, \\
        c^2_{spin \, 0} = \dfrac{c_{123}}{c_{14}} \, \dfrac{2 - c_{14}}{[2 \, (1 + c_2)^2 - c_{123} \, (1 + c_2 + c_{123})]} \, c^2 \,,
    \end{cases}
\end{equation}
where we explicitly showed factors of $c$ (the speed of light in vacuum, everywhere else equal to $1$). Since these propagation speeds can be superluminal, signals associated with these modes may in principle escape from the usual Killing horizon. The existence of a genuine signal-trapping surface is therefore far from obvious.

Remarkably, black hole solutions can nevertheless exist in these theories. In particular, a special causal boundary known as the \emph{universal horizon} can arise and act as the ultimate barrier for signals, even those that propagate arbitrarily fast. The existence of this structure is closely related to the hypersurface-orthogonality of the \AE{ther} flow, which holds in many relevant solutions.

When this condition holds, the \AE{ther} field $\uu$ can be written locally as the normalized gradient of a scalar field, the Khronon of the previous Subsection.\footnote{Note that the requirement of hypersurface-orthogonality does not need to hold everywhere, as long as Eq.~\eqref{eq: twist} holds at the would-be horizon~\cite{DelPorro:2025zyv}.} Spacetime is therefore locally foliated by constant Khronon hypersurfaces. These leaves define surfaces of simultaneity in the preferred frame, and signals must always propagate toward increasing values of the Khronon field.

In typical situations these surfaces of simultaneity extend to spatial infinity $i^0$.
However, sufficiently strong gravitational fields can cause one particular leaf of the foliation to become compact. Since all signals move toward the future in terms of increasing Khronon, once they cross such a surface they cannot return to the exterior region, regardless of how large their propagation speed may become. The first compact constant-Khronon hypersurface is called the \emph{universal horizon}. Its boundary is future time-like infinity $i^+$.

The universal horizon can be characterized geometrically through the conditions
\begin{equation}\label{eq:UHcond}
    (\iota_\xi \uu) \eval_{r_\textsc{uh}} = 0 \,, \quad \mathrm{and} \quad (\iota_\xi \aa) \eval_{r_\textsc{uh}} \neq 0 \,,
\end{equation}
where again $\xi$ is the Killing vector field associated with stationarity and $a^\mu = u^\nu \, \nabla_\nu u^\mu$ the \AE{ther} acceleration. The first condition states that the \AE{ther} becomes orthogonal to the Killing vector at the universal horizon: since $u^\mu$ is everywhere time-like, this can only occur in a region where $\xi$ is space-like, which implies that the universal horizon must lie inside the Killing horizon. The second condition ensures that the horizon is non-degenerate and that the corresponding surface gravity is non-vanishing. Indeed, the surface gravity of the universal horizon can be defined as
\begin{equation}\label{eq: kappa UH}
    \kappa_\textsc{uh} = \frac{1}{2} \, (\iota_\xi \aa) \eval_{r_\textsc{uh}} \,,
\end{equation}
and interpreted in terms of the peeling behavior of constant-Khronon surfaces near the horizon~\cite{Cropp:2013sea}. 

In the presence of an \AE{ther} field, it is possible to construct covariant field equations describing superluminal matter fields with modified dispersion relations 
\begin{equation}\label{eq: modified dispersion relation}
    \omega^2 = k^2 + \sum\limits_{j \, = \, 2}^n \dfrac{\beta_{2j}}{\Lambda^{2j - 2}_\textsc{uv}} \, k^{2j} \,,
\end{equation}
where $\omega$ and $k$ are respectively the frequency and the modulus of the wave vector of the superluminal mode in an adapted frame, $\beta_{2j}$ are coefficients depending on the ultraviolet completion of the theory, and $\Lambda_{\textsc{uv}}$ is an ultraviolet cut-off. It was recently proven that such superluminal matter field radiate from the universal horizon with a Hawking temperature $T_\textsc{uh} = \kappa_\textsc{uh}/\pi$ which is afterwards partially reprocessed at the Killing horizon for low energy modes~\cite{DelPorro:2023lbv}.

Universal horizons were first identified in static asymptotically flat black hole solutions of Einstein–\AE{ther} theory~\cite{Blas:2011ni,Barausse:2011pu,Berglund:2012fk}. Similar structures have subsequently been found in several other classes of solutions, including
\begin{itemize}
    \item four-dimensional static asymptotically (Anti-) de Sitter spacetimes in the sector $c_{14} = 0$~\cite{Bhattacharyya:2014kta};
    \item four-dimensional stationary asymptotically flat slowly rotating black holes in the sector $c_{14} = 0$~\cite{Barausse:2012qh};
    \item three-dimensional rotating solutions in the sector $c_{14} = 0$~\cite{Sotiriou:2014gna};
    \item three- and four-dimensional charged static black holes in both sectors $c_{14} = 0$ and $c_{123} = 0$~\cite{Ding:2016wcf, Ding:2015kba, Ding:2015fyx}.
\end{itemize}
In many cases of interest, static spherically symmetric solutions can be parametrized as
\begin{equation}\label{eq: parametrization g and u}
    \begin{cases}
        ds^2 = - e(r) \, dt^2 + \dfrac{dr^2}{e(r)} + r^2 \, d\Omega_2^2 \,, \\
        \uu = (\iota_\xi \uu) \, dt - \dfrac{\iota_\xi \ss}{e(r)} \, dr \,,
    \end{cases}
\end{equation}
where $\xi = \partial_t$ is a Killing vector of both the metric and the \AE{ther} field (namely $\mathcal{L}_\xi g = 0 = \mathcal{L}_\xi \uu$) and $\ss$ is a normalized 1-form orthogonal to $\uu$ that spans the radial plane together with it,
\begin{equation}
    \ss = (\iota_\xi \ss) \, dt - \dfrac{\iota_\xi \uu}{e(r)} \, dr \,.
\end{equation}

In the \AE{ther} frame, defined by $\uu \sim d \tau$, the metric takes the form
\begin{equation}
    ds^2 = (\iota_\xi \uu)^2 \, d\tau^2 + \frac{(\iota_\xi \uu)^2}{e^2(r)} \, d\rho^2 + r^2 \, d\Omega_2^2 \,,
\end{equation}
for a suitably redefined spatial coordinate $\rho$ (see~\cite{DelPorro:2022kkh} for the explicit construction). This expression makes it clear that, besides the usual Killing horizons defined by $e(r) = 0$, the metric also becomes degenerate at radii where $(\iota_\xi \uu) = 0$. These loci correspond precisely to universal horizons. See Fig.~\ref{fig:KhronoSlices}.

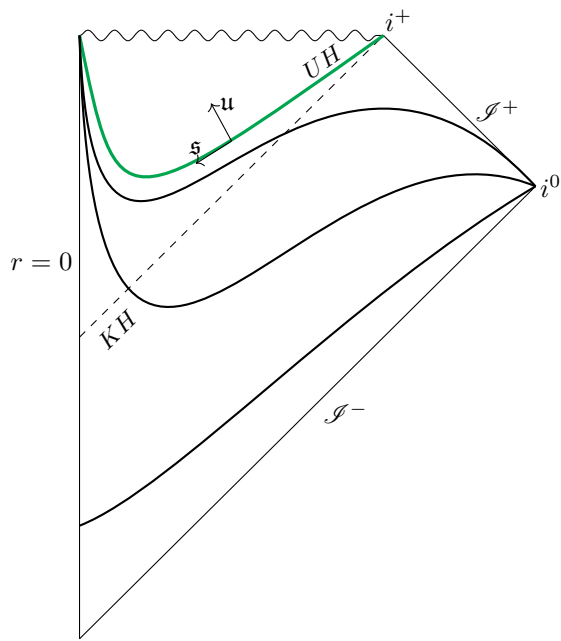
\begin{figure}[H]
    \centering
        \scalebox{1}{
        \begin{tikzpicture}
            \draw (-2,2) -- (-2,-6);
            \draw (-2,-6) -- (4,0);
            \draw (4,0) -- (2,2);
            \draw[dashed] (-2,-2) -- (2,2);
            \draw[decorate, decoration={snake, amplitude=0.7mm}] (-2,2) -- (2,2);
            
            \draw[Green, very thick] (-2,2) .. controls (-1.5,-0.5) .. (2,2);
            \draw[thick] (-2,2) .. controls (-1.8,-3.5) and (0.8,3.5) .. (4,0);
            \draw[thick] (-2,2) .. controls (-1.8,-5.5) and (1,1.25) .. (4,0);
            \draw[thick] (-2,-4.5) .. controls (-0.75,-4) and (1.5,-1.5) .. (4,0);
                
            \draw (1.5,-3) node {{$\mathscr{I}^-$}};
            \draw (3.5, 1) node {{$\mathscr{I}^+$}};
            \draw (4.2,0) node {{$i^0$}};
            \draw (2.2,2.2) node {{$i^+$}};
            \draw (-2.5,-1) node {{$r = 0$}};
                
            \node[rotate=45] at (-1.5,-1.9) {\small $KH$};
            \node[rotate=35] at (1.2,1.65) {\small$UH$};
                
            \draw[->] (0,0.6) -- (-0.27,1.1) node[anchor=west] {{$\uu$}};
            \draw[->] (0,0.6) -- (-0.47,0.3) node[anchor=south] {{$\ss$}};
        \end{tikzpicture}
        }
        \caption{\small Carter--Penrose diagram for a Schwarzschild-like solution featuring a universal horizon and three constant-Khronon slices.}
        \label{fig:KhronoSlices}
    \end{figure}

\section{Covariant Phase Space Analysis of Einstein--\AE{ther} theory}\label{sec: CPS for EA}

In this section, we apply the covariant phase space formalism that we reviewed in Section~\ref{sec: CPS Formalism} to the solutions of Einstein--\AE{ther} theory. This analysis constitutes the main contribution of this work.

Let us consider again the action of the theory
\begin{equation}\label{eq: EAE action}
    \S_{\textnormal{\AE}}[g, \uu] = \dfrac{1}{16\pi\, G_\mathrm{N}} \, \int_\M \boldsymbol{\epsilon}_\M \Bigl[ R - 2 \, \Lambda + \mathcal{L}_\uu + \lambda \, \Bigl( \uu_\mu \, \uu_\nu \, g^{\mu\nu} + 1 \Bigr) \Bigl] \,,
\end{equation}
with $\mathcal{L}_\uu$ given in Eq.~\eqref{eq: Aether Lagrangian}.

As a first step to apply the algorithm of covariant phase space, we need to compute the variation of this action:~\footnote{Notice that $\dv u^\alpha$ is the variation of the vector field dual to the \AE{ther}. When needed, we will use the boldface character to denote the variation of the form field. Moreover, the tensor $\dv g^{\mu\nu}$ is the variation of the metric $g$ with both indices raised and \textit{not} the variation of the inverse metric $g^{\mu\nu}$.}
\begin{equation}\label{eq: variation S_EA}
    \begin{split}
            \dv \S_\textnormal{\AE} = & \dfrac{1}{16\pi \, G_\mathrm{N}} \, \int_\M \epsilon_\M \Bigl[ \Bigl( G_{\mu\nu} + g_{\mu\nu} \, \Lambda - T^{\textnormal{\AE}}_{\mu\nu} \Bigr) \, \dv g^{\mu\nu} + 2 \, \Bigl( \textnormal{\AE}_\mu + \lambda \, \uu_\mu \Bigr) \, \dv u^\mu + \Bigl( \uu_\mu \, \uu_\nu \, g^{\mu\nu} + 1 \Bigr) \, \dv \lambda \Bigr] + \\
            & - \dfrac{1}{16\pi \, G_\mathrm{N}} \, \int_\M \epsilon_\M \nabla_\mu \left[ g^{\alpha\beta} \, \nabla^\mu \dv g_{\alpha\beta} - g^{\mu\beta} \, \nabla^\alpha \dv g_{\alpha\beta} + X^\mu {}_{\alpha\beta} \, \dv g^{\alpha\beta} + 2 \, Y^\mu {}_\alpha \, \dv u^\alpha \right] \,.
    \end{split}
\end{equation}
Let us analyze the various terms that appear in this expression. First of all, for convenience, we defined the $X$ and $Y$ tensors
\begin{subequations}
    \begin{gather}
        Y^\mu {}_\rho := Z^{\mu\nu} {}_{\rho\sigma} \, \nabla_\nu u^\sigma\,,\\
        X^\mu {}_{\alpha\beta} := u^\mu \, Y_{(\alpha\beta)} + \uu_{(\alpha} \, Y^\mu {}_{\beta)} - \uu_{(\beta} \, Y_{\alpha)} {}^\mu \,,
        \end{gather}
\end{subequations}
as they will appear frequently in our computations.
In the first line of Eq.~\eqref{eq: variation S_EA}, we also introduced
\begin{subequations}
    \begin{gather}
        \textnormal{\AE}_\mu := c_4 \, \aa_\rho \, \nabla_\mu u^\rho + \nabla_\rho Y^\rho {}_{\mu} \,,\\
        T^\textnormal{\AE}_{\mu\nu} := c_1 \, (\nabla_\mu \uu_\rho \,  \nabla_\nu u^\rho - \nabla_\rho \uu_\mu \, \nabla^\rho \uu_\nu) + c_4 \, \aa_\mu \, \aa_\nu + \nabla_\rho  X^\rho {}_{\mu\nu} + \lambda \, \uu_\mu \, \uu_\nu + \dfrac{1}{2} \, g_{\mu\nu} \, \L_\uu\,,
        \end{gather}
\end{subequations}
with the former entering in the \AE{ther} field equations and the latter being the \AE{ther} stress-energy tensor.

From the bulk piece of Eq.~\eqref{eq: variation S_EA}, we read the equations of motion of the theory
\begin{equation}\label{eq: eom Einstein-Aether}
    \begin{cases}
        & G_{\mu\nu} + g_{\mu\nu} \, \Lambda = T^{\textnormal{\AE}}_{\mu\nu} \,;\\
        & \textnormal{\AE}_\mu + \lambda \, \uu_\mu = 0 \,;\\
        & \uu_\mu \, u^\mu = - 1 \,.
    \end{cases}
\end{equation}
Since $\uu$ is normalized, we can remove $\lambda$ from the second equation and write
\begin{equation}\label{eq: eom EA without lambda}
    \begin{cases}
        & G_{\mu\nu} + g_{\mu\nu} \, \Lambda = T^{\textnormal{\AE}}_{\mu\nu} \,;\\
        & \underleftarrow{\textnormal{\AE}}_\mu = 0\,,
    \end{cases}
\end{equation}
where the left-pointing arrow under a tensor $\underleftarrow{(\_)}$ means that its free indices are projected orthogonally to the \AE{ther} $1$-form with the projector operator $p^\mu_{\nu} = \delta^\mu_{\nu} + u^\mu \, \uu_\nu$ that we introduced in the hydrodynamical description (see Eq.~\eqref{eq: hydro quantities}).
On-shell of the equations of motion for the Lagrange multiplier $\lambda$, we can rewrite the \AE{ther} stress-energy tensor as
\begin{equation}\label{eq: T EA}
    T^{\textnormal{\AE}}_{\mu\nu} = Y_\mu {}^\rho \, \nabla_\nu \uu_\rho - Y^\rho {}_\nu \, \nabla_\rho \uu_\mu - \uu_\mu \, \nabla_\rho Y^\rho {}_\nu + \uu_\mu \, \underleftarrow{\textnormal{\AE}}_\nu + \nabla_\rho X^\rho {}_{\mu\nu} + \dfrac{1}{2} \, g_{\mu\nu} \, \L_\uu \,.
\end{equation}
This will be useful in constructing explicitly the diffeomorphism charge associated to the \AE{ther} as we do in Appendix \ref{app: Noecharge}, because charges are evaluated on-shell and in this form it is manifest now that one term vanishes when going on-shell of the equations of motion of the \AE{ther}.

From the boundary piece of the variation of the action Eq.~\eqref{eq: variation S_EA}, we can read off the symplectic potential of the theory.
We recognize that the first two terms give the boundary symplectic potential of pure General Relativity
\begin{equation}
    \TTh_\textsc{gr} = - \dfrac{1}{16\pi \, G_\mathrm{N}} \, \pqty{g_{\alpha\beta} \, \nabla^\mu \dv g^{\alpha\beta} - \nabla_\alpha \dv g^{\mu\alpha}} \, \boldsymbol{\epsilon}_\mu \,,
\end{equation}
so that we can identify the last two as the \AE{ther} contribution
\begin{equation}
    \TTh_\textnormal{\AE} := - \dfrac{1}{16\pi \, G_\mathrm{N}} \, \left( X^{\mu\alpha\beta} \, \dv g_{\alpha\beta} + 2 \, Y^\mu {}_\alpha \, \dv u^\alpha \right) \, \boldsymbol{\epsilon}_\mu \,.
\end{equation}
The total symplectic potential of the theory can thus be written as
\begin{equation}\label{eq: total symplectic potential}
    \TTh_\textsc{tot} = \TTh_\textsc{gr} + \TTh_\textnormal{\AE} \,,
\end{equation}
where $\boldsymbol{\epsilon}_\mu = \iota_{\partial_\mu} \boldsymbol{\epsilon}$.

\subsection{Diffeomorphism Charges for Einstein--\AE{ther}}\label{subsec: diffeo charges EAE}

With these ingredients, we can construct the diffeomorphism charge as defined in Eq.~\eqref{eq: Noecharge}. A more detailed derivation is included in Appendix~\ref{app: Noecharge}. The result is
\begin{equation}
    \begin{split}
        \qq_\textsc{tot}[\vv_\xi] & = - \dfrac{1}{8\pi \, G_\mathrm{N}} \, \pqty{\nabla^{[\mu} \xi^{\nu]} + u^{[\mu} \, Y^{\nu]} {}_{\rho} \, \xi^\rho + u^{[\mu} \, Y_\rho {}^{\nu]} \, \xi^\rho + Y^{[\mu\nu]} \, \uu_\rho \, \xi^\rho} \,  \boldsymbol{\epsilon}_{\mu\nu} = \\
        & = \qq_K[\vv_\xi] + \qq_{\textnormal{\AE}}[\vv_\xi] \,,
    \end{split}
\end{equation}
where $\boldsymbol{\epsilon}_{\mu\nu} = \iota_{\partial_\mu} \iota_{\partial_\nu} \boldsymbol{\epsilon}_\M$, and we identified the contribution coming from the metric field, which is the usual \emph{Komar charge density} of General Relativity,
\begin{equation}\label{eq: Komar charge}
    \qq_K[\vv_\xi] := -\dfrac{1}{16\pi \, G_\mathrm{N}} \, \Bigl( \nabla^\mu\xi^\nu - \nabla^\nu\xi^\mu \Bigr) \, \boldsymbol{\epsilon}_{\mu\nu} \,,
\end{equation}
and that coming from the \AE{ther} field
\begin{equation}\label{eq: carica dell'etere}
    \qq_{\textnormal{\AE}}[\vv_\xi] := - \dfrac{1}{8\pi \, G_\mathrm{N}} \, \pqty{u^{[\mu} \, Y^{\nu]} {}_{\rho} \, \xi^\rho + u^{[\mu} \, Y_\rho {}^{\nu]} \, \xi^\rho + Y^{[\mu\nu]} \, \uu_\rho \, \xi^\rho} \, \boldsymbol{\epsilon}_{\mu\nu} \,.
\end{equation}
The presence of two pieces already suggests the existence of an independent \AE ther contribution to the entropy, as defined in Eq.~\eqref{eq: entropy and temperature}. In the following we will provide a precise derivation of this claim.

Quite interestingly, we can see that, just as
\begin{equation}\label{eq: Komar current}
    \nabla_\nu \, q_K^{\mu\nu} \, = - \dfrac{1}{16\pi \, G_\mathrm{N}} \, \bigl( 2 \, G^{\mu} {}_\nu + R \, \delta^\mu_\nu \bigr) \, \xi^\nu \,,
\end{equation}
as in Komar's original construction~\cite{PhysRev.113.934}, it similarly holds that
\begin{equation}\label{eq: aether current}
    \nabla_\nu \, q_{\textnormal{\AE}}^{\mu\nu} = \dfrac{1}{16\pi \, G_\mathrm{N}} \pqty{2 \, T_{\textnormal{\AE}}{}^{\mu}{}_\nu - \L_\uu \, \delta^\mu_\nu} \, \xi^\nu \,,
\end{equation}
upon the sole imposition of the equations of motion for the \AE{ther}, $\underleftarrow{\textnormal{\AE}}_\mu = 0$ (since we do not need to go fully on-shell, we did not use $\doteq$ in this equality). The analogy between these two equations is very intriguing. 
The conservation equation Eq.~\eqref{eq: Komar current} is a consequence of Bianchi identities, which means that it follows from the constraints associated with diffeomorphism invariance. It thus appears that the equations of motion for the \AE{ther} --- at least those that are sufficient to guarantee Eq.~\eqref{eq: aether current} --- can be interpreted as the diffeomorphism constraints for the \AE{ther} contribution of the theory.

Notice that, from Eq.~\eqref{eq: aether current}, it follows the conservation of the right hand side
\begin{equation}
    \nabla_\mu \left[ \dfrac{1}{16\pi G_\mathrm{N}} \, \left(2 \, T_{\textnormal{\AE}} {}^\mu {}_\nu - \L_\uu \, \delta^\mu_\nu \right) \xi^\nu \right] = \nabla_\mu \nabla_\nu \, q_{\textnormal{\AE}}^{\mu\nu} = 0  \,.
\end{equation}
This last equality will be instrumental for the physical interpretation of the \AE{ther} contributions to the first law. Indeed, this conserved current will be related to the flux of Killing energy associated with the \AE{ther} across the horizon, as we will show in greater detail in Subsection~\ref{subsec: 1st law}

\subsection{Boundary Conditions for Einstein--\AE{ther} Gravity from Khronometric Theory}

As thoroughly explained in Section~\ref{dynamics}, a proper specification of the dynamics in a manifold with boundaries requires the imposition of boundary conditions.  The standard choice for those boundary conditions in $\T$-theory is Dirichlet boundary conditions for the metric --- i.e.~$\dv i^\ast_\Gamma g_{\mu\nu} = 0$ --- and Neumann boundary conditions for the Khronon --- i.e.~$\underleftarrow{\nabla_\mu \dv \T} = 0$. 
Since the hypersurface-orthogonal sector of Einstein--\AE{ther} theory is equivalent to $\T$-theory, we take the boundary conditions of the latter as a guiding principle.

Notice that the two theories do not have the same dynamics, as Einstein--\AE{ther} theory admits solutions with non-vanishing twist. We will show later that the boundary conditions we obtain are nonetheless enough to ensure compatibility with the variational principle.

If we compute the variation of the \AE{ther} 1-form, locally written in terms of the Khronon as in Eq.~\eqref{eq: u in terms of T}, we get
\begin{equation}
    \dv \uu_\mu = - \dv N \, \nabla_\mu \T - N \, \dv \nabla_\mu \T \,.
\end{equation}
From the normalization condition, the variation of $N$ reads
\begin{equation}
   \dv N = - \dfrac{1}{2} \, N^3 \, \dv g^{\mu\nu} \, \nabla_\mu\T \, \nabla_\nu\T + N^3 \, g^{\mu\nu} \, \nabla_\mu\T \, \dv\nabla_\nu\T \,,
\end{equation}
so that
\begin{equation}\label{eq: aether form variation}
    \dv\uu_\mu = -\frac{1}{2} \, \dv g^{\alpha\beta} \, \uu_a \, \uu_\beta \, \uu_\mu - N \, \underleftarrow{\nabla_\mu\dv\T} \,.
\end{equation}
Using the boundary conditions from $\T$-theory, it is easy to see that the pull-back of this expression vanishes if the vector field $u^\mu$ is tangent at the boundary. 
When the boundary is at infinity, it makes sense to take $\uu$ to be aligned to the time direction of the asymptotic observer, which automatically makes it tangent. For this reason, we shall call
\begin{equation}
    g^{-1}(\nn, \uu)\eval_\Gamma = 0
\end{equation}
the \textit{alignment condition}. Albeit the formalism is suited to work with finite-distance boundaries as well, where one should in principle entail the possibility that the \AE{ther} is not orthogonal to the normal form defining the boundary (non-alignment condition), we shall assume alignment in this work.

Considering Eq.~\eqref{eq: aether form variation}, it follows that, when the \AE{ther} is aligned, a suitable choice of boundary conditions (which is at least compatible with the aforementioned boundary conditions in $\T$-theory) is to enforce Dirichlet conditions for both $g$ and $\uu$, that is
\begin{equation}\label{eq: BC}
    \boxed{\dv q_{ab} = 0, \quad \mathrm{and} \quad \dv \uu_a = 0} \,,
\end{equation}
where $q_{ab}$ and $\uu_a$ are respectively the pullback of the metric and the \AE{ther} 1-form to the boundary.

Given the boundary embedding, we can also construct the \emph{extrinsic curvature} of the boundary~\footnote{For simplicity, we use the same symbol $q$ both for the induced metric to $\Gamma$ and the projector.}
\begin{equation}\label{eq: extrinsic curvature}
    K^{(\Gamma)}_{\mu\nu} := \dfrac{1}{2} \, \pounds_n q_{\mu\nu} =  q_{(\mu}^{\alpha} \, \nabla_\alpha \, \nn_{\nu)} \,,
\end{equation}
which, together with $q_{ab}$ and $\uu_a$, will allow us to construct a suitable boundary Lagrangian using only objects that are intrinsic to the boundary $\Gamma$.

Let us consider the total symplectic potential Eq.~\eqref{eq: total symplectic potential} $\TTh_\textsc{tot} = \TTh_\textsc{gr} + \TTh_{\textnormal{\AE}}$ and pull-back it to the boundary. We know that the symplectic potential of General Relativity can be written as~\footnote{The reader might find the following formulas useful in rederiving the results of this Subsection:
\begin{align*}
        &\dv \nn_\mu = \dfrac{1}{2} \, n^\alpha \, (\delta^\beta _{\mu} - q^\beta _{\mu}) \, \dv g_{\alpha\beta} \,,\\
        &\dv\Gamma_{\alpha\beta} {}^{\mu} = \dfrac{1}{2} \, g^{\mu\nu} \, (\nabla_\alpha \, \dv g_{\nu\beta} + \nabla_\beta \dv g_{\alpha_\nu} - \nabla_\nu \dv g_{\alpha\beta}) \,,\\
        &\dv K^{(\Gamma)} = - \dfrac{1}{2} \, K^{(\Gamma)\mu\nu} \, \dv g_{\mu\nu} + \dfrac{1}{2} \, g^{\mu\nu} \, n^\lambda \, \nabla_\lambda \dv g_{\mu\nu} - \dfrac{1}{2} \, n^\alpha \, \nabla^\beta \, \dv g_{\alpha\beta} + \dfrac{1}{2} \, D_a c^a \\
        &\dv \boldsymbol{\epsilon}_\Gamma = \dfrac{1}{2} \, q^{\mu\nu} \, \dv g_{\mu\nu} \, \boldsymbol{\epsilon}_\Gamma \,,
\end{align*}
}
\begin{equation}\label{eq: ambiguity resolution GR}
    i_\Gamma^* \TTh_\textsc{gr} = \BB_\textsc{gr} - \dv \eell_\textsc{gr} + d \tth_\textsc{gr}
\end{equation}
with
\begin{gather}
    \BB_\textsc{gr} = \dfrac{1}{16\pi \, G_\mathrm{N}} \, \Bigl(K^{(\Gamma)} q^{\mu\nu} - K^{(\Gamma)\mu\nu}\Bigr) \,\dv q_{\mu\nu} \, \boldsymbol{\epsilon}_\Gamma \,,\\
    \eell_\textsc{gr} = \ell_\Gamma \, \boldsymbol{\epsilon}_\Gamma =  \dfrac{1}{8\pi \, G_\mathrm{N}} \, K^{(\Gamma)} \, \boldsymbol{\epsilon}_\Gamma \,,\\
    \tth_\textsc{gr} = \dfrac{1}{16\pi \, G_\mathrm{N}} \, c^\mu \, \ttau_\mu \, \boldsymbol{\epsilon}_{\partial\Sigma} \,,
\end{gather}
where $c^\mu := - q^{\mu\nu} \, n^\alpha \, \dv q_{\nu\alpha}$, and $K^{(\Gamma)} := q^{\mu\nu} \, K^{(\Gamma)}_{\mu\nu} = g^{\mu\nu} \, K^{(\Gamma)}_{\mu\nu}$ is the trace of the extrinsic curvature.
With few algebraic manipulations, we can show that
\begin{equation}
        \nn_\mu \, \Theta_{\textnormal{\AE}}^\mu = - \frac{1}{16\pi\,  G_\mathrm{N}} \, \Bigl[ \nn_\mu \, \Bigl( u^\mu \, Y^{\alpha\beta} - 2 \, u^\alpha \, Y^{(\mu\beta)} \Bigr) \, \dv g_{\alpha\beta} + 2 \, \nn_\mu \, Y^{\mu\nu} \, \dv \uu_\nu \Bigr] \,;
\end{equation}
\begin{equation}\label{eq:aether potential pull-back}
            \implies \quad \boxed{i^\ast_\Gamma \TTh_\textnormal{\AE} = - \frac{1}{16\pi\,  G_\mathrm{N}} \, \left[ P^{ab} \, \dv q_{ab} + 2 \, W^a \, \dv \uu_a + 2 \, W_\perp \, \dv (\iota_n \uu) \right]} \,,
\end{equation}
where we defined
\begin{equation}
    P^{\alpha \beta} := \nn_\mu \, \Bigl( u^\mu \, Y^{\alpha\beta} - 2 \, u^\alpha \, Y^{(\mu\beta)} \Bigr) \,, \quad W^\alpha := \nn_\mu \, Y^{\mu\alpha} \,, \quad W_\perp := \nn_\alpha \, W^{\alpha} \,.
\end{equation}
Notice that the normal components of these tensors do not enter the boundary symplectic potential Eq.~\eqref{eq:aether potential pull-back}, because of the following cancellations 
\begin{equation}
    n^\mu \, \dv \uu_\mu = \dv (n^\mu \, \uu_\mu) + u^\mu \, (\dv \nn_\mu - \mathbf{c}_\mu) \,,
\end{equation}
together with
\begin{equation}
    P^{(\alpha\beta)} \, \nn_\alpha \, \mathbf{c}_\beta + \frac{1}{2} \, W_\perp \, u^\mu \, \mathbf{c}_\mu = 0 \,, \quad P^{\alpha\beta} \, \nn_\alpha \, \nn_\beta + \frac{1}{2} \, (\iota_n \uu) \, W_\perp = 0 \,,
\end{equation}
as desirable for the covariance of the boundary symplectic potential~\cite{Harlow:2019yfa}.

In Eq.~\eqref{eq:aether potential pull-back} we notice that everything vanishes when we impose the Dirichlet boundary conditions~\eqref{eq: BC} and the alignment condition, meaning that the entire right hand side belongs to $\BB_\textnormal{\AE}$.
\begin{equation}
    \boxed{\BB_\textnormal{\AE} = - \frac{1}{16\pi\,  G_\mathrm{N}} \, \left[ P^{ab} \, \dv q_{ab} + 2 \, W^a \, \dv \uu_a + 2 \, W_\perp \, \dv (\iota_n \uu) \right]} \,.
\end{equation}
As a consequence, we can set $\eell_\textnormal{\AE}=0$ and $\tth_\textnormal{\AE}=0$ in the analogous of Eq.~\eqref{eq: ambiguity resolution GR} for $\TTh_\textnormal{\AE}$.

Putting everything together, we can write Eq.~\eqref{eq: dL_1} for the entire theory
\begin{equation}
    \begin{split}
        i_\Gamma^\ast \TTh_\textsc{tot} + \dv \eell_\textsc{tot} =  \BB_\textsc{tot}+ d \tth_\textsc{tot} \,,
    \end{split}
\end{equation}
where
\begin{align}
    \dv \eell_\textsc{tot}&=\dv \eell_\textsc{gr}\,,\\
    \BB_\textsc{tot}&=\BB_\textsc{gr} + \BB_\textnormal{\AE} \,,\\
    \tth_\textsc{tot}&=\tth_\textsc{gr}\,.
\end{align}
This means that the \AE{ther} does not contribute to the corner symplectic potential, as expected from its kinetic term being first order in derivatives.

\subsection{Symplectic Flux for Einstein--\AE{ther} Gravity}\label{subsec: symp flux AE}

An interesting feature of the Einstein--\AE{ther} Lagrangian is that it is almost exact on-shell, with the only term spoiling the exactness being proportional to the cosmological constant contribution
\begin{equation}\label{eq: L on-shell}
    \LL \doteq d \left[\dfrac{1}{16\pi \, G_\mathrm{N}} \, \Bigl( u^\alpha \, Y_\alpha{}^\mu - u^\mu \, Y^\alpha {}_\alpha \Bigr) \, \boldsymbol{\epsilon}_\mu \right] + \dfrac{\Lambda}{8\pi \, G_\mathrm{N}} \, \boldsymbol{\epsilon}_\M \,.
\end{equation}
This phenomenon is reminiscent of the Einstein--Hilbert Lagrangian, which vanishes on-shell when $\Lambda = 0$ and reduces to a constant (times the volume form) otherwise. 

Let us then for now focus on the case $\Lambda = 0$. From Eq.~\eqref{eq: L on-shell} we can define a $(n - 1)$-form $\AA_\textnormal{\AE}$ 
\begin{equation}\label{eq: definition of A}
    \AA_\textnormal{\AE} := \dfrac{1}{16\pi \, G_\mathrm{N}} \, \Bigl( u^\alpha \, Y_\alpha {}^\mu - u^\mu \, Y^\alpha {}_\alpha \Bigr) \, \boldsymbol{\epsilon}_\mu \,,
\end{equation}
in such a way that $\LL \doteq d \AA_\textnormal{\AE}$.
Combining Eq.~\eqref{eq: variation of L_0} and Eq.~\eqref{eq: L on-shell} we get
\begin{equation}\label{eq: d Potenziale simplettico on shell}
    d (\dv \AA_{\textnormal{\AE}}) \doteq \dv \LL \doteq d \TTh_\textsc{tot} \,,
\end{equation}
which leads to
\begin{equation}\label{eq: Potenziale simplettico on shell}
    \TTh_\textsc{tot} \doteq \dv \AA_\textnormal{\AE} + d  \boldsymbol{\zeta} \,,
\end{equation}
where $\boldsymbol{\zeta}$ is a $(n - 2, 1)$-form locally constructed out of fields.\footnote{Given Eq.~\eqref{eq: d Potenziale simplettico on shell}, one concludes that $\TTh_\textsc{tot} - \dv \AA_\textnormal{\AE}$ is a closed form. Thanks to Wald's theorem proved in~~\cite{Wald:1990mme}, we can write it as an exact one.}

The existence of $\AA_\textnormal{\AE}$ allows us to integrate the symplectic flux present in Eq.~\eqref{eq: Ext first law with Omega}
\begin{equation}\label{eq: symplectic flux AE g}
    \int_\H \iota_\xi \TTh \doteq \int_\H \iota_\xi \Bigl( \dv \AA_\textnormal{\AE} \Bigr) + \int_\H \iota_\xi d  \boldsymbol{\zeta} = \int_\H \dv \Bigl( \iota_\xi \AA_\textnormal{\AE} \Bigr) - \int_\H \iota_{\dv\xi} \AA_\textnormal{\AE} + \int_\H \pounds_\xi \boldsymbol{\zeta} \,,
\end{equation}
where, in the second equality, we used Cartan's Magic Formula and dropped the total derivative contribution $d \iota_\xi \boldsymbol{\zeta}$ as the boundary of a boundary is empty. Recalling that the first law holds when $\xi$ is a Killing vector field, the last term vanishes because of stationarity.

On the other hand, the terms involving $\AA_\textnormal{\AE}$ have to be treated with more care. If $\H$ is the bifurcation surface of a Killing horizon, we could naively say that $\iota_\xi \AA_\textnormal{\AE}$ vanishes because $\xi \to 0$ there. However, this is true only if $\AA_\textnormal{\AE}$ is regular on that surface, as we mentioned in the discussion after Eq.~\eqref{eq: entropy and temperature}. In our case, $\AA_\textnormal{\AE}$ diverges on the Killing horizon, which means that the contraction $\iota_\xi \AA_\textnormal{\AE}$ might give us a finite result.

Notice that the terms containing $\AA_\textnormal{\AE}$ in Eq.~\eqref{eq: symplectic flux AE g} appear in the same combination as a charge density (confront Eq.~\eqref{eq.dq-qd}), which suggests the following charge redefinition
\begin{equation}\label{eq: new charge h}
    \boxed{\hh_\textnormal{\AE}[\vv_\xi] := \qq_\textnormal{\AE}[\vv_\xi] - \iota_\xi \AA_\textnormal{\AE}} \,.
\end{equation}

If the cosmological constant does not vanish, $\Lambda \neq 0$, Eq.~\eqref{eq: Potenziale simplettico on shell} does not hold and has to be replaced by
\begin{equation}
    \TTh_\textsc{tot} \doteq \dv \AA_\textnormal{\AE} + \boldsymbol{\eta} \,,
\end{equation}
where $\boldsymbol{\eta}$ is a $(n - 1, 1)$-form satisfying
\begin{equation}
    d \boldsymbol{\eta} = \dv \left( \dfrac{\Lambda}{8\pi \, G_\mathrm{N}} \, \boldsymbol{\epsilon}_\M \right) \,.
\end{equation}
Given this expression, it is reasonable to expect that $\boldsymbol{\eta}$ is regular at the bifurcation surface $\H$, although we do not have a formal proof at the moment. We nevertheless verified that ignoring the $\eta$ contribution does not lead to any inconsistency in the thermodynamical description of the BTZ black hole presented in Subsection \ref{subsec: BTZ}. Assuming $\eta$ is regular, the symplectic flux can again be reduced to $\iota_\xi \AA_{\textnormal{\AE}}$ as we did in in Eq.~\eqref{eq: symplectic flux AE g}.
% Given a Killing vector field $\xi$ such that $\pounds_\xi\boldsymbol{\eta}=0$, we can prove that
% %
% \begin{equation}
%     d (\iota_\xi \boldsymbol{\eta})=-\iota_\xi \dv \left(\dfrac{\Lambda}{8\pi \, G_\mathrm{N}}\, \boldsymbol{\epsilon}_\M\right)=\dv \left(-\dfrac{\Lambda}{8\pi \, G_\mathrm{N}}\, \iota_\xi \boldsymbol{\epsilon}_\M\right)-\iota_{\dv \xi}\left(-\dfrac{\Lambda}{8\pi \, G_\mathrm{N}}\, \boldsymbol{\epsilon}_\M\right)\,.
% \end{equation}
% %
% The last term gives therefore a bulk contribution
% %
% \begin{equation}
%     \int_\H \iota_\xi \boldsymbol{\eta} = \int_\Sigma (\dv \bb[\vv_\xi] - \bb[\vv_{\dv\xi}]) -\int_{\S_\R} \iota_\xi \boldsymbol{\eta}
% \end{equation}
% %
% where $\bb[\vv_\xi] = - \frac{\Lambda}{8\pi \, G_N}\iota_\xi\epsilon_\M$.

Now that we have computed all the relevant quantities, we can proceed in computing the first law for a spacetime endowed with a horizon in Einstein--\AE{ther} theory.

\section{First Law of Einstein--\AE{ther} Black Hole Thermodynamics}
\label{sec: first law for EAE}

In this section, we are going to specialize the general form of the first law given in Eq.~\eqref{eq:1stlawnoomega} with the various terms computed in Section \ref{sec: CPS for EA}. A key point of this section will be the choice of surface $\Sigma$ on which to compute the integral of currents.

\subsection{Finite Speed Modes and their Disformal Frame}\label{subsec: finite speed modes}

As we already mentioned in Subsection \ref{subsec:UH}, black holes in Einstein--\AE{ther} theory generically exhibit a richer causal structure than in General Relativity. Indeed, since matter can be generically coupled to the \AE{ther}, perturbations can propagate at arbitrarily high speeds and any eventual Killing horizon does not represent an actual trapping horizon for those that move faster than light.

In the presence of modes with arbitrarily high propagation speeds, initial data specified on a given surface of simultaneity can propagate essentially instantaneously throughout spacetime. Consequently, any surface that is tilted with respect to the constant–khronon foliation fails to be achronal and is therefore not a valid initial data surface.

Already in the vacuum theory, the presence of three modes with (parametrically) arbitrary propagation speeds implies that the causal structure is not governed by the metric light cones, but rather by those associated with the fastest mode. For this reason, we introduce a probe mode that is taken to propagate with speed $c_s$, larger than that of any other mode in the theory. We dub this the $s$-mode. By construction, the light cones of the 
$s$-mode determine the causal structure of the spacetime, and any causal horizon of the $s$-mode will also trap all other signals.

Finding the explicit position of the causal horizon for the $s$-mode is not trivial. If the $s$-mode is superluminal, we can say for sure that it will ``peel'' from an \emph{effective horizon} located inside the Killing one, but determining the precise location of the former would require integrating the $s$-mode trajectory. In Appendix \ref{app: position of the causal horizon}, we use the idea of peeling to show how the horizon of the $s$-mode can be determined in a simple way in the case of a stationary and spherically-symmetric black hole. 
Consequently, as we will show in Appendix~\ref{app: ray tracing}, their associated temperature is controlled by the \emph{peeling surface gravity}~\cite{Cropp:2013zxi}, a quantity governing the exponential divergence of rays near the effective horizon.

Besides the method we discussed earlier, based on the peeling of rays, there is another one that is widely used in the Einstein--\AE{ther} literature (see e.g.~\cite{Foster:2005ec, Barausse:2012ny}): applying a \emph{disformal transformation}. We will exploit the fact that the disformal transformation maps the causal horizon of the $s$-mode in the original metric to the Killing one in the transformed metric. 

Whereas a conformal transformation scales the metric by an overall factor, a disformal transformation modifies the metric only anisotropically. For the scope of this work, we consider only changes in one direction, which we take to be along the \AE{ther} field. We thus consider only disformal transformations of the form
\begin{equation}
    g_{\mu\nu} \mapsto \overline{g}_{\mu\nu} = g_{\mu\nu} + B \, \uu_\mu \, \uu_\nu \,,
\end{equation}
where $B < 1$ is called the \emph{disformal factor}. In order to maintain the normalization of the \AE{ther} 1-form, we need to transform that as well 
\begin{equation}
    \uu_\mu \mapsto \overline{\uu}_\mu = \sqrt{1 - B} \, \uu_\mu \,.
\end{equation}
It is easy to see that, if $\xi$ was a Killing vector field for $g$ and $\uu$, it is also a Killing vector field for the transformed quantity $\overline{g}$ and $\overline{\uu}$.

Among all disformal transformations, we can use the one with $B = 1 - c_s^2$ to make the velocity vector of the $s$-mode null with respect to the disformal metric $\overline{g}_{\mu\nu}$ (see Appendix~\ref{app: scelta di B} for a proof).
The causal horizon of the $s$-mode is then a null hypersurface which coincide with the standard Killing horizon of $\overline{g}_{\mu\nu}$.

For the rest of this work, we are going to consider the disformal transformation
\begin{equation}\label{eq: disformal gbar to g}
    \begin{cases}
        \overline{g}_{\mu\nu} = g_{\mu\nu} + (1 - c_s^2) \, \uu_\mu \, \uu_\nu \,,\\
        \overline{\uu}_\mu = c_s \, \uu_\mu \,,
    \end{cases}
\end{equation}
We will call the pair $(\overline{g}, \overline{\uu})$ the \textit{disformal frame}, contrasted with the original frame $(g, \uu)$. For convenience, we also present the inverse relation, which gives the original frame in terms of the disformal one
\begin{equation}\label{eq: disformal g to gbar}
    \begin{cases}
        g_{\mu\nu} = \overline{g}_{\mu\nu} + \pqty{1 - \dfrac{1}{c_s^2}} \, \overline{\uu}_\mu \, \overline{\uu}_\nu \,,\\
        \uu_\mu = \dfrac{1}{c_s} \, \overline{\uu}_\mu \,.
    \end{cases}
\end{equation}
Notice that the inverse transformation is fairly simple: the $s$-mode squared speed $c_s^2 = 1 - B$ becomes its reciprocal.

\subsection{Disformal Symmetry}

One of the most interesting properties of Einstein--\AE{ther} theory is how it changes under a disformal transformation.

A tedious but straightforward evaluation of $\S_\textnormal{\AE}[g(\overline{g}, \overline{\uu}), \uu(\overline{g}, \overline{\uu}); c]$ with Eq.~\eqref{eq: disformal g to gbar} (where we made explicit the parametric dependence on the coupling constants, including $G_\mathrm{N}$)\footnote{The symbol $c$, denoting the collection of coupling constants of the theory, should not be confused with the speed of light in vacuum, which is set to $1$ throughout this work, except in Eq.~\eqref{eq: mode speed}.} shows that the action can be written, up to a boundary term, as the same functional evaluated on $(\overline{g},\overline{\uu})$. The only difference is that the coupling constants are mapped to a linear combination thereof. Indeed we have
\begin{equation}\label{eq: action passive transformation}
    \begin{split}
        \S_\textnormal{\AE}[g, \uu; c] = & \dfrac{1}{16\pi \, \overline{G}_\mathrm{N}} \, \int_\M \overline{\boldsymbol{\epsilon}}_\M \Bigl[ \overline{R} - 2 \, \Lambda - \overline{c}_1 \, \overline{\nabla}_\mu \, \overline{\uu}_\nu \, \overline{\nabla}^\mu \, \overline{u}^\nu - \overline{c}_2 \, \overline{\nabla}_\mu \, \overline{u}^\mu \, \overline{\nabla}_\nu \, \overline{u}^\nu - \overline{c}_3 \, \overline{\nabla}_\mu \, \overline{\uu} _\nu \, \overline{\nabla}^\nu \, \overline{u}^\mu + \overline{c}_4 \, \overline{\aa}_\mu \, \overline{a}^\mu \Bigr] +\\&
        + \dfrac{1}{16\pi \, \overline{G}_\mathrm{N}} \, \int_\M \overline{\boldsymbol{\epsilon}}_\M \, \lambda \, \Bigl( \overline{u}^\mu \, \overline{\uu}_\mu + 1 \Bigr) - \dfrac{1 - c_s^{2}}{8\pi\, \overline{G}_\mathrm{N}} \, \int_\M \overline{\boldsymbol{\epsilon}}_\M \, \overline{\nabla}_\mu \, \Bigl( \overline{u}^\mu \, \overline{\nabla}_\nu \,  \overline{u}^\nu \Bigr) \,,
    \end{split}
\end{equation}
where we denoted with a bar all the quantities evaluted in the disformal frame and the transformed couplings. Their explicit expression is given by
\begin{equation}\label{eq: passive coefficients transformations}
    \overline{G}_\mathrm{N} = c_s \, G_\mathrm{N} \,, \quad \mathrm{and} \quad \begin{cases}
        \overline{c}_1 = 1 - \dfrac{1 - c_1 + c_3 + c_s^4 \, (1 - c_{13})}{2 \, c_s^2} \,;\\
        \overline{c}_2 = - 1 + c_s^2 \, (1 + c_2) \,;\\
        \overline{c}_3 = \dfrac{1 - c_1 + c_3 - c_s^4 \, (1 - c_{13})}{2 \, c_s^2} \,;\\
        \overline{c}_4 = \dfrac{1 - c_{1} + c_3 + c_s^4 \, (1 - c_{13}) - 2 \, c_s^2 \, (1 - c_{14})}{2 \, c_s^2} \,.
    \end{cases}
\end{equation}
Extending the shorthand notation Eq.~\eqref{eq: c_ij} to the barred couplings, we notice that, for all $c_s$,
\begin{equation}
    \overline c_{14} = c_{14} \quad \mathrm{and} \quad \overline c_{123} = c_{123} \, c_s^2 \,.
\end{equation}
These remarkable properties imply that the $c_{14} = 0$ and $c_{123} = 0$ sectors of the theory, in which explicit solutions have been found~\cite{Bhattacharyya:2014kta}, are preserved by the disformal transformation.

If we forget about the boundary term, we can summarize this covariance property as
\begin{equation}\label{eq: disformal symmetry of the action}
    \S_\textnormal{\AE}[\overline{g}, \overline{\uu}; \overline{c}] = \S_\textnormal{\AE}[g,\uu; c] \,.
\end{equation}
This relation might be not surprising given that the Einstein--\AE{ther} action is constructed to be the most general second-order action for the \AE{ther}~\cite{Jacobson:2000xp}, hence its form is preserved by disformal transformations.

Even if the equations of motion are preserved by the disformal transformation, and hence we can write that the transformation Eq.~\eqref{eq: disformal gbar to g} maps
\begin{equation}
    \mathrm{Sol}[c] \to \mathrm{Sol}[\overline{c}]
\end{equation}
under a disformal transformation, the additional boundary term might alter the symplectic structure of the theory. However, as long as we impose the alignment condition $g^{-1}(\nn, \uu) = 0$, we can see that the contribution of the new boundary term vanishes on $\Gamma$. This tells us that the boundary symplectic potential we obtain from the covariant phase space algorithm satisfies $\TTh[\overline{g}, \overline{\uu}; c]= \TTh[g, \uu; \overline{c}]$ and that, as long as we impose the same boundary conditions Eq.~\eqref{eq: BC} on $\overline{g}$ and $\overline{\uu}$, the corner symplectic potential changes in the same way. As a result, we can interpret the disformal transformation as a symplectomorphism between the phase spaces defined by the two sets of coupling constants
\begin{equation}\label{eq: phase space covariance}
    \P[c] \to \P[\overline{c}] \,.
\end{equation}
This fact tells us that all the computations we presented so far are equally valid in the disformal frame, as long as we replace the coupling constants with their barred version defined in Eq.~\eqref{eq: passive coefficients transformations}. For example, this means that the Noether charge density evaluated on the disformal frame $\qq_\textsc{tot}(\overline{g},\overline{\uu}; c)$ will have the same functional form as that in the original frame, but with barred constants $\qq_\textsc{tot}(g, \uu; \overline{c})$. This will be enough for what comes next. For convenience, we will denote any expression evaluated on the barred couplings with a bar on top: in our example, we will then write
\begin{equation}
    \overline{\qq}_\textsc{tot}(g, \uu) \equiv \qq_\textsc{tot}(g, \uu; \overline{c}) \,.
\end{equation}

\subsection{On the Relation Between the Two Frames}\label{subsec: 1st law}

Now that we showed the covariance property of the Einstein--\AE{ther} phase space Eq.~\eqref{eq: phase space covariance}, we are ready to formulate the first law of black hole thermodynamics in the disformal frame.

Given a physical solution $(g^\circ, \uu^\circ) \in \P[c]$ which defines the original frame, we know that the disformal frame $(\overline{g}^\circ, \overline{\uu}^\circ)$ --- despite not generally being a solution of the original theory anymore --- is going to be a solution of the transformed theory, thanks to the disformal symmetry Eq.~\eqref{eq: disformal symmetry of the action}.

Therefore, our first step is to formulate the first law in the transformed phase space $\P[\overline{c}]$. This means that, instead of computing quantities in the disformal frame, we just evaluate them in the original frame and replace each coupling constant with its barred version.
Since the first law will hold for every stationary solution in $\P[\overline{c}]$, it will hold also for $(\overline{g}^\circ, \overline{\uu}^\circ)$.

Thanks to our choice of disformal factor, the causal horizon for the $s$-mode in the disformal frame will coincide with the Killing horizon of the metric $\overline{g}$. This allows us to take $\H$ as the bifurcation surface of that horizon and straightforwardly import Wald's construction, as reviewed in Section~\ref{sec: First Law BH}. 
All we need to do then to obtain a consistent first law in the disformal frame is ``put a bar'' on every quantity of Section~\ref{sec: First Law BH}, consistently with the notation we introduced at the end of the previous subsection.

As we already discussed, the symplectic flux contribution at the bifurcation surface, which is usually discarded on account on the metric field regularity, is non-vanishing in Einstein--\AE{ther} theory and can be integrated. If we define the surface gravity~\footnote{All the definitions of surface gravity coincide on a Killing horizon.}
\begin{equation}\label{eq: barred surface gravity}
    \overline{\kappa}_\textsc{bh} = \sqrt{-\dfrac{1}{2} \, {\nabla}_\mu \, \xi_\nu \, {\nabla}^\mu \, \xi^\nu}\eval_\H \,,
\end{equation}
and consider perturbations that preserve the rescaled Killing vector field $\hat{\xi} \equiv \xi/\overline{\kappa}_\textsc{bh}$, we obtain
\begin{equation}
    \dv \overline{\hh}_\textnormal{\AE}[\vv_\xi] - \overline{\hh}_\textnormal{\AE}[\vv_{\dv\xi}] = \overline{\kappa}_\textsc{bh} \, \dv \overline{\hh}_\textnormal{\AE}[\vv_{\hat\xi}] \,,
\end{equation}
with $\hh_\textnormal{\AE}[\vv_\xi]$ defined in Eq.~\eqref{eq: new charge h}, on top of the standard metric contribution. This suggest that the presence of the \AE{ther} leads to an additional entropic contribution
\begin{equation}
    \overline{S}_{\textnormal{\AE}} := 2\pi \, \int_\H \overline{\hh}_\textnormal{\AE}[\vv_{\hat\xi}] \,.
\end{equation}

Taking this into account, we then write
\begin{equation}\label{eq: barred first law}
    \dv \overline{M} = \overline{T}_\textsc{bh} \, \dv \overline{S}_\textsc{gr} + \overline{T}_\textsc{bh} \, \dv \overline{S}_{\textnormal{\AE}} + \overline{\Omega}_H \, \dv \overline{J} \,,
\end{equation}
with a temperature obtained from the surface gravity as
\begin{equation}
    \overline{T}_\textsc{bh} = \dfrac{\overline{\kappa}_\textsc{bh}}{2\pi} \,.
\end{equation}

The relation Eq.~\eqref{eq: barred first law} we just obtained is not yet the law we are after as it does not involve physical quantities. These pertain to solutions of the original theory --- with the unbarred coupling constants. To see how the properties of solutions of $\P[\overline{c}]$ are related to those of $\P[c]$, we need to analyze more carefully the relation between the original and the disformal frame.

As already mentioned, we know that a disformal transformation maps the solution $(g^\circ, \uu^\circ) \in \P[c]$ to a solution $(\overline{g}^\circ, \overline{\uu}^\circ) \in \P[\overline{c}]$, but we are not guaranteed that the functional form of $\overline{g}^\circ(c)$ and $\overline{\uu}^\circ(c)$ in terms of the coupling constants (plus a certain set of parameters determined by the initial and boundary conditions), is going to be the same as $g^\circ(\overline{c})$ and $\uu^\circ(\overline{c})$. Indeed, we do not expect them to be simply related in the general case.

On the other hand, if we assume that the two are connected by a diffeomorphism --- as is the case in all the explicit examples we have considered --- in the sense that $(g^\circ(\overline{c}), \uu^\circ (\overline{c})) = \psi_*(\overline{g}^\circ(c), \overline{\uu}^\circ (c))$ for some $\psi: \M \to \M$, then we can use the covariance of the charges to relate barred quantities to the original ones.
In principle, the diffeomorphism $\psi$ has to be determined on a case-by-case basis, but we sketch here what happens in stationary, spherically-symmetric solutions, which is the case for all the examples in Section~\ref{sec: Examples}. 
First of all, we fix the boundary to be an hypersurface at some large radius $r = \R$ in the coordinates adapted to the isometries $(t, r, \varphi_I)$, where $\varphi_I$ ($I = 1, \dots, n - 2$) is a set of angles. We thus represent the diffeomorphism as $\psi:(t, r, \varphi_I) \to (t', r', \varphi'_I)$.
In order to preserve the boundary, we want the new radial coordinate to be the same as the old one. We also see that, with standard fall-off conditions, the norm of $\partial_t$ in the disformal frame tend to
\begin{equation}
    \overline{g}(\partial_t, \partial_t) = g(\partial_t, \partial_t) + (1 - c_s^2) \, (\iota_{\partial_t} \uu)^2 \to - c_s^2
\end{equation}
as $r \to +\infty$. This means that, asymptotically, we need to require $t' \approx c_s \, t$ to achieve unit time lapse. Similar considerations for $\xi = \partial_t + \Omega_\textsc{h} \, \partial_\varphi$ tell us that
\begin{equation}
    \psi^* \xi = \frac{1}{c_s} \, \xi \,,
\end{equation}
where the $c_s$ factor ultimately takes into account that the disformal frame does not have unit lapse at infinity.

We thus obtain the following relations
\begin{equation}\label{eq: unphysical to physical 1}
    \overline{M} = \dfrac{M}{c_s} \,, \quad \overline{T}_\textsc{bh} = \dfrac{T_\textsc{bh}}{c_s} \,, \quad \overline{\Omega}_H = \dfrac{\Omega_\textsc{h}}{c_s} \,,
\end{equation}
while
\begin{equation}\label{eq: unphysical to physical 2}
    \overline{S}_\textsc{gr} = S_\textsc{gr} \,, \quad \overline{S}_{\textnormal{\AE}} = S_{\textnormal{\AE}} \,, \quad \overline{J} = J \,.
\end{equation}
Notice that the choice of scaling $\Omega_\textsc{h}$ rather than $J$ is ultimately a matter of convention, whereas entropies do not change as they are computed with the rescaled Killing vector field $\hat\xi$, which is insensitive to $c_s$.

Finally, we can state the first law in terms of properties of the original solution
\begin{equation}\label{eq: First Law Einstein-AEther (pezzo irreversibile)}
    \boxed{\dv M = T_\textsc{bh} \, \dv {S}_\textsc{gr} + T_\textsc{bh} \, \dv S_{\textnormal{\AE}} + \Omega_\textsc{h} \, \dv J} \,.
\end{equation}
This formula is the main result of this work. Instrumental is the realization that the \AE{ther} itself contributes to the entropy balance via the heat flux that it carries inside the horizon.

\paragraph{\AE{ther} heat flux}

We expand here on the previous comment about the \AE{ther} contribution to the entropy balance. Let us consider a section of the horizon $H$ that extends from the bifurcation surface $\H$ to some cut in the distant future.
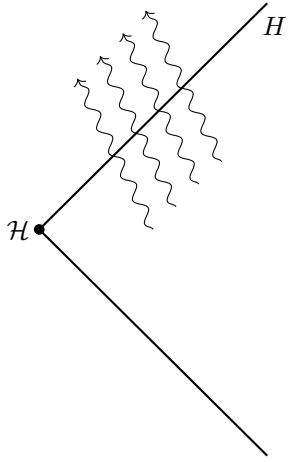
\begin{figure}[H]
    \centering
        \scalebox{1}{
        \begin{tikzpicture}
            \draw[thick] (0,0) -- (3,3);
            \draw[thick] (0,0) -- (3,-3);

            \draw[decorate, decoration={snake, amplitude=0.7mm}, ->] (1.5,0) -- (0.5,2);
            \draw[decorate, decoration={snake, amplitude=0.7mm}, ->] (1.8,0.3) -- (0.8,2.3);
            \draw[decorate, decoration={snake, amplitude=0.7mm}, ->] (2.1,0.6) -- (1.1,2.6);
            \draw[decorate, decoration={snake, amplitude=0.7mm}, ->] (2.4,0.9) -- (1.4,2.9);

            \fill (0,0) circle (2pt) node[anchor = east] {$\H$};
            \node at (3.1,2.7) {$H$};
        \end{tikzpicture}
        }
        \caption{\small \AE{ther} heat flux entering the horizon $H$ of the s-mode in the disformal frame.}
        \label{fig:AEflux}
    \end{figure}

Assuming that the contribution of the charges in the distant future can be neglected, we can straightforwardly integrate Eq.~\eqref{eq: aether current} along $H$ to get
\begin{equation}
    \int_\H \qq_\textnormal{\AE}[\vv_\xi] = \dfrac{1}{8\pi \, G_\mathrm{N}} \, \int_H \pqty{T^{\textnormal{\AE}\mu} {}_\nu -\frac{1}{2} \, \L_\uu \, \delta^\mu_\nu} \, \xi^\nu \, \boldsymbol{\epsilon}_\mu \,.
\end{equation}
Moreover, given Eqs.~\eqref{eq: L on-shell} and~\eqref{eq: definition of A}, we can similarly integrate $\iota_\xi \AA_\textnormal{\AE}$ along $H$ to get
\begin{equation}
    \int_\H \iota_\xi \AA_{\textnormal{{\AE}}} = - \dfrac{1}{16\pi \, G_\mathrm{N}} \, \int_H (\L_\uu + R) \, \xi^\mu \, \boldsymbol{\epsilon}_\mu = - \dfrac{1}{16\pi \, G_\mathrm{N}} \, \int_H \L_\uu \, \xi^\mu \, \boldsymbol{\epsilon}_\mu \,,
\end{equation}
where we discarded the Ricci scalar in view of its regularity because $\xi^\mu \, \boldsymbol{\epsilon}_\mu \sim \xi^2 = 0$. As mentioned several times, this is not true for $\L_\uu$. Considering the difference of these two equations, we finally obtain
\begin{equation}
    \int_\H \hh^\textnormal{\AE}[\vv_\xi] = \dfrac{1}{8\pi \, G_\mathrm{N}} \, \int_H T^{\textnormal{\AE}\mu} {}_\nu \, \xi^\nu \, \boldsymbol{\epsilon}_\mu \,,
\end{equation}
which means that the integral of $\hh^\textnormal{\AE}[\vv_\xi]$ over the bifurcation surface can be interpreted as the total flux of Killing energy associated with the \AE{ther} that has crossed the horizon over a long time (see Fig.~\ref{fig:AEflux}).

Inspired by Jacobson’s celebrated derivation of spacetime thermodynamics~\cite{Jacobson:1995ab}, we then associate the energy flowing outside the causal domain as a heat term in the equations of thermodynamics, namely
\begin{equation}
    \dvbar Q_\textnormal{\AE} := \dfrac{\kappa_\textsc{bh}}{8\pi \, G_\mathrm{N}} \, \dv \int_\H T^{\textnormal{\AE}\mu}{}_\nu \, \hat\xi^\nu \, \boldsymbol{\epsilon}_\mu \,,
\end{equation}
where we used $\dvbar$ to indicate that $Q_\textnormal{\AE}$ is not a phase space function. We then use Clausius' relation to show that the entropic contribution of the \AE{ther} is a genuine thermodynamical contribution associated with the production of heat at a temperature given by the $s$-mode horizon
\begin{equation}
    \dv {S}_{\textnormal{\AE}} = 2\pi \, \dv \int_\H \hh^\textnormal{\AE}[\vv_{\hat{\xi}}] = \frac{\dvbar Q_\textnormal{\AE}}{{T}_\textsc{bh}} \,.
\end{equation}
As we argue at the end of Appendix~\ref{app: ray tracing}, the disformal frame is constructed so that the propagation speed of the $s$-mode is fixed to $\overline{c}_s = 1$. It then follows that the peeling surface gravity Eq.~\eqref{eq: kappa peeeeling} reduces to the standard definition we used in Eq.~\eqref{eq: barred surface gravity}. Consequently, the black hole temperature $\overline{T}_\textsc{bh}$ is the one set by the peeling of rays close to their causal horizon. Taking into account the proper factors of $c_s$ as we did in Eq.~\eqref{eq: unphysical to physical 1}, this shows that the temperature that appears in the first law is the physical one associated with Hawking radiation~\cite{DelPorro:2023lbv}.

\paragraph{Smarr Formula for Einstein--\AE{ther}}\label{subsec: Weyl for AE}

The results of this section also allows us to formule the Smarr formula for Einstein--\AE{ther} gravity. As we reviewed in Subsection~\ref{subsec: Weyl Invariance}, all we have to do is to contract the first law Eq.~\eqref{eq: First Law Einstein-AEther (pezzo irreversibile)} along an infinitesimal Weyl transformation.
Given that the theory is generally not scale invariant, we have to first consider an extended framework, in which couplings are allowed to vary. What we notice is that, unlike the cosmological constant, the \AE{ther} couplings do not need to vary in order to achieve global scale invariance. Hence we consider the infinitesimal transformation
\begin{equation}
    W: \begin{cases}
        g_{\mu\nu} \mapsto 2 \, \varepsilon \, g_{\mu\nu} \,;\\
        \uu_\mu \mapsto \varepsilon \, \uu_\mu \,;\\
        \Lambda \mapsto - 2 \, \varepsilon \, \Lambda \,.
    \end{cases}
\end{equation}
Using Eq.~\eqref{eq: variabile coniugata al potenziale chimico}, we obtain the conjugate variable to the cosmological constant
\begin{equation}\label{eq: variabile coniugata a Lambda}
    \Psi_\Lambda = \int_\Sigma \dfrac{\partial \JJ[\vv_\xi]}{\partial \Lambda} = \int_\Sigma \dfrac{\iota_\xi \boldsymbol{\epsilon}_\M}{8\pi \, G_\mathrm{N}} \,,
\end{equation}
which is basically the 3-volume of the hypersurface $\Sigma$. Notice that, upon regularization, this potential reduces to the `flat' volume of the black hole $V_\textsc{bh}$, as explained in~\cite{Kubiznak:2016qmn}, leading to the known fact that the cosmological constant contributes a work term $V_\textsc{bh} \, \dv p_\Lambda$ to the first law, where $p_\Lambda$ is the pressure associated with the cosmological constant $p_\Lambda := - \Lambda/8\pi \, G_\mathrm{N}$.

Given the extended first law
\begin{equation}
    \dve M = T_\textsc{bh} \, \dve S_\textsc{bh} + T_\textsc{bh}\, \dve S_\textnormal{\AE} + \Omega_\textsc{h} \, \dve J + V_\textsc{bh} \, \dve p_\Lambda + \int_{\S_\R}\iota_{\xi}\BB \,.
\end{equation}
we readily obtain the Smarr formula
\begin{equation}
    \boxed{(n-3) \, M = (n-2) \, T_\textsc{bh} \, S_\textsc{bh} + (n-2) \, T_\textsc{bh} \, S_\textnormal{\AE} + (n-2) \, \Omega_\textsc{h} \, J - 2 \, V_\Sigma \, p_\Lambda} \,. 
\end{equation}

\section{Examples}\label{sec: Examples}

We are now in the condition to explicitly analyze some of the solutions known in Einstein--\AE ther theory. The ones we consider have been found in~\cite{Bhattacharyya:2014kta} and~\cite{Sotiriou:2014gna}. All the results we are showing now have been obtained using the software Wolfram Mathematica\texttrademark, specifically the xAct packages ``xTensor'' and ``xCoba''.

To specialize our discussion, we must first select an appropriate foliation of the spacetime $\M$ into hypersurfaces $\Sigma_t$. The leaves of this foliation are required to have a non-empty intersection with the bifurcation surface of the Killing horizon associated with the disformally transformed metric. We adopt a foliation defined by constant coordinate time slices, $\Sigma_t :={t = const}$. This choice is motivated by the fact that, in the disformal frame of a stationary black hole, these slices always converge at the bifurcation surface and do not penetrate into the causally disconnected region. As discussed in Subsection~\ref{subsec: finite speed modes}, these slices are well-suited to specify initial data given that the $s$-mode is defined to be the fastest propagating mode.

Notice that taking any surface that is not a constant-Khronon one, it is not generically achronal. However, for every generic choice of the parameters $c_i$, we can always find one hypersurface that is Cauchy. 

We also remark that computing the quantities at infinity involves a limiting procedure, as mentioned after Eqs.~\eqref{eq: dM} and~\eqref{eq: dJ}. This limit is typically divergent and so we have to use a regularization prescription to extract finite quantities. As customary, we use the so-called \emph{background subtraction procedure}, which we review in Appendix~\ref{app: background subtraction}. As shown there, the proper definition of the mass is
\begin{equation}
    M = \lim_{\R \to +\infty} \, \left[E_\R - \dfrac{\nu_\R}{\nu^{bg}_\R} \, E_\R^{bg}\right] \,,
\end{equation}
where $E_\R$ is the asymptotic charge associated with the Killing vector $\partial_t$ (called $M$ in Eq.~\eqref{eq: dM}) and $\nu_\R$ is a \emph{redshift factor}.

\subsection{$3+1$ Schwarzschild with $c_{123} = 0$}

The first solution we are going to analyze is an asymptotically flat solution in the sector of coefficients $c_{123} = 0$. Using the parametrization for the metric and the \AE{ther} 1-form presented in Eq.~\eqref{eq: parametrization g and u}, we have
\begin{equation}
    \begin{cases}
        & e(r) = 1 - \dfrac{2 \, r_\textsc{uh}}{r} + C \, \dfrac{r_\textsc{uh}^2}{r^2} \,;\\
        & (\iota_\xi \uu)(r) = - 1 + \dfrac{r_\textsc{uh}}{r} \,;\\
        & (\iota_\xi\ss)(r) = \dfrac{r_\textsc{uh}}{r} \, \sqrt{1 - C} \,,
    \end{cases}
\end{equation}
where $r_\textsc{uh}$ is the position of the universal horizon, $C := \dfrac{(c_{14} - 2 \, c_{13})}{2 \, (1 - c_{13})}$, and the Killing vector field is $\xi = \partial_t$.
The position of the Killing horizon of the disformal metric (the horizon of the $s$-mode) is obtained by solving the equation $\overline{e}(r) = 0$ with the barred couplings. We obtain
\begin{equation}\label{eq: rKH c123}
    \overline{r}_\textsc{h}(c_s) = r_\textsc{uh} \, \left( 1 + \sqrt{1 - \overline{C}} \right) \,,
\end{equation}
where $\overline{C}$ is obtained by $C$ replacing each coefficients with its barred version. When expressed as a function of the original couplings according to Eq.~\eqref{eq: passive coefficients transformations}, the dependence of $\overline{r}_\textsc{h}$ on $c_s$ becomes manifest.
The contribution on $\S_\R$ in Eq.~\eqref{eq: first law with omega} is
\begin{equation}
        \overline{E}_\R = \int_{\S_\R} \left(\overline{\qq}_\textsc{tot}[\vv_\xi] - \iota_{\vv_\xi} \overline{\tth}_\textsc{tot} + \iota_{\xi} \overline{\eell}_\textsc{tot}\right)\,
         \underset{{\R \to +\infty}}{\sim} - \dfrac{\R}{\overline{G}_\mathrm{N}} + \dfrac{\, \left(2 - \frac{\overline{c}_{14}}{2}\right) \, r_\textsc{uh}}{ \overline{G}_\mathrm{N}} + o\left(\dfrac{1}{\R}\right) \,.
\end{equation}
The term for the background subtraction is obtained from this in the limit $r_\textsc{uh}\to 0$
\begin{equation}
    \overline{E}^{bg}_\R \underset{{\R \to +\infty}}{\sim} - \dfrac{\R}{\overline{G}_\mathrm{N}} + o\left(\dfrac{1}{\R}\right) \,.
\end{equation}
Taking into account the redshift factor
\begin{equation}
   \dfrac{\overline{\nu}_\R}{\overline{\nu}_\R^{bg}} := \dfrac{\sqrt{\overline{e}(\R)}}{(\sqrt{\overline{e}(\R)} |_{{r_\textsc{uh} }\to 0})}\, \underset{{\R \to +\infty}}{\sim} 1 -\frac{r_\textsc{UH}}{\R} + o\left(\dfrac{1}{\R}\right) \,.
\end{equation}
Hence the ADM mass turns out to be
\begin{equation}
    \overline{M} = \lim_{\R \to +\infty} \left[ \overline{E}_\R - \dfrac{\overline{\nu}_\R}{\overline{\nu}_\R^{bg}} \, E^{bg}_\R \right] = \left( 1 - \frac{\overline{c}_{14}}{2} \right) \, \dfrac{r_\textsc{uh}}{\overline{G}_\mathrm{N}} \,.
\end{equation}

The surface gravity and the associated temperature are found to be
\begin{equation}
    \overline{\kappa}_\textsc{bh} = \dfrac{\overline{e}'(r)}{2}\eval_{r = r_\textsc{h}(c_s)} = \dfrac{\sqrt{1 - \overline{C}}}{(1 + \sqrt{1 - \overline{C}}) \, \overline{r}_\textsc{h}(c_s)} \, \quad \Rightarrow \quad \overline{T}_\textsc{bh} = \dfrac{\overline{\kappa}_\textsc{bh}}{2\pi} \,.
\end{equation}

The entropy coming from the Komar charge density is the usual Bekenstein--Hawking contribution, while the \AE{ther} one is again proportional to the area because of the high degree of symmetry of the problem:
\begin{equation}
    \overline{S}_\textsc{gr} = 2\pi \, \int_\H \overline{\qq}_K[\vv_{\hat{\xi}}] = \dfrac{\overline{A}_\textsc{h}}{4 \, \overline{G}_\mathrm{N}} \,,
 \end{equation}
and
\begin{equation}
    \overline{S}_\textnormal{\AE} = 2\pi \int_\H \left( \overline{\qq}_{\textnormal{\AE}}[\vv_{\hat{\xi}}] - \iota_{\hat{\xi}} \overline{\AA}_\textnormal{\AE} \right) = - \frac{\overline{A}_\textsc{h}}{4\, \overline{G}_\mathrm{N}}\dfrac{ \overline{c}_{14} \, \left( 1 + \sqrt{1 - \overline{C}} \right) - 2 \, \overline{C}}{2 \, \sqrt{1 - \overline{C}} \, (1 + \sqrt{1 - \overline{C}})} \,.
\end{equation}
One quickly verifies that their sum satisfies the expected Smarr formula
\begin{equation}
    2 \, \overline{T}_{\textsc{bh}}(\overline{S}_\textsc{gr}+\overline{S}_\textnormal{\AE} ) = \frac{8\pi\, \overline{r}_\textsc{h}^2(c_s)}{8\pi\, \overline{G}_\mathrm{N}\, \overline{r}_\textsc{h}(c_s)} \, \left( 1 + \dfrac{c_{14}}{2} \right) \, \left( 1 + \sqrt{1 - \overline{C}} \right) = \overline{M} \,.
\end{equation}

The first law, recast in terms of the physical quantities defined in Eqs.~\eqref{eq: unphysical to physical 1} and~\eqref{eq: unphysical to physical 2}, has thus the form
\begin{equation}
    \dv M = T_\textsc{bh} \, \dv S_\textsc{gr} + T_\textsc{bh} \, \dv S_\textnormal{\AE} \,.
\end{equation}
Notice that, given the symmetry of the solution, one could conflate the entropy contribution of the \AE ther with the Bekenstein--Hawking one, to obtain an effective area law
\begin{equation}
    S_{\textsc{tot}} = S_{\textsc{gr}} + S_{\textnormal{\AE}} = \frac{A_\textsc{h}}{4 \, G_\textnormal{\AE}} \,,
\end{equation}
where the \textit{effective Newton constant}
\begin{equation}\label{eq: Costante di Newton Fisica}
    G_\textnormal{\AE} := \dfrac{G_\mathrm{N}}{1 - \frac{c_{14}}{2}} \,
\end{equation}
is also the one that dictates the proportionality factor between the universal horizon radius and its mass
\begin{equation}
    r_\textsc{uh}=G_\textnormal{\AE{}} \, M,
\end{equation}
as shown in~\cite{Herrero-Valea:2023zex, Jacobson:2007veq}.

\subsection{$2+1$ Rotating BTZ Black Hole in $c_{14} = 0$}\label{subsec: BTZ}

The second example we analyze is a non-stationary and axisymmetric solution in the $c_{14} = 0$ sector that describes a BTZ black hole in $2 + 1$ dimensions. This solution has been found in~\cite{Sotiriou:2014gna} and can be written using the following parametrization
\begin{equation}
    \begin{alignedat}{2}
    ds^2 &= - e(r)\, dt^2 + \frac{dr^2}{e(r)} + r^2 \bigl(d\varphi + \Omega(r)\, dt\bigr)^2\,,
    &\qquad\qquad&
    \begin{aligned}
    \uu &= (\iota_\xi \uu)\, dt - \frac{(\iota_\xi \ss)}{e(r)}\, dr\,, \\
    \ss &= (\iota_\xi \ss)\, dt - \frac{(\iota_\xi \uu)}{e(r)}\, dr\,.
    \end{aligned}
    \end{alignedat}
\end{equation}
The Killing vector field of this solution is $\xi = \partial_t + \Omega_\textsc{h} \, \partial_\varphi$, where $\Omega_\textsc{h} := - \Omega(r_\textsc{h}(c_s))$.

The specific solution is given by
\begin{equation}
    \begin{cases}
        e(r) = - r_0 + \dfrac{\mathcal{J}^2}{4 \, r^2} - \Tilde{\Lambda} \, r^2\,;\\
        \Omega(r) = - \dfrac{\mathfrak{J}}{2 \, r^2}\,;\\
        (\iota_\xi \uu)(r) = - \dfrac{1}{l} \, \left(\dfrac{r^2 - r_\textsc{uh}^2}{r}\right)\,;\\
        (\iota_\xi \ss)(r) = \dfrac{r}{\lambda} + \dfrac{1}{r} \, \sqrt{\dfrac{r_\textsc{uh}^4}{l^2 \, (1 - c_{13})} - \dfrac{\mathfrak{J}^2}{4}}\,,
    \end{cases}
\end{equation}
where we introduced the following auxiliary functions
\begin{subequations}
    \begin{align}
        \Tilde{\Lambda} &= \Lambda - \dfrac{2 \, c_2 + c_{13}}{\lambda^2} =: \dfrac{1}{\lambda^2} - \dfrac{1}{l^2} \,,\\
        \mathcal{J}^2 &= \mathfrak{J}^2 - \dfrac{4 \, c_{13} \, r_\textsc{uh}^4}{l^2 \, (1 - c_{13})} \,,\\
        r_0 &= \dfrac{2 \, r_\textsc{uh}^2}{l^2} + \dfrac{2}{\lambda} \, \sqrt{\dfrac{r_\textsc{uh}^4}{l^2 \, (1 - c_{13})} - \dfrac{\mathfrak{J}^2}{4}} \,.
    \end{align}
\end{subequations}
In this solution, $\lambda$ is known as the \emph{misalignment parameter}. Its name comes from the fact that the \AE{ther vector} is in general not aligned with the vector field generating time translations at infinity, and consequently no first law of black hole thermodynamics can be formulated in this setting. We therefore restrict our attention to the aligned case, corresponding to the limit $\lambda\to\infty$.

Here the radius of the Killing horizon, in the barred world, becomes
\begin{equation}\label{eq: rKH BTZ}
    \overline{r}_\textsc{h}(c_s) = \sqrt{\dfrac{l^2 \, \overline{r}_0 + l \, \sqrt{l^2 \, \overline{r}_0^2 - \overline{\mathcal{J}}^2}}{2}} \,.
\end{equation}
In this case, the computations for the asymptotic mass leads to~\footnote{Since we are in $2+1$, the asymptotic boundary is a circle.}
\begin{equation}
     \begin{split}
        \overline{E}_\R = \int_{\S_\R} \left( \overline{\qq}_\textsc{tot}[\vv_t] - \iota_{\vv_t} \tth_\textsc{tot} + \iota_{t} \overline{\eell}_\textsc{tot} \right) \underset{{\R \maps +\infty}}{\sim} - \dfrac{ \, \R^2}{4 \, \overline{G}_\mathrm{N} \, l^2} + \dfrac{ \, \overline{r}_0}{4 \, \overline{G}_\mathrm{N}} + o \left( \dfrac{1}{\R} \right) \,.
    \end{split}
\end{equation}
Considering again the background value, obtained for $\overline{r}_0 \sim \overline{r}_\textsc{uh}^2 \to 0$, leads us to
 \begin{equation}
      \dfrac{\overline{\nu}_\R}{\overline{\nu}_\R^{bg}} \, \overline{E}^{bg}_\R \underset{{\R \maps +\infty}}{\sim} - \dfrac{\,\R^2}{4 \, \overline{G}_\mathrm{N} \, l^2} + \dfrac{\, \overline{r}_0}{8\, \overline{G}_\mathrm{N}} + o\left( \dfrac{1}{\R} \right) \,,
\end{equation}
and
\begin{equation}\label{eq: massa AdS}
        \overline{M} =\dfrac{\overline{r}_0}{8\, \overline{G}_\mathrm{N}} \,.
\end{equation}
On the other hand, the asymptotic angular momentum does not require background subtraction in this case
\begin{equation}
    \overline{J} = - \lim_{\R \to +\infty} \int_{\S_\R} \left( \overline{\qq}_\textsc{tot}[\vv_\varphi] - \iota_{\vv_\varphi} \tth_\textsc{tot} + \iota_{\varphi} \overline{\eell}_\textsc{tot} \right) = (1 - \overline{c}_{13}) \, \dfrac{\overline{\mathfrak{J}}}{8 \, \overline{G}_\mathrm{N}} \,.
\end{equation}
The surface gravity is given by
\begin{equation}\label{eq: k BTZ}
    \overline{\kappa}_\textsc{bh} = - \dfrac{\overline{\mathcal{J}}^2}{4 \, \overline{r}^3_\textsc{h}(c_s)} + \dfrac{\overline{r}_\textsc{h}(c_s)}{l^2} \, \quad \Rightarrow \quad \overline{T}_\textsc{bh} = \dfrac{\overline{\kappa}_\textsc{bh}}{2\pi} \,,
\end{equation}
while the entropy contributions are
\begin{equation}
    \overline{S}_\textsc{gr} = 2\pi \, \int_\H \overline{\qq}_K[\vv_{\hat{\xi}}] = \dfrac{\overline{P}_\textsc{h}}{4 \, \overline{G}_\mathrm{N}} \,, 
\end{equation}
and
\begin{equation}
    \overline{S}_\textnormal{\AE} = 2\pi \int_\H \left( \overline{\qq}_{\textnormal{\AE}}[\vv_{\hat{\xi}}] - \iota_{\hat{\xi}} \overline{\AA}_\textnormal{\AE} \right) = \dfrac{\overline{c}_{13}}{16\, \overline{G}_\mathrm{N}}\pqty{\dfrac{l^2\, \mathcal{J}^2}{4\,\overline{r}_\textsc{h}^3(c_s)} - \, \overline{r}_\textsc{h}(c_s)} \,,
\end{equation}
where $\overline{P}_\textsc{h}$ is the perimeter of the Killing horizon.

The first law we obtain in terms of the physical quantities is
\begin{equation}
    \dv M = T_\textsc{bh} \, \dv S_\textsc{gr} + T_\textsc{bh} \, \dv S_\textnormal{\AE} + \, \Omega_\textsc{h} \, \dv J \,,
\end{equation}
where the angular momentum term appears.

As this solution features a non-zero cosmological constant, we cannot recover the Smarr formula from the first law only, but we must take into account the conjugate variable to $\Lambda$ as well, which turns out to be the (euclidean) area of the slice $\Sigma$:
\begin{equation}
    \overline{\Psi}_\Lambda =  \int_\Sigma \dfrac{\iota_\xi\boldsymbol{\epsilon}_\M}{8\pi \, G_\mathrm{N}} = - \dfrac{1}{8\pi \, G_\mathrm{N}}\int_{\overline{r}_\textsc{h}(c_s)}^\R dr \, 2\pi \, r.
\end{equation}
Upon regularization, we obtain a chemical potential that is equal to the (euclidean) area enclosed by the $s$-mode horizon
\begin{equation}
     \overline{\Psi}_\Lambda^{(reg)} = \overline{\Psi}_\Lambda - \overline{\Psi}_\Lambda\eval_{\substack{r_0\to 0,\\\mathfrak{J} \to 0}} =  \dfrac{\overline{r}^2_\textsc{h}(c_s)}{8 \, G_\mathrm{N}} \,.
\end{equation}
Therefore, the first law of extended thermodynamics reads
\begin{equation}
     \dve M = T_\textsc{bh} \, \dve S_\textsc{gr} + T_\textsc{bh} \, \dve S_\textnormal{\AE} + \Omega_\textsc{h} \, \dve J - \Psi_\Lambda^{(reg)} \, \dve \Lambda \,,
\end{equation}
yielding a consistent Smarr formula~\footnote{Notice that $n - 3 = 0$, which is the reason why the mass does not enter the Smarr formula in $2 + 1$ dimensions. On the other hand, the coefficient in front of $\Psi_\Lambda \, \Lambda$ does not depend on spacetime dimensionality.}
\begin{equation}
    T_\textsc{bh} \, S_\textsc{gr} + T_\textsc{bh} \, S_\textnormal{\AE} +  \Omega_H \, {J} + 2 \, \Psi_\Lambda^{(reg)} \, \Lambda = 0\,.
\end{equation}

\section{Comparison with Previous Results at the Universal Horizon}\label{sec: comparison}

In this section, we present two among the attempts in the literature at deriving a first law for universal horizons. We compare them and show how our findings help clarify their relation and shortcomings.

In~\cite{Bhattacharyya:2014kta} the authors derive a first law at the universal horizon requiring that the entropy is proportional to the horizon area, as in Bekenstein's result for General Relativity. It was then found, by consistency, that the temperature was fixed by a ``surface gravity''  given by $\kappa_\textsc{h} = (\boldsymbol{\epsilon}^{\mu\nu} \, \nabla_\mu \xi_\nu)/2$, evaluated at the universal horizon. Notice that, since the universal horizon is not a null hypersurface, there is no guarantee that $\nabla_\mu \xi_\nu$ is proportional to $\boldsymbol{\epsilon}_{\mu\nu}$. In spite of its mathematical consistency, the physical interpretation of $\kappa_\textsc{h}$ is quite problematic as no field would radiate from the universal horizon with such a temperature. As the authors of~\cite{Bhattacharyya:2014kta} remark, the so-defined $\kappa_\textsc{h}$ is generally not a surface gravity, but only the redshifted acceleration of a static observer measured at infinity.

A different strategy was pursued in~\cite{Pacilio:2017swi}, where the goal was to construct a thermodynamic first law using $\kappa_\textsc{uh} = (\iota_\xi \aa)/2$, which gives the physical temperature $T_\textsc{uh} = \kappa_\textsc{uh}/\pi$ of the Hawking radiation of infinite-speed particles. Indeed, the trajectory of these particles peels away from the universal horizon with a peeling surface gravity equal to $\kappa_\textsc{uh}$ (see e.g.~\cite{DelPorro:2023lbv}). In this case however, it was found that consistency of the first law did not allow the identification of the black hole entropy with the area of the universal horizon $A_\textsc{uh}$, as evident in the cases where a non-zero cosmological constant $\Lambda$ is turned on.

Satisfactorily, we are now in the condition to clarify the relation between these alternative approaches. 

Considering the discussion in Section~\ref{sec: first law for EAE}, we have so far assumed that the $s$-mode propagates at a finite speed. This requirement ensures that the disformal transformation Eq.~\eqref{eq: disformal gbar to g} remains well-defined. In the limit $c_s \to +\infty$ however, the disformal metric $\overline{g}$ degenerates into a rank-1 tensor and thus becomes non-invertible.

A more consistent way to interpret this limit is to view it as describing test fields with progressively higher propagation speeds, but still finite. As the speed increases, these fields probe wider causal domains, eventually probing the whole space out from the universal horizon. In the extreme limit, the test field becomes sensitive to the causal structure associated with the universal horizon.

Let us pause to stress a crucial subtlety: taking $c_s \to +\infty$ for a \emph{test} field is not the same as sending to infinity the speed of the dynamical modes of the gravitational theory. The speed of the spin-0, spin-1, and spin-2 perturbations is fixed by the coupling constants of the theory (confront Eq.~\eqref{eq: mode speed}), and sending one of those speeds to infinity typically corresponds to a singular region in parameter space. Such limit can alter the structure of the constraints or reduce the number of degrees of freedom, and thus cannot be considered smooth limits. To avoid these pathological behaviors, we only consider external test fields. Said differently, we want to study the causal structure of different spacetime regions within the same theory, and this can be done only using test fields as probes.

As we said, in the limit $c_s \to +\infty$, the relevant causal boundary for the $s$-mode becomes the universal horizon. 
Hawking radiation was computed for superluminal modes with dispersion relations of the form Eq.~\eqref{eq: modified dispersion relation}, and the temperature was found to be exactly $T_\textsc{uh}$~\cite{DelPorro:2023lbv}, which can then be interpret as the physical temperature of the black hole (albeit energy dependent grey-body factors are introduced for low energy modes by the Killing horizon~\cite{DelPorro:2023lbv}). We show in Eq.~\eqref{eq: kappa UH da G} of Appendix~\ref{app: ray tracing}, that in the $c_s \to +\infty$ limit the $s$-mode temperature $T_\textsc{bh}$, as defined in Eq.~\eqref{eq: unphysical to physical 1}, does indeed approach $T_\textsc{uh}$ as expected. Hence, the relevant temperature  governing the thermodynamics of arbitrarily fast modes is the physical temperature determined by the Hawking radiation of the universal horizon.

On the other hand, since we can apply Wald's construction for any high, but finite, speed, the entropy is always given by the Noether charge at the bifurcation surface, which we showed to contain an additional \AE{ther} contribution. Because of this, the entropy cannot simply reduce to the area of the universal horizon or, more precisely, cannot do so in solutions characterized by more than a simple mass scale, like e.g.~those endowed with a cosmological constant. Indeed, in single-scale, spherically symmetric scenarios, dimensional analysis essentially fixes the structure of all quantities at the universal horizon. In these cases, the \AE{ther} flux term can be absorbed into a simple rescaling of the $S/A$ proportionality coefficient, leaving the overall thermodynamics unchanged. In more general cases this is not true and the \AE{ther} flux term has to be treated separately.

For the two examples considered in the previous section, the $c_s \to +\infty$ limit can be exhibited explicitly. In the Schwarzschild-like solution in the $c_{123} = 0$ sector we have
\begin{equation}
    \sqrt{1 - \overline{C}} = \sqrt{\dfrac{1 - C}{c_s^2}} \,, \quad \overset{\eqref{eq: rKH c123}}{\implies} \quad \lim_{c_s \to +\infty} \, \overline{r}_\textsc{h}(c_s) = r_\textsc{uh} \,.
\end{equation}
Similarly
\begin{equation}
    T^{(Schw)}_\textsc{uh} := \lim_{c_s \to +\infty} \, c_s \, \overline{T}^{(Schw)}_\textsc{bh}
    = \lim_{c_s \to +\infty} \, \dfrac{c_s \, \sqrt{1 - \overline{C}}}{(1 + \sqrt{1 - \overline{C}}) \, 2\pi \, \overline{r}_\textsc{h}(c_s)}
    = \dfrac{\sqrt{1 - C}}{2\pi \, r_\textsc{uh}} \,,
\end{equation}
which coincides with the temperature found in~\cite[Eq.~(42)]{Pacilio:2017swi}.

For the BTZ black hole in the $c_{14} = 0$ sector, we can show again that $\overline{r}_\textsc{h}(c_s)$, Eq.~\eqref{eq: rKH BTZ}, approaches $r_\textsc{uh}$ in the limit
\begin{equation}
    \lim_{c_s \to +\infty} \, \sqrt{\dfrac{2 \, r_\textsc{uh}^2 + l \, \sqrt{\frac{2 \, r_\textsc{uh}^2 - \mathcal{J}^2}{c_s^2}}}{2}} = r_\textsc{uh} \,.
\end{equation}
For the temperature given in Eq.~\eqref{eq: k BTZ}, we obtain
\begin{equation}
    T^{(\textsc{BTZ})}_\textsc{uh} := \lim_{c_s \to +\infty} \, c_s \, \overline{T}^{(\textsc{BTZ})}_\textsc{uh}
    = \lim_{c_s \to +\infty} \dfrac{c_s}{2\pi} \, \left( \dfrac{\overline{r}_\textsc{h}(c_s)}{l^2} - \dfrac{\overline{\mathcal{J}^2}}{4 \, \overline{r}^3_\textsc{h}(c_s)} \right)
    = \dfrac{\sqrt{\frac{r_\textsc{uh}^4}{(1 - c_{13}) \, l^2} - \frac{J^2}{4}}}{\pi \, l \, r_\textsc{uh}} \,,
\end{equation}
which agrees with~\cite[Eq.~(61)]{Pacilio:2017swi} in the aligned limit $\lambda \to +\infty$.

So, in summary, our conclusion is that both of the previous approaches where incomplete. Requiring to recover always an area law as in~\cite{Bhattacharyya:2014kta} fixes the temperature to an unphysical value; imposing the physical temperature as in~\cite{Pacilio:2017swi} without taking into account the \AE{ther} entropic contribution yields a gravitational entropy at the universal horizon which does not always agree with the Bekenstein--Hawking formula.

The resolution is that the temperature of the universal horizon is always the physical one $T_\textsc{uh}$, but the \AE{ther} has a non-zero contribution to the entropy balance. Indeed, it can be shown that the \AE{ther} Killing energy flux across the universal horizon is non-zero as long as the \AE{ther} flow is not geodesic, which is necessarily the case if a universal horizon exists in the first place (see the second condition in Eq.~\eqref{eq:UHcond}).

\section{Conclusions}\label{sec:conclus}

In this work we have revisited black hole thermodynamics in vacuum Einstein--\AE{ther} theory, with the intent of studying whether a first law survive --- and in which form --- once the causal structure is no longer universal. 

To address this problem, we employed the covariant phase space formalism \emph{with boundaries} to obtain a robust description of the charges of the theory. Moreover, we adopted the standard disformal strategy that has proven useful in the Einstein--\AE{ther} literature to be able to apply the standard construction à-la Wald. Namely, for a chosen mode with propagation speed $c_s$ (the $s$-mode), we introduced a disformal metric whose Killing horizon coincides with the $s$-mode causal horizon, so that the surface gravity built from the Killing generator matches the peeling surface gravity of that mode by construction~\cite{Cropp:2013zxi}. After we obtained the familiar first law of black hole thermodynamics in the disformal frame, we translated it back to the physical (unbarred) frame. The outcome is a well-defined first law, but with a crucial new ingredient absent in General Relativity: an \AE{ther} \emph{flux} contribution, reflecting the presence of an unavoidable Killing flux of the \AE{ther} across the horizon.

In the same framework, we also formulated an ``extended'' version of the first law which allows variations of the cosmological constant $\Lambda$, when present, and derived the Smarr formula for black hole solutions.

A conceptual by-product of this analysis, as pointed out in Section \ref{sec: comparison}, is a clearer understanding of the status of universal-horizon thermodynamics and of apparent discrepancies in the literature. On the one hand, approaches that enforce an area law for the universal-horizon entropy naturally infer a temperature tied to the metric Killing construction~\cite{Bhattacharyya:2014kta}; on the other hand, ray-tracing/peeling arguments for arbitrarily fast probes determine a different universal-horizon Hawking temperature~\cite{DelPorro:2023lbv}, at the price of requiring an entropy which is not always the horizon area~\cite{Pacilio:2017swi}. Considering the aforementioned $s$-mode as a test field, and taking the $c_s \to +\infty$ limit, we show how the associated surface gravity continuously approaches $\kappa_\textsc{uh}$, as determined in~\cite{Cropp:2013sea}, recovering the physical temperature $T_\textsc{uh} = \kappa_\textsc{uh}/\pi$~\cite{DelPorro:2023lbv}.
The entropy, however, does \emph{not} generically reduce to a pure area term, but splits into the standard (geometric) contribution plus an independent \AE{ther} contribution, sourced by the \AE{ther} Killing flux across the horizon. This resolves the tension between the two prescriptions: fixing the entropy to be purely proportional to area obscures this unavoidable \AE{ther} contribution, while fixing the physical temperature without allowing for this extra contribution forces an apparently ``non-area'' entropy.

Several important caveats and prospects follow. First, a thermodynamic interpretation tied to a single mode is known to be problematic once the full multi-mode content of Lorentz-violating gravity is taken seriously, as distinct sectors generically perceive distinct horizons and temperatures (see~\cite{Dubovsky:2006vk,Eling:2007qd,Jacobson:2010fat}). This points to the need for additional ultraviolet input in order to recover a fully consistent thermodynamics. Indeed, when superluminal modified dispersion relations --- of the type naturally realized in Ho\v{r}ava-like settings Eq.~\eqref{eq: modified dispersion relation} --- are included, the causal structure effectively reorganizes around a unique horizon, the universal horizon, and a unique temperature (up to energy-dependent grey-body factors), $T_\textsc{uh}$~\cite{DelPorro:2023lbv}, is recovered, thereby restoring a consistent thermodynamic picture. In this sense, universal-horizon thermodynamics cannot fully captured by the simple \emph{vacuum} relativistic-dispersion setup in Einstein--\AE{ther} gravity, as only the universal horizon is expected to control the genuinely high-energy thermodynamics.

Second, the infinite-speed limit is intrinsically subtle and not uniformly meaningful across all sectors. In particular, attempting to send the \emph{graviton} speed to infinity renders the disformal construction ill-defined, and regularity requirements typically fail for gravitational modes. In particular, the quantity $\iota_\xi \ss$ must remain regular at the horizon, and this requirement fails generically for gravitons. In the special Schwarzschild-like solution with $c_{123} = 0$, regularity can instead be achieved for the spin-1 mode infinite speed limit, yielding a thermodynamic behavior controlled by $T_\textsc{uh}$.

Finally, we stress that Ho\v{r}ava--Lifshitz gravity cannot be straightforwardly incorporated into our present analysis. Although Einstein--\AE{ther} and Ho\v{r}ava--Lifshitz gravity share the same spherically symmetric black hole solutions (with hypersurface-orthogonal \AE{ther} flow), Ho\v{r}ava--Lifshitz gravity differs in a crucial way: hypersurface-orthogonality is built in, reducing the symmetry from full diffeomorphisms to foliation-preserving diffeomorphisms. This reduction is tied to the presence of an instantaneous (infinite-speed) ``elliptic'' mode~\cite{LukeMartin:2024kof}. Understanding how this mode enters the phase space construction, and how to formulate the covariant phase space consistently in a reduced-symmetry setting, will be the subject of a forthcoming companion work.

%End Mainmatter -----------------------------------------------------

%Begin Backmatter ---------------------------------------------------

\section*{Acknowledgments}

The authors wish to thank David Mattingly, Costantino Pacilio and Francesco Del Porro for illuminating comments and suggestions on a preliminary version of this article.

\appendix

\section{Some Proofs in Covariant Phase Space}

\subsection{Independence of the Slice of Integration}\label{app: proofprop}

Here we give a proof of Proposition~\ref{prop:sliceindep}, which states that
\begin{star-prop}
    The canonical generator defined in Eq.~\eqref{eq: Hamiltonian} 
    \begin{equation*}
        \Q_\Sigma[\vv] :\doteq \int_\Sigma \JJ[\vv] - \int_{\partial\Sigma} \jj[\vv] \,,
    \end{equation*}
    does not depend on the slice of integration $\Sigma$. Moreover the vertical variation of this quantity is related to the symplectic form as in Eq.~\eqref{eq: vOmega = -dH}
    \begin{equation*}
        \mathbb{I}_\vv \OOmega_\Sigma \doteq - \dv \Q_\Sigma[\vv] \,.
    \end{equation*}
\end{star-prop}

\begin{proof}
    Let us consider two slices $\Sigma_a$ and $\Sigma_b$. If we now compute $\Delta \Q[\vv] := \Q_{\Sigma_b}[\vv] - \Q_{\Sigma_a}[\vv]$ we get
    \begin{equation}\label{eq: DeltaH}
        \begin{split}
            \Delta \Q[\vv] = \Q_{\Sigma_b}[\vv] - \Q_{\Sigma_a}[\vv] = &\int_{\Sigma_b}\JJ[\vv] - \int_{\partial\Sigma_b} \jj[\vv] - \int_{\Sigma_a}\JJ[\vv] + \int_{\partial\Sigma_a}\jj[\vv] = \\
            = & \int_\V d \JJ[\vv] - \int_{\Gamma_\V} \JJ[\vv] + \int_{\Gamma_\V} d \jj[\vv] = (\ast) \,,
        \end{split}
    \end{equation}
    where from the first to the second line we added and subtracted the integral of $\JJ[\vv]$ over $\Gamma_\V$, the portion of the boundary $\Gamma$ between $\Sigma_a$ and $\Sigma_b$, to construct (using Stokes' Theorem) a bulk integral over the volume $\V$ enclosed between $\Sigma_a$ and $\Sigma_b$. And again, (using Stoke's Theorem) we converted the integral over the boundaries of the two hypersurfaces of $\jj[\vv]$ to an integral over $\Gamma_\V$. Recall that going from $\partial\Sigma$ to $\partial\Gamma$ we get a minus sign.

    Taking the pull-back of definition Eq.~\eqref{eq: J_1}, using definition Eq.~\eqref{eq: J_2}, Eq.~\eqref{eq: dL_1}, and Eq.~\eqref{eq: sigma2}, we can now compute separately the three integrands
    \begin{gather}
        d \JJ[\vv] = - \EEL \, \mathbb{I}_\vv \dv \phi \,,\\
        i^\ast_\Gamma \JJ[\vv] = \mathbb{I}_\vv i^\ast_\Gamma \TTh - i^\ast_\Gamma \ssigma[\vv] \,, \\
        d \jj[\vv] = i^\ast_\Gamma \TTh - \mathbb{I}_\vv \BB - i^\ast_\Gamma \ssigma[\vv] \,.
    \end{gather}
    Putting everything together, we finally get:
    \begin{equation}
        \begin{split}
            (\ast) & = - \int_\V \EEL \, \mathbb{I}_\vv \dv \phi - \int_{\Gamma_\V} \Bigl( \mathbb{I}_\vv \TTh - \ssigma[\vv] \Bigr) + \int_{\Gamma_\V} \Bigl( \mathbb{I}_\vv \TTh - \mathbb{I}_\vv \BB - \ssigma[\vv] \Bigr) = \\
            & = - \int_\V \EEL \, \mathbb{I}_\vv \dv \phi - \int_{\Gamma_\V}\mathbb{I}_\vv \BB \,.
        \end{split}
    \end{equation}
    Going on-shell, the right hand side of the last equivalence vanishes and this proves that $\Delta \Q[\vv] \doteq 0$~\footnote{With a similar computation it can be shown that $\OOmega$ is also independent on the choice of hypersurface $\Sigma$ chosen in the definition Eq.~\eqref{eq: Omega}.}.

    To prove the second part of Proposition~\ref{prop:sliceindep}, we would like to show that $\dv \Q[\vv] + \mathbb{I}_\vv \OOmega \doteq 0$. To do so, we consider the portion of spacetime $\V$ enclosed between $\Sigma$ and $\Sigma_1$, the bottom lid of the bigger spacetime. As before, we call $\Gamma_\V$ the portion of boundary that is boundary also of the spacetime $\V$. This choice is made because the theory of calculus of variations impose that the variational vector must vanish on the two lids that define the initial and final state and so some quantities will vanish too.
    Going through the computations, we get
    \begin{equation}\label{eq:dvH}
        \begin{split}
            \dv \Q[\vv] + \mathbb{I}_\vv \OOmega = & \dv \left( \int_\Sigma \JJ[\vv] - \int_{\partial\Sigma} \jj[\vv] \right) + \mathbb{I}_\vv \left( \int_\Sigma \dv \TTh - \int_{\partial\Sigma} \dv \tth \right) = \\
            = & \int_\Sigma \Bigl(\mathbb{L}_\vv \TTh - \dv\ssigma[\vv] \Bigr) - \int_{\partial\Sigma} \Bigl( \mathbb{L}_\vv \tth -\dv \sssigma[\vv] \Bigr) = \\
            = &  \int_\V d \Bigl( \mathbb{L}_\vv \TTh -\dv \ssigma[\vv] \Bigr) - \int_{\Gamma_\V} \Bigl( \mathbb{L}_\vv \TTh - \dv \ssigma[\vv] \Bigr) + \int_{\Gamma_\V} d \Bigl( \mathbb{L}_\vv \tth - \dv \sssigma[\vv] \Bigr) = (\ast\ast) \,,  
        \end{split}
    \end{equation}
    where in the first step we used the definitions Eq.~\eqref{eq: Hamiltonian} and Eq.~\eqref{eq: Omega}, in the second one we used the definition of the currents Eq.~\eqref{eq: J_1} and Eq.~\eqref{eq: J_2}, we rearranged terms according to the domains of integration and we used Cartan's Magic Formula on the configuration space to make the configuration space Lie derivative appear. Moreover, in the last equivalence we added and subtracted the integral over the bottom lid (that is zero because everything is linear in $\vv$) and over the boundary $\Gamma_\V$, then we used Stokes' Theorem as we did to prove the first part of this proposition.
    Rearranging the various terms and using the definitions Eq.~\eqref{eq: sigma1} and Eq.~\eqref{eq: sigma2} of $\ssigma[\vv]$ and $\sssigma[\vv]$, it becomes
    \begin{equation}
        \begin{split}
           (\ast\ast) = & \int_\V \mathbb{L}_\vv \Bigl( d \TTh - \dv \LL \Bigr) + \int_\Gamma \mathbb{L}_\vv \Bigl( d \tth - \TTh - \dv \eell \Bigr) = \\
            = & - \int_\V \mathbb{L}_\vv \Bigl( \EEL \, \dv \phi \Bigr) - \int_{\Gamma_\V} \mathbb{L}_\vv \BB \doteq 0 \,,
        \end{split}
    \end{equation}
    which completes the proof of Proposition~\ref{prop:sliceindep}.
\end{proof}

\subsection{Variation of the Canonical Generator}

Here, instead, we compute the variation of the canonical generator, giving a proof of Eq.~\eqref{eq: first law with omega}
\begin{prop}
    The variation of the canonical generator defined in Eq.~\eqref{eq: Hamiltonian} is given by
    \begin{equation*}
        \dv \Q[\vv_\xi] \doteq \dv \int_{\S_\R} \Bigl( \qq[\vv_\xi] - \jj[\vv_\xi] \Bigr) - \dv \int_\H \qq[\vv_\xi] + \int_\H \iota_\xi \TTh \,.
    \end{equation*}
\end{prop}
\begin{proof}
    If we suppose that everything is covariant in the sense of definition Eq.~\eqref{eq: covariance}, this comes from the explicit calculation of the variation of the canonical generator. Let us start from computing the variation $\dv \JJ[\vv_\xi]$
    \begin{equation}
        \begin{split}
            \dv \JJ[\vv_\xi] & = \dv (\iota_{\vv_\xi} \TTh) - \iota_\xi \dv \LL = \\
            & = \dv (\iota_{\vv_\xi} \TTh) - \iota_\xi \EEL \, \dv \phi - \iota_\xi d \TTh = \\
            & = - \iota_{\vv_\xi} \dv \TTh + \mathbb{L}_{\vv_\xi} \TTh - \iota_\xi \EEL \, \dv \phi - \pounds_\xi \TTh + d (\iota_\xi \TTh) \,,
        \end{split}
    \end{equation}
    where in the first step we used the definition of $\JJ[\vv_\xi]$ --- Eq.~\eqref{eq: J_1} --- in the second one Eq.~\eqref{eq: variation of L_0} and in the third Cartan's Magic Formula and the hypothesis that $\TTh$ is covariant.

    For the variation of $\jj[\vv_\xi]$ we get:
    \begin{equation}
        \begin{split}
            \dv \jj[\vv_\xi] & = \dv (\iota_{\vv_\xi} \tth) - \iota_{\overline{\xi}} \dv \eell = \\
            & = \dv (\iota_{\vv_\xi} \tth) - \iota_{\overline{\xi}} \BB - \iota_{\overline{\xi}} d \tth + \iota_{\overline{\xi}} i_{\partial\M}^\ast \TTh = \\
            & = - \iota_{\vv_\xi} \dv \tth + \mathbb{L}_{\vv_\xi} \tth - \iota_{\overline{\xi}} \BB - \pounds_\xi\tth + d (\iota_{\overline{\xi}} \tth) + \iota_{\overline{\xi}} i_{\partial\M}^\ast \TTh \,,
        \end{split}
    \end{equation}
    where in the first step we used the definition of $\jj[\vv_\xi]$ --- Eq.~\eqref{eq: J_2} --- in the second one we substituted the variation Eq.~\eqref{eq: dL_1} and in the third, as before, Cartan's Magic Formula and the covariance of $\tth$.

    Integrating this on a hypersurface $\Sigma$ we get:
    \begin{equation}
        \begin{split}
            \int_\Sigma \dv \JJ[\vv_\xi] - \int_{\partial\Sigma} \dv \jj[\vv_\xi] = & - \iota_{\vv_\xi} \int_\Sigma \dv \Bigl( \TTh - d \tth \Bigr) - \int_\Sigma \iota_\xi \EEL \, \dv \phi + \\
            & + \int_{\S_\R}\iota_{\overline{\xi}} \TTh - \int_\H \iota_{\overline{\xi}} \TTh + \int_{\S_\R} \iota_{\overline{\xi}} \BB + \\
            & - \int_{\partial\Sigma} d \Bigl(  \iota_\xi \tth \Bigr) - \int_{\S_\R}\iota_{\overline{\xi}} \TTh = \\
            = & - \iota_{\vv_\xi} \int_\Sigma \oomega - \int_\H \iota_{\overline{\xi}} \TTh - \int_\Sigma \iota_\xi \EEL \, \dv \phi + \int_{\S_\R} \iota_{\overline{\xi}} \BB \,,
        \end{split}
    \end{equation}
    where in the second step we split the integral over the boundary of $\Sigma$ into an integral over $\S_\R$ and one on $\H$ (recalling that they have opposite orientations and this explains the change in sign) and used the fact that $\int_{\partial\Sigma} d (\iota_\xi \tth)$ vanishes because the boundary of a boundary is empty. 

    Substituting $\JJ[\vv_\xi] = d \qq[\vv_\xi]$ and the definition of $\jj[\vv_\xi]$ --- Eq.~\eqref{eq: J_2} --- and going on-shell, we finally get:
    \begin{equation}
        \dv \int_{\S_\R} \Bigl( \qq[\vv_\xi] - \iota_{\vv_\xi} \tth + \iota_{\overline{\xi}} \eell \Bigr) - \dv \int_\H \qq[\vv_\xi] + \int_\H \iota_{\overline{\xi}} \TTh \doteq -\iota_{\vv_\xi} \int_\Sigma \oomega \,.
    \end{equation}
    This completes the proof of Eq.~\eqref{eq: first law with omega}.
\end{proof}

\subsection{Variation of the Currents in the Extended Framework}
The extended variation of the Noether current is
\begin{equation}
    \begin{split}
        \dve \JJ[\vv_\xi] = & \dv \JJ[\vv_\xi] + \sum\limits_{i = 1}^k \dfrac{\partial \JJ[\vv_\xi]}{\partial c_i} \dve c_i = \\ 
        = & \dv (\iota_{\vv_\xi} \TTh) - \iota_\xi \dv \LL + \sum\limits_{i = 1}^k \dfrac{\partial \JJ[\vv_\xi]}{\partial c_i} \dve c_i = \\
        = & - \iota_{\vv_\xi} \dv \TTh + \pounds_\xi \TTh - \iota_\xi \EEL \, \dv \phi + d (\iota_\xi \TTh) - \pounds_\xi \TTh + \sum\limits_{i = 1}^k \dfrac{\partial \JJ[\vv_\xi]}{\partial c_i} \dve c_i \,,
    \end{split}
\end{equation}
where, in the second step, we used the definition Eq.~\eqref{eq: J_1} and, in the third, we used Eq.~\eqref{eq: variation of L_0} and Cartan's Magic Formula.

On the other hand, the variation of the boundary current $\jj[\vv_\xi]$ yields
\begin{equation}
    \begin{split}
        \dve \jj[\vv_\xi] = & \dv \jj[\vv_\xi] + \sum\limits_{i = 1}^k \dfrac{\partial \jj[\vv_\xi]}{\partial c_i} \dve c_i = \\
        = & \dv (\iota_{\vv_\xi} \tth) - \iota_{\overline{\xi}} \dv \eell + \sum\limits_{i = 1}^k \dfrac{\partial \jj[\vv_\xi]}{\partial c_i} \dve c_i = \\
        = & - \iota_{\vv_\xi} \dv \tth + \pounds_\xi \tth - \iota_{\overline{\xi}} \BB + \iota_{\overline{\xi}} i_\Gamma^\ast \TTh - \pounds_\xi \tth + d (\iota_{\overline{\xi}} \tth) + \sum\limits_{i = 1}^k \dfrac{\partial \jj[\vv_\xi]}{\partial c_i} \dve c_i \,,
    \end{split}
\end{equation}
where, in the second step, we used the definition of $\jj[\vv]$ in Eq.~\eqref{eq: J_2} and, in the third step, we used the definitions of $\eell$, $\BB$ and $\tth$ that we wrote in Eq.~\eqref{eq: dL_1}, together with Cartan's Magic Formula on the hypersurface $\Sigma$.

Finally, after integrating over a slice $\Sigma$ with boundary $\S_\R \cup \H$, we get
\begin{equation}
    \begin{aligned}
        \dv \int_{\S_\R} \Bigl( \qq[\vv_\xi] - \iota_{\vv_\xi} \tth + \iota_{\xi} \eell \Bigr) & + \Psi^i \, \dve c_i - \int_{\S_\R} \iota_{\xi} \BB + \\
        & - \dv \int_\H \qq[\vv_\xi] + \int_\H \qq[\vv_{\dv\xi}] + \int_\H \iota_\xi \TTh = - \iota_{\vv_\xi} \int_\Sigma \oomega \,.
      \end{aligned}
\end{equation}

\section{Computation of the Noether Charge}\label{app: Noecharge}
In this appendix, we are going to show how to derive the Noether charge associated with diffeomorphisms for Einstein--\AE{ther} theory, with a trick that works for every theory of gravity with a dynamical metric field. As mentioned around Eq.~\eqref{eq: Noecharge}, we would like to express the current vector $J^\mu$ as the divergence of an antisymmetric 2-tensor. To do so, we start from considering the symplectic potential
\begin{equation}
     \TTh_\textsc{tot} = - \dfrac{1}{16\pi \, G_\mathrm{N}} \, \left[ g^{\alpha\beta} \, \nabla^\mu \, \dv g_{\alpha\beta} - g^{\mu\beta} \, \nabla^\alpha \, \dv g_{\alpha\beta} + X^{\mu\alpha\beta} \, \dv g_{\alpha\beta} + 2 \, Y^\mu {}_\alpha \, \dv u^\alpha \right] \, \boldsymbol{\epsilon}_\mu \,.
\end{equation}
Contracting it with the vector $\vv_\xi$ and recalling that fields are covariant (so that $\iota_{\vv_\xi} \dv \phi = \mathbb{L}_{\vv_\xi} \phi = \pounds_\xi \phi$), we obtain
\begin{equation}
    \iota_{\vv_\xi} \TTh_\textsc{tot} = - \dfrac{1}{16\pi \, G_\mathrm{N}} \, \left[ g^{\alpha\beta} \, \nabla^\mu \pounds_\xi g_{\alpha\beta} - g^{\mu\beta} \, \nabla^\alpha \, \pounds_\xi g_{\alpha\beta} + X^{\mu\alpha\beta} \, \pounds_\xi g_{\alpha\beta} + 2 \, Y^\mu {}_\alpha \, \pounds_\xi u^\alpha \right] \, \boldsymbol{\epsilon}_\mu \,.
\end{equation}
The first two pieces are due to the Einstein--Hilbert part of the action and we will focus on them first.
\begin{equation}\label{eq: Komar charge trick}
    \begin{split}
        - \dfrac{1}{16\pi \, G_\mathrm{N}} \, \left[ g^{\alpha\beta} \, \nabla^\mu \, \pounds_\xi g_{\alpha\beta} - g^{\mu\beta} \, \nabla^\alpha \,  \pounds_\xi g_{\alpha\beta} \right] & = -\dfrac{1}{16\pi \, G_\mathrm{N}} \, \left[ 2g^{\alpha\beta} \, \nabla^\mu \, \nabla_\alpha \xi_\beta - g^{\mu\beta} \, \nabla^\alpha \, \Bigl( \nabla_\alpha \xi_\beta + \nabla_\beta \xi_\alpha \Bigr) \right] = \\
        & = - \dfrac{1}{16\pi \, G_\mathrm{N}} \, \left[ 2 \, \nabla^\mu \nabla_\alpha \xi^\alpha - \nabla_\alpha \nabla^\alpha \xi^\mu - \nabla_\alpha \nabla^\mu \xi^\alpha \right] \,,
    \end{split}
\end{equation}
where in the second step we used the symmetry of the metric to sum the first two terms and rewrote the other two changing names to dummy indices.

Now here's the trick: we have to use the commutation properties of the covariant derivative applied to a vector to substitute the first term and introduce the Ricci tensor, since next we would like to go on-shell and use Einstein Equations
\begin{equation}
    \nabla^\mu \nabla_\alpha \xi^\alpha - \nabla_\alpha \nabla^\mu \xi^\alpha = - R^\mu_\alpha \, \xi^\alpha \Rightarrow \nabla^\mu \nabla_\alpha \xi^\alpha = \nabla_\alpha \nabla^\mu \xi^\alpha - R^\mu_\alpha \, \xi^\alpha \,.
\end{equation}
This trick is useful in any gravitational theory to get the first part of the charge. However, we have to be brave and do it on the \underline{first} term in the last equality of equation Eq.~\eqref{eq: Komar charge trick}, even though it seems more natural to do it on the second one. 

Hence we get:
\begin{equation}
    \begin{split}
        - \dfrac{1}{16\pi \, G_\mathrm{N}} \, \left[ g^{\alpha\beta} \, \nabla^\mu \, \pounds_\xi g_{\alpha\beta} - g^{\mu\beta} \, \nabla^\alpha \, \pounds_\xi g_{\alpha\beta} \right] & = - \dfrac{1}{16\pi \, G_\mathrm{N}} \, \left[ 2 \, \nabla^\mu \nabla_\alpha \xi^\alpha - \nabla_\alpha \nabla^\alpha \xi^\mu - \nabla_\alpha \nabla^\mu \xi^\alpha \right] = \\
        & = - \dfrac{1}{16\pi \, G_\mathrm{N}} \, \left[ 2 \, \nabla_\alpha \nabla^\mu \xi^\alpha - 2 \, R^\mu_\alpha \, \xi^\alpha - \nabla_\alpha \nabla^\alpha \xi^\mu - \nabla_\alpha \nabla^\mu \xi^\alpha \right] = \\
        & = - \dfrac{1}{16\pi \, G_\mathrm{N}} \, \left[ \nabla_\alpha \, \Bigl( \nabla^\mu \xi^\alpha - \nabla^\alpha \xi^\mu \Bigr) - 2 \, R^\mu_\alpha \xi^\alpha \right] \,.
    \end{split}
\end{equation}
The two terms in round brackets are the \emph{Komar charge density} we also get for pure General Relativity. Indeed, going on-shell for this theory we get that the Einstein Equations are $R_{\mu\nu}(g) = 0$ and so the result is an exact term as expected. In presence of matter, like in our analysis, there is an energy momentum tensor and so we have to be more careful. The Einstein Equations in this case are
\begin{equation}
    R_{\mu\nu} - \dfrac{1}{2} \, g_{\mu\nu} \, R + g_{\mu\nu} \, \Lambda = T^\textnormal{\AE}_{\mu\nu} \Rightarrow - 2 \, R^\mu_\alpha \, \xi^\alpha = - 2 \, T^{\textnormal{\AE}\mu} {}_\alpha \, \xi^\alpha - R \, \xi^\mu + 2 \, \Lambda \, \xi^\mu \,.
\end{equation}
Recalling that
\begin{equation}
     T^{\textnormal{\AE}}_{\mu\nu} = Y_\mu {}^\rho \, \nabla_\nu \uu_\rho - Y^\rho {}_\nu \, \nabla_\rho \uu_\mu - \uu_\mu \, \nabla_\rho Y^\rho {}_\nu + \uu_\mu \, \underleftarrow{\textnormal{\AE}}_\nu + \nabla_\rho X^\rho {}_{\mu\nu} + \dfrac{1}{2} \, g_{\mu\nu} \, \L_\uu \,.
\end{equation}
Again, more than a trick this one is a suggestion to make it easier to go on-shell. It is always possible to express part of the energy-momentum tensor in terms of the equations of motion of the matter it is associated with, so it is worth spending some time to make those equations pop up.

\vspace{0.3cm}
\noindent
Using the \AE{ther} equations of motion Eq.~\eqref{eq: eom Einstein-Aether} and expanding the covariant derivative of $X$, we get that
\begin{equation}
    \begin{split}
        - 2 \, T^{\textnormal{\AE}\mu} {}_\alpha \, \xi^\alpha \doteq & - \L_\uu \, \xi^\mu -Y^{\alpha\mu} \, \xi^\beta \, \nabla_\alpha \uu_\beta + Y^{\mu\alpha} \, \xi^\beta \, \nabla_\alpha \uu_\beta + Y^{\alpha\beta} \, \xi_\beta \, \nabla_\alpha u^\mu - u^\alpha \, \xi^\beta \, \nabla_\alpha Y_\beta {}^\mu - u^\alpha \, \xi^\beta \, \nabla_\alpha Y^\mu {}_\beta + \\
        & - 2 \, Y^{\mu\alpha} \, \xi^\beta \, \nabla_\beta \uu_\alpha - Y^{\alpha\mu} \, \xi_\alpha \, \nabla_\beta u^\beta - Y^{\mu\alpha} \, \xi_\alpha \, \nabla_\beta u^\beta + Y^{\alpha\beta} \, \xi_\alpha \, \nabla_\beta u^\mu + u^\mu \, \xi^\alpha \, \nabla_\beta Y_\alpha {}^\beta + u^\mu \, \xi^\alpha \, \nabla_\beta Y^\beta {}_\alpha + \\
        & - u^\alpha \, \xi_\alpha \, \nabla_\beta Y^{\beta\mu} + u^\alpha \, \xi_\alpha \, \nabla_\beta Y^{\mu\beta} \,,
    \end{split}
\end{equation}
\begin{equation}
    \begin{split}
        \iota_{\vv_\xi} \, \TTh_\textsc{tot} - \dfrac{1}{16\pi \, G_\mathrm{N}} \, \Bigl[ R - 2 \, \Lambda \Bigr] \, \xi^\mu \, \boldsymbol{\epsilon}_\mu \doteq & - \dfrac{1}{16\pi \, G_\mathrm{N}} \, \Bigl[ 2 \, \nabla_\alpha \nabla^{[\mu} \xi^{\alpha]} - 2 \, T^{\textnormal{\AE}\mu} {}_\alpha \, \xi^\alpha + \\
        & + 2 \, X^{\mu\alpha\beta} \, \nabla_\alpha \xi_\beta + 2 \, Y^\mu {}_\alpha \Bigl( \xi^\beta \, \nabla_\beta u^\alpha -u^\beta \, \nabla_\beta \xi^\alpha \Bigr) \Bigr] \, \boldsymbol{\epsilon}_\mu \,,
    \end{split}
\end{equation}
where we also expanded the Lie derivative of the \AE{ther} vector and used the symmetry of $X$ in the last two indices.

The right hand side turns out to be
\begin{equation}
    \begin{split}
        - \dfrac{1}{16\pi \, G_\mathrm{N}} \, \left[ -\L_\uu \, \xi^\mu + 2 \, \nabla_\alpha \nabla^{[\mu} \xi^{\alpha]} \right. & - \textcolor{red}{\underline{\textcolor{black}{Y^{\alpha\mu} \, \xi^\beta \, \nabla_\alpha \uu_\beta}}}  + \textcolor{red}{\underline{\textcolor{black}{ Y^{\mu\alpha} \, \xi^\beta \, \nabla_\alpha \uu_\beta}}} + \textcolor{blue}{\underline{\textcolor{black}{Y^{\alpha\beta} \, \xi_\beta \, \nabla_\alpha u^\mu}}} - \textcolor{Green}{\underline{\textcolor{black}{u^\alpha \, \xi_\beta \, \nabla_\alpha Y^{\beta\mu}}}} + \\
        & - \textcolor{blue}{\underline{\textcolor{black}{u^\alpha \, \xi_\beta \, \nabla_\alpha Y^{\mu\beta}}}} + \textcolor{blue}{\underline{\textcolor{black}{u^\mu Y^{\alpha\beta} \, \nabla_\alpha \xi_\beta}}} - \textcolor{Green}{\underline{\textcolor{black}{u^\alpha \, Y^{\beta\mu} \, \nabla_\alpha \xi_\beta}}} - \textcolor{blue}{\underline{\textcolor{black}{u^\alpha \, Y^{\mu\beta} \, \nabla_\alpha \xi_\beta}}} +\\
        & - \textcolor{Green}{\underline{\textcolor{black}{Y^{\beta\mu} \, \xi_\beta \, \nabla_\alpha u^\alpha}}} - \textcolor{blue}{\underline{\textcolor{black}{Y^{\mu\beta} \, \xi_\beta \, \nabla_\alpha u^\alpha}}} + \textcolor{Green}{\underline{\textcolor{black}{Y^{\beta\alpha} \, \xi_\beta \, \nabla_\alpha u^\mu}}} + \textcolor{Green}{\underline{\textcolor{black}{u^\mu \, \xi_\beta \, \nabla_\alpha Y^{\beta\alpha}}}} + \\
        & + \textcolor{blue}{\underline{\textcolor{black}{u^\mu \, \xi_\beta \, \nabla_\alpha Y^{\alpha\beta}}}} - \textcolor{red}{\underline{\textcolor{black}{u^\beta \, \xi_\beta \, \nabla_\alpha Y^{\alpha\mu}}}} + \textcolor{red}{\underline{\textcolor{black}{u^\beta \, \xi_\beta \, \nabla_\alpha Y^{\mu\alpha}}}} + \textcolor{Green}{\underline{\textcolor{black}{u^\mu \, Y^{\beta\alpha} \, \nabla_\alpha \xi_\beta}}} + \\
        & - \left. \textcolor{red}{\underline{\textcolor{black}{u^\beta \, Y^{\alpha\mu} \, \nabla_\alpha \xi_\beta}}} + \textcolor{red}{\underline{\textcolor{black}{u^\beta \, Y^{\mu\alpha} \, \nabla_\alpha \xi_\beta}}} \right] \,,
    \end{split}
\end{equation}
\hspace{0.5cm} $\Rightarrow - \dfrac{1}{16\pi \, G_\mathrm{N}} \, 2 \, \nabla_\alpha \, \left[ \nabla^{[\mu} \xi^{\alpha]} + \textcolor{blue}{\underline{\textcolor{black}{u^{[\mu} \, Y^{\alpha]\beta} \, \xi_\beta}}} + \textcolor{Green}{\underline{\textcolor{black}{u^{[\mu} \, Y^{\beta|\alpha]} \, \xi_\beta}}} + \textcolor{red}{\underline{\textcolor{black}{Y^{[\mu\alpha]} \, u^\beta \, \xi_\beta}}} \right] \, \boldsymbol{\epsilon}_\mu$.

\vspace{0.2cm}
Finally we get the charge density, because
\begin{equation}
    \begin{split}
        \JJ_\textnormal{\AE}[\vv_\xi] = \iota_{\vv_\xi} \TTh_\textsc{tot} & - \dfrac{1}{16\pi \, G_\mathrm{N}} \, \left[ R - 2 \, \Lambda + \L_\uu \, \right] \, \xi^\mu \, \boldsymbol{\epsilon}_\mu \doteq\\
        & \doteq \nabla_\alpha \, \left( - \dfrac{1}{8\pi \, G_\mathrm{N}} \, \left[ \nabla^{[\mu} \xi^{\alpha]} + u^{[\mu} \, Y^{\alpha]\beta} \, \xi_\beta + u^{[\mu} \, Y^{\beta|\alpha]} \, \xi_\beta + Y^{[\mu\alpha]} \, u^\beta \, \xi_\beta \right] \, \boldsymbol{\epsilon}_\mu \right)
    \end{split}
\end{equation}
\begin{equation*}
    \implies \qq_\textsc{tot}[\vv_\xi] = - \dfrac{1}{8\pi \, G_\mathrm{N}} \, \left[ \nabla^{[\mu} \xi^{\alpha]} + u^{[\mu} \, Y^{\alpha]\beta} \, \xi_\beta + u^{[\mu} \, Y^{\beta|\alpha]} \, \xi_\beta + Y^{[\mu\alpha]} \, u^\beta \, \xi_\beta \right] \, \nn_\alpha \, \boldsymbol{\epsilon}_\mu \,.
\end{equation*}

\section{Disformal Metric Computations}

\subsection{Choice of the Disformal Factor}\label{app: scelta di B}

Given a radial geodesic trajectory for a mode propagating at speed $c_s$ (the $s$-mode, in the main text), we can decompose its momentum $\mathbf{V}$ with respect to the frame $(\uu, \ss)$ as
\begin{equation}\label{eq: V decomposition}
     \mathbf{V}_\mu = - \omega \, \uu_\mu + k \, \ss_\mu \,,
\end{equation}
where $\omega := g^{-1}(\uu, \mathbf{V})$ is the frequency of the mode in the \AE{ther} frame and $k := g^{-1}(\ss, \mathbf{V})$ is the radial wave vector.
In the \AE{ther} frame, $\omega$ and $k$ are related by the \emph{dispersion relation}~\footnote{Notice that the dispersion relation is relativistic, since the theory is second order in the derivatives.}
\begin{equation}\label{eq: dispersion relation}
    \omega^2 = c_s^2 \, k^2 \quad \implies \quad \omega = \pm \, c_s \, k \,,
\end{equation}
where the $+$ is for outgoing trajectories and the $-$ for ingoing ones.

The momentum $\mathbf{V}$ is generically not null in the original metric
\begin{equation}
    g^{-1}(\mathbf{V}, \mathbf{V}) = g^{-1}(- \omega \, \uu + k \, \ss, - \omega \, \uu + k \, \ss) = - \omega^2 + k^2 = (1-c_s^2) \, k^2 \,.
\end{equation}
With a suitable disformal transformation, we can make $\mathbf{V}$ to be null. Using that the inverse disformal metric is given by
\begin{equation}
    \overline{g}^{\mu\nu} = g^{\mu\nu} - \dfrac{B}{1 - B} \, u^\mu \, u^\nu \,,
\end{equation}
we can impose $\overline{g}^{-1}(\mathbf{V}, \mathbf{V}) = 0$ to get
\begin{equation}
    g^{\mu\nu} \, \mathbf{V}_\mu \, \mathbf{V}_\nu - \dfrac{B}{1 - B} \, (u^\mu \, \mathbf{V}_\mu)^2 = - \omega^2 + k^2 - \dfrac{B}{1 - B} \, \omega^2 = \frac{1 - c_s^2 - B}{1 - B} k^2 = 0 \,.
\end{equation}
which implies $B = 1 - c_s^2$.

\subsection{Position of the Causal Horizon for the $s$-mode}\label{app: position of the causal horizon}
\begin{star-prop}
    For solutions of Einstein--\AE{ther} theory which correspond to a stationary, spherically-symmetric spacetime, the position $r_\textsc{h}(c_s)$ of the causal horizon for a mode moving at speed $c_s$ along a radial geodesic is implicitly given by the condition
    \begin{equation}\label{eq: posizione causal horizon s-mode}
        (\iota_\xi \ss) \eval_{r_\textsc{h}(c_s)} - c_s \, (\iota_\xi \uu) \eval_{r_\textsc{h}(c_s)} = 0 \,.
    \end{equation}
\end{star-prop}
\begin{proof}
    The idea to find the causal horizon is to follow an outgoing trajectory of the $s$-mode and look for the position where the modulus of the wave vector diverges.
    
    Thanks to the Killing vector field $\xi$, we can define the \emph{Killing frequency}
    \begin{equation}
        \Omega := - \iota_\xi \mathbf{V} \,,   
    \end{equation}
    which is conserved along the curve, since it is a geodesic. Substituting the decomposition Eq.~\eqref{eq: V decomposition} and using the dispersion relation Eq.~\eqref{eq: dispersion relation} associated with outgoing trajectories, we can write
    \begin{equation}
        \Omega = c_s \, k \, (\iota_\xi \uu) - k \, (\iota_\xi \ss) \,.
    \end{equation}
    Solving for $k$ and imposing that it has a pole in $r_\textsc{h}(c_s)$, we get the thesis.
    \begin{equation}
        k = \dfrac{\Omega}{c_s \, (\iota_\xi \uu) - (\iota_\xi \ss)} \quad \implies \quad (\iota_\xi \ss) \eval_{r_\textsc{h}(c_s)} - c_s \, (\iota_\xi \uu) \eval_{r_\textsc{h}(c_s)} = 0 \,.
    \end{equation}
\end{proof}

\subsection{Ray Tracing}\label{app: ray tracing}
In this appendix we are going to fully trace rays of speed $c_s$. We will denote by $W^\mu$ the vector tangent to the incoming radial trajectories of the $s$-mode and decompose it as $W^\mu = u^\mu + c_s \, s^\mu$.
Assuming stationarity and spherical symmetry, we can write the line element and the tangent vector $W$ in ingoing Eddington--Finkelstein coordinates
\begin{equation}
    ds^2 = - e(r) \, dv^2 + 2 \, dv \, dr + r^2 \, d\Omega^2_2 \,,
\end{equation}
\begin{equation}
    W = W^v \, \partial_v + W^r \, \partial_r \,.
\end{equation}
In these coordinate, the Killing vector field is $\xi = \partial_t = \partial_v$, and a simple computation leads to
\begin{equation}\label{eq: Killing components}
        (\iota_\xi \uu) = - e(r) \, u^v + u^r \,, \quad \mathrm{and} \quad (\iota_\xi \ss) = - e(r) \, s^v + s^r \,.
\end{equation}
Our goal is to compute the peeling behavior of rays close to the horizon. To do so, we first compute $dr/dv$ and expand it around the $s$-mode horizon radius $r_\textsc{h}(c_s)$
\begin{equation}
        \dfrac{dr}{dv}(r) = \dfrac{dr}{dv} \eval_{r_\textsc{h}(c_s)} + \left( \dfrac{d}{dr} \, \dfrac{dr}{dv} \right) \eval_{r_\textsc{h}(c_s)} \, (r - r_\textsc{h}(c_s)) + o((r - r_\textsc{h}(c_s))^2 \,.
\end{equation}
By definition, the \textit{peeling surface gravity} will be the coefficient of the linear term:
\begin{equation}\label{eq: kappa peeeeling}
    \kappa_{peel} := \left( \dfrac{d}{dr} \, \dfrac{dr}{dv} \right) \eval_{r_\textsc{h}(c_s)} \,.
\end{equation}

Given the tangent vector $W^\mu$, we can write
\begin{equation}
    \dfrac{dr}{dv} = \dfrac{W^r}{W^v} = \dfrac{u^r + c_s \, s^r}{u^v + c_s \, s^v} \,.
\end{equation}
so we set out to compute the components of $u$ and $s$. Since these are normalized, we can express their radial component in terms of their (advanced) temporal one. Namely
\begin{subequations}
    \begin{alignat}{2}
        -1 &= g(u,u) = g_{vv} \, u^v \, u^v + 2 \, g_{vr} \, u^v \, u^r = - e(r) \, (u^v)^2 + 2 \, u^v \, u^r \, &&\Rightarrow \, u^r = - \dfrac{1}{2 \, u^v} + \dfrac{e(r)}{2} \, u^v \,,\\
        1 &= g(s,s) = g_{vv} \, s^v \, s^v + 2 \, g_{vr} \, s^v \, s^r = - e(r) \, (s^v)^2 + 2 \, s^v \, s^r \, &&\Rightarrow \, s^r = \dfrac{1}{2 \, s^v} + \dfrac{e(r)}{2} \, s^v \,.
    \end{alignat}
\end{subequations}
Combining these with Eq.~\eqref{eq: Killing components}, we can solve for the advanced time components
\begin{subequations}
    \begin{alignat}{2}
        &(\iota_\xi\uu) = - e(r) \, u^v + u^r = - e(r) \, u^v - \dfrac{1}{2 \, u^v} + \dfrac{e(r)}{2} \, u^v \, &&\Rightarrow \, u^v = \dfrac{- (\iota_\xi \uu) \pm (\iota_\xi\ss)}{e(r)} \,,\\
        &(\iota_\xi\ss) = - e(r) \, s^v + s^r = - e(r) \, s^v + \dfrac{1}{2 \, s^v} + \dfrac{e(r)}{2} \, s^v \, &&\Rightarrow \, s^v = \dfrac{-(\iota_\xi\ss) \pm (\iota_\xi \uu)}{e(r)} \,,
    \end{alignat}
\end{subequations}
where we also used the relation $e(r) = (\iota_\xi \uu)^2 - (\iota_\xi \ss)^2$ to simplify the expressions. To choose the sign in the expressions for $u^v$ and $s^v$, we use that $g(u,s) = 0$, which fixes them to be the same. We can then fix the overall sign by requiring that, as $c_s \to 1$, we get back the usual definition of surface gravity for the Killing horizon $e'(r)/2$. So we obtain
\begin{equation}
    \dfrac{dr}{dv} = \dfrac{W^r}{W^v} = \dfrac{\pm \left[ (\iota_\xi \ss) + c_s \, (\iota_\xi \uu) \right] \, e(r)}{(\pm 1 - c_s) \, \left[ (\iota_\xi \uu) + (\iota_\xi \ss) \right]} =: G_{s}(r) \,.
\end{equation}
As $c_s \to 1$, the expression diverges for the $+$ sign and approaches $e(r)/2$ for the $-$ sign, so we choose the latter, leading to the final form of $G_s(r)$
\begin{equation}\label{eq: G peeling}
    G_{s}(r) = \dfrac{\left[ (\iota_\xi \ss) + c_s \, (\iota_\xi \uu) \right] \, e(r)}{(1 + c_s) \, \left[ (\iota_\xi \uu) + (\iota_\xi \ss) \right]} \,.
\end{equation}
At the causal horizon $G_s(r_\textsc{h}(c_s)) = 0$, compatibly with Eq.~\eqref{eq: posizione causal horizon s-mode}, while its first derivative there will give us the peeling surface gravity Eq.~\eqref{eq: kappa peeeeling}
\begin{equation}
    \kappa_{peel} = G_s'(r_\textsc{h}(c_s)) \,. 
\end{equation}

In the limit $c_s \to +\infty$, assuming all the other quantities are well-behaved,~\footnote{This is not the case if the $s$-mode is taken to be one of the gravitational propagating modes, as we mentioned in Section~\ref{sec: comparison}.} we can write
\begin{equation}
    G_\infty(r) = \frac{(\iota_\xi \uu) \, e(r)}{(\iota_\xi \uu) +(\iota_\xi \ss)} \,.
\end{equation}
This vanishes for $(\iota_\xi \uu) = 0$, which defines the position of the universal horizon (see Eq.~\eqref{eq:UHcond}) and gives the peeling surface gravity
\begin{equation}\label{eq: kappa UH da G}
    G'_\infty(r_\textsc{h}(+\infty)) =  - (\iota_\xi \ss) \, (\iota_\xi \uu)' \,,
\end{equation}
which can be seen to coincide with the surface gravity $\kappa_\textsc{uh}$ defined in Eq.~\eqref{eq: kappa UH} and found in~\cite{DelPorro:2023lbv}.

\paragraph{In the disformal frame}
If we apply the disformal transformation Eq.~\eqref{eq: disformal gbar to g}, the tangent vector $W^\mu$ becomes null
\begin{equation}
    \overline{g}(W, W) = g(W, W) + (1 - c_s^2) (\iota_W \uu)^2 = - 1 + c_s^2 + (1 - c_s^2) = 0 \,,
\end{equation}
which is equivalent to set $\overline{c}_s = 1$ by construction. In this frame, we thus have the simple relation $G_s(r) = \overline{e}(r)$. The causal horizon $r_\textsc{h}(c_s)$ will coincide with the Killing horizon of the metric and the peeling surface gravity $\kappa_{peel}$ will be given by one of the (equivalent) definitions of Killing surface gravity.

\section{Background Subtraction Procedure}\label{app: background subtraction}
In this appendix we review the regularization procedure for the asymptotic quantities called \emph{background subtraction}. As already emphasized, this regularization is needed to work with quantities that are well defined in the limit where the boundary is pushed to infinite distance.

As customary in General Relativity, we model a family of \emph{stationary observers} with a congruence of curves that are everywhere tangent to the vector field
\begin{equation}
    U := \dfrac{1}{||\partial_t||} \, \partial_t \,,
\end{equation}
which is normalized so that $||U||^2 = -1$. In an adapted system of coordinates $(\tau, \eta^1, \eta^2, \eta^3)$, where $\tau$ is the proper time of the observer, the time component of the metric gets transformed (in natural units)
\begin{equation}\label{eq: redshift factor}
    -e(r) \, dt^2 = - d\tau^2 \quad \implies \quad e(r) = \left( \dfrac{d\tau}{dt} \right)^2 =: \nu_r^2 \quad \implies \quad \partial_t = \nu_r \, \partial_\tau \,,
\end{equation}
where $\nu_r$ is the \emph{redshift factor} at radius $r$.

The energy entering the first law is the one associated with the vector $\partial_t$
\begin{equation}
    E_\R := \int_{\S_\R} \qq_\textsc{tot}[\vv_t] - \jj[\vv_t] \,,
\end{equation}
but the one \emph{measured} by the stationary observer is the one associated with $\partial_\tau$
\begin{equation}
    \E_\R := \int_{\S_\R} \qq_\textsc{tot}[\vv_\tau] - \jj[\vv_\tau] \,.
\end{equation}
From the relation Eq.~\eqref{eq: redshift factor} and using linearity, we infer the relation
\begin{equation}\label{eq: relation}
    E_\R = \nu_\R \, \E_\R \,
\end{equation}
between these two quantities.

We are now ready to do the actual background subtraction. The divergence in asymptotic quantities usually comes from some background contribution, hence, as the name suggests, the key idea of this procedure is to subtract the energy measured by the stationary observers in a suitable background to hopefully obtain a finite contribution.
The regularized measured energy is then
\begin{equation}
    \E_\R^{(reg)} := \E_\R - \E_\R^{(bg)} \,,
\end{equation}
where $\E_\R^{(bg)}$ is the energy computed in the background solution.
Keeping in mind the previous relation~\eqref{eq: relation} between the proper energy and the caonical one, we can deduce that
\begin{equation}
    E_\R^{(reg)} = \nu_\R \, \E_\R^{(reg)} = \nu_\R \, \left( \E_\R - \E_\R^{(bg)} \right) = \nu_\R \, \left( \dfrac{1}{\nu_\R} \, E_\R - \dfrac{1}{\nu^{(bg)}_\R} \, E_\R^{(bg)} \right) = E_\R - \dfrac{\nu_\R}{\nu_\R^{(bg)}} \, E_\R^{(bg)} \,.
\end{equation}
Finally, we define the regularized mass entering the first law as
\begin{equation}
    M = \lim_{\R \to +\infty} \left( E_\R - \dfrac{\nu_\R}{\nu_\R^{(bg)}} \, E_\R^{(bg)}\right) \,.
\end{equation}

\addcontentsline{toc}{section}{Bibliography}
%========================================================
\bibliographystyle{alphaurl}
%========================================================
\bibliography{CompleteBibliography}

\nocite{*}
\end{document}